%% file: 0-OpinionVote.tex
\documentclass[conference]{IEEEtran}
\IEEEoverridecommandlockouts

\input{macros}

\begin{document}
\bstctlcite{IEEEexample:BSTcontrol}
\title{Voting-based Opinion Maximization
\thanks{Xiangyu Ke is the corresponding author. Arijit Khan acknowledges support from the Novo Nordisk Foundation grant NNF22OC0072415. Laks V.S. Lakshmanan's research was supported by a grant from the Natural Sciences and Engineering Research Council of Canada (NSERC).}}

\newcommand{\revise}[1]{#1}

\author{\IEEEauthorblockN{Arkaprava Saha}
\IEEEauthorblockA{
\textit{NTU, Singapore}\\
saha0003@e.ntu.edu.sg}
\and
\IEEEauthorblockN{Xiangyu Ke}
\IEEEauthorblockA{
\textit{ZJU, China}\\
xiangyu.ke@zju.edu.cn}
\and
\IEEEauthorblockN{Arijit Khan}
\IEEEauthorblockA{
\textit{AAU, Denmark}\\
arijitk@cs.aau.dk}
\and
\IEEEauthorblockN{Laks V.S. Lakshmanan}
\IEEEauthorblockA{
\textit{UBC, Canada}\\
laks@cs.ubc.ca}
}

\maketitle

\begin{abstract}
We investigate the novel problem of \textit{voting-based opinion maximization} in a social network:
\textit{Find a given number of seed nodes for a target campaigner, in the presence
of other competing campaigns, so as to maximize a voting-based score for the target
campaigner at a given time horizon.} 

The bulk of the influence maximization literature assumes that social network users can switch between only two discrete states, inactive and
active, and the choice to switch is frozen upon one-time activation.
In reality, even when having a preferred opinion, a user may not
completely despise the other opinions, and the preference level may {\em vary
over time} due to social influence. To this end, we employ
models rooted in {\em opinion formation and diffusion}, and use several {\em voting-based scores} to
determine a user's vote for each of the multiple campaigners {\em at a given time horizon}.

Our problem is \NP-hard and non-submodular for various scores.
We design greedy seed selection algorithms with quality guarantees for our
scoring functions via sandwich approximation. To improve the {\em efficiency}, we develop random walk and
sketch-based opinion computation, with {\em quality guarantees}.
Empirical results validate our effectiveness, efficiency, and scalability. 
\end{abstract}


\input{1-intro}
\input{2-preliminaries}
\input{3-algo}

\input{4-random_walk}

\input{5-sketching}
\input{6-related}
\input{7-exp}
\input{8-conclusions}
\input{appendix}


\balance
\bibliographystyle{IEEEtran}
\bibliography{ref}

\end{document}

%% file: macros.tex
\usepackage{hyperref} 
\usepackage{graphicx}
\usepackage{booktabs}
\usepackage{amsmath}
\usepackage{amsthm}
\usepackage{bbm}
\usepackage{amsfonts}
\usepackage{epstopdf}
\usepackage{graphicx}
\usepackage{subfigure}
\usepackage{url}
\usepackage{balance}
\usepackage[noend]{algorithmic}
\usepackage{algorithm}
\usepackage{multirow}
\usepackage{xspace}

\usepackage{pifont}
\newtheorem{theor}{Theorem}

\newtheorem{lem}{Lemma}

\newtheorem{defn}{Definition}

\newtheorem{exam}{Example}
\newtheorem{problem}{Problem}

\newcommand{\spara}[1]{\smallskip\noindent{\bf #1}}

\newcommand{\squishlist}{
 \begin{list}{$\bullet$}
  {  \setlength{\itemsep}{0pt}
     \setlength{\parsep}{3pt}
     \setlength{\topsep}{3pt}
     \setlength{\partopsep}{0pt}
     \setlength{\leftmargin}{2em}
     \setlength{\labelwidth}{1.5em}
     \setlength{\labelsep}{0.5em}
} }
\newcommand{\squishlisttight}{
 \begin{list}{$\bullet$}
  { \setlength{\itemsep}{0pt}
    \setlength{\parsep}{0pt}
    \setlength{\topsep}{0pt}
    \setlength{\partopsep}{0pt}
    \setlength{\leftmargin}{2em}
    \setlength{\labelwidth}{1.5em}
    \setlength{\labelsep}{0.5em}
} }

\newcommand{\squishdesc}{
 \begin{list}{}
  {  \setlength{\itemsep}{0pt}
     \setlength{\parsep}{3pt}
     \setlength{\topsep}{3pt}
     \setlength{\partopsep}{0pt}
     \setlength{\leftmargin}{1em}
     \setlength{\labelwidth}{1.5em}
     \setlength{\labelsep}{0.5em}
} }

\newcommand{\squishend}{
  \end{list}
}

\newcommand{\eat}[1]{}

\newcommand{\NP}{\ensuremath{\mathbf{NP}}\xspace}

\newcounter{ccc}

\DeclareMathOperator*{\argmax}{arg\,max}
\newcommand{\bigO}{\mathcal{O}}

\usepackage[show]{boxnotes}
\usepackage[normalem]{ulem}
\newcommand\redout{\bgroup\markoverwith
{\textcolor{red}{\rule[.5ex]{2pt}{2pt}}}\ULon}

%% file: 1-intro.tex
\vspace{-2mm}
\section{Introduction}
\label{sec:intro}
%
\vspace{-1mm}
Social influence studies have
attracted extensive attention in the data management research community \cite{2018Khan2,GalhotraAR16,TSX15,AroraGR17,TangHXLTSL19,GBL11, TXS14, TTXY18}.
The classic influence maximization (IM) problem \cite{KKT03,DR01} identifies the top-$k$ seed users in a social network to maximize the expected number of influenced users in the network, starting from those seed nodes and following an influence diffusion
model (e.g., independent cascade (IC) and linear threshold (LT) \cite{KKT03}).
Several works also
focus on competitive influence maximization \cite{BKS07, CNWZ07, HSCJ12, LL15, KLRS21,BAA11,LuB0L13,KhanZK16} which aims to find
the seed set that maximizes the influence spread for a particular campaigner relative to the others or maximally blocks the diffusion of a competitor.

However, prior works on IM have two major limitations in modelling real-world opinion
formation and spreading. First, they consider maximizing the expected number of users adopting a specific campaign,
assuming that the reaction of each user to the campaign is {\em binary} (adopt or not). 
In reality, a user may not be completely opposed to the competing opinions, although
she could have a preference for one opinion,
where the degree of preference could vary among users.
This scenario can be accurately modelled by allowing the opinion of a user for each campaign to be a real number in $[0,1]$.
Second, 
in the IC and LT models, a
user's choice is frozen upon one-time activation -- not permitting to switch opinions later.
While this is realistic when purchasing one of the many competing products due to the user's limited budget,
it is insufficient for modelling {\it opinion formation and manipulation over time}, e.g., in scenarios like paid movie
services, elections, social issues, where a user's opinion is highly likely to change over time.

Due to the above shortcomings, we deviate from the classic influence diffusion (e.g., IC and LT models)
and investigate the problem of {\em opinion maximization} by employing models rooted in
opinion formation and diffusion, e.g., DeGroot \cite{DeG74} and Friedkin-Johnsen (FJ) \cite{FJ90,FJ99}.
In these settings, each user in a network has a \textit{real-valued} opinion about each campaign at
\textit{every} timestamp. Moreover, for each campaign,
the \textit{opinions of the users evolve} over discrete timestamps according
to an opinion diffusion model such as DeGroot or FJ (defined in \S~\ref{sec:diff_model}).
Given a target campaign and a time horizon (a future timestamp $t$), our problem is to select a seed set of size $k$
for the target campaigner, so that the target campaigner's odds
of being the {\em winner} at the time horizon $t$ are as high as possible.

Since opinion values are
non-binary, we require more sophisticated winning criteria than the \textit{expected influence spread} employed in  classic IM \cite{KKT03}.
Voting offers a well-understood  mechanism for determining winners in an election
among campaigners by considering the preferences of
users (``voters'') in a principled manner. We investigate voting-based scores
\cite{VM19,Gae06,Fish74} such as aggregated opinion values of all
users about a campaigner (\textit{cumulative}), rank of the target campaigner relative to others for
all users (\textit{\revise{plurality}}), or
the number of campaigners against whom the target campaigner wins in one-on-one
competitions (\textit{\revise{Copeland}}). These are natural choices based on voting theory when users
have non-binary opinion values towards multiple competitors.
Existing works on finding the top-$k$ seeds for opinion maximization \cite{GionisTT13, AKPT18}
\textit{are restricted to a single campaigner and
consider neither a given finite time horizon\footnote{\scriptsize In practice, the voting is held at a specific time horizon, instead of waiting for the diffusion to reach the Nash equilibrium as is done in \cite{GionisTT13}.}, nor voting-based scores with multiple competing campaigners\footnote{\scriptsize Only our cumulative score is similar to theirs due to its aggregate nature.}. To the best of our knowledge, voting-based opinion maximization
in the presence of multiple competing campaigns is a novel problem}.

\vspace{-0.3mm}
\spara{\revise{Applications.}}
\revise{
Our problem and solutions can be effective where users vote and the winner among multiple candidates is decided based on the election outcome. 
Examples include the presidential election, voting in the parliament, a plebiscite or a referendum  
(e.g., the referendum on the independence of Scotland) \cite{bruno94,Emerson2020}, etc.
We conduct a real-world case study about the ACM general election 2022 (\S \ref{sec:case}). 
Our case study shows that the election result might have reversed after introducing only 100 optimal seed
users. Our solution selects influential seeds based on (1) their common research interests with respect to the target candidate 
and (2) the initial preferences of the users in various research domains.
Moreover, our approach smartly focuses on switching the preferences of more neutral users. 
These demonstrate the usefulness of our problem and the effectiveness of our solution.}

\vspace{-0.3mm}
\spara{Challenges and Our Contributions.}
With multiple competing campaigns in a network, we formulate and study a \textit{novel problem in opinion maximization}:
Find the top-$k$ seed nodes for a target campaign that maximize a voting-based winning criterion for the target
at a given time horizon (\S~\ref{sec:prob}). Our contributions are as follows.

\noindent$\bullet$ {\bf Opinion Maximization and Voting Scores:}
To the best of our knowledge, opinion manipulation by introducing seed nodes has not been investigated before,
except, e.g., \cite{GalhotraAR16, GionisTT13, AKPT18,Li0WZ13,0005BSC15}. However, apart from \cite{GionisTT13, AKPT18}, prior works do not consider sophisticated
DeGroot/FJ opinion models. Also, opinion maximization at a finite time horizon with multiple campaigners
has not been explored even in \cite{GionisTT13, AKPT18}.
One of our novel contributions is bridging two different paradigms: {\bf (1)} seed selection for opinion formation and diffusion
till {\em a given finite time horizon}, and {\bf (2)} voting-based winning criteria (e.g., \revise{plurality}, 
\revise{Copeland}) with {\em multiple campaigners}.

\noindent$\bullet$ {\bf Sandwich Approximation:}
Our problem is \NP-hard (\S~\ref{sec:hard})
and non-submodular (\S~\ref{sec:subm}) under various winning criteria\footnote{\scriptsize The proofs of these results in \cite{GionisTT13} cannot
be extended trivially even to our basic model of the cumulative score for any finite time
horizon, warranting new techniques.}.
Despite these, we design bound functions for
all our non-submodular scores to derive accuracy guarantees for the
greedy algorithm via {\em sandwich approximation} \cite{LCL15} (\S~\ref{sec:sandwich}).

\noindent$\bullet$ {\bf Random Walks:}
Computing opinion values at the time horizon via DeGroot/ FJ requires
iterative matrix-vector multiplications, which is expensive.
{\em To improve the efficiency, we next propose random walk and sketching-based computations with approximation guarantees.}
Random walks have been used earlier to improve the efficiency of matrix multiplication
and PageRank computation \cite{ALNO07,CohenL99}. Our novelty is using random walks to find the $k$ seed nodes maximizing
a voting-based score by approximating the opinion values via the walks in $k$ iterations.
Also, we provide novel bounds on the number of walks required for each voting-based scoring function (\S~\ref{sec:random_walk}).

\noindent$\bullet$ {\bf Sketches:}
While sketches have been used
in  classic IM \cite{TSX15,TXS14,BBCL14},
\textit{ours is the first work that
uses sketches for opinion computation.} 
We adapt sketches for opinion diffusion models and voting-based scores, and derive non-trivial accuracy guarantees (\S~\ref{sec:sketching}).
Moreover, our sketches are simpler and less memory-consuming than RR-sets-based sketches \cite{TSX15,TXS14}.

Our thorough experimental evaluation and case study 
over five real-world social network datasets demonstrates the effectiveness, efficiency, and scalability of our solutions, 
over several baselines (\S~\ref{sec:exp}). Related work is discussed in \S~\ref{sec:related}, while in \S~\ref{sec:concl} we conclude and discuss future work.

%% file: 2-preliminaries.tex
\vspace{-2mm}
\section{Preliminaries}
\label{sec:preliminaries}
\vspace{-1mm}
A social network is modeled as a (directed) graph $\mathcal{G}=(V,E)$, where $V$ is the set of $n$ nodes
and $E\subseteq V\times V$ is the set of $m$ edges.
Each node is a user, and an edge
represents social relation between two users. We denote  matrices with upper-case letters and use lower-case
ones for their entries. We denote an $n\times n$ diagonal matrix by $diag(d_1, d_2, ..., d_n)$,
and the $n\times n$ identity matrix by $I_n$. A matrix $A=(a_{ij})$ is {\em column-stochastic}
if $a_{ij}\geq 0, \; \forall i,j$,  and $\sum_{i=1}^{n}a_{ij}=1, \; \forall j$.

Different news, campaigns, or opinions can propagate concurrently in the network,
leading to competitions \cite{BKS07,CNWZ07,BAA11}. They can be information about similar products of different brands,
multiple politicians campaigning for the same position, or different attitudes towards a
topic, e.g., for or against gun control. We call them 
\textit{candidates} and assume that there are $r>1$ candidates:  $C=\{c_1, c_2, ..., c_r\}$. All users' opinions (in the interval $[0,1]$) on all candidates are represented by an opinion matrix $B\in [0,1]^{r \times n}$.
$B_q\in [0,1]^{1\times n}$ is the $q^{th}$ row of $B$ (denoting all users' opinions on candidate $c_q$),
and $b_{qi}$ is its $i^{th}$ entry (opinion of user $i$ on  candidate $c_q$). The opinions evolve over
discrete timestamps $\{0,1,...,t\}$. 
We denote the opinion(s) at timestamp $t$
by, e.g., $B_q^{(t)}$ and $b^{(t)}_{qi}$.

\vspace{-2mm}
\subsection{Opinion Diffusion Models}
\label{sec:diff_model}
Unlike the classic influence diffusion, opinion diffusion involves aggregating the peers' opinions at each timestamp
\cite{Noor20}.
We introduce a column-stochastic influence matrix \cite{DeG74,HM20} $W\in [0,1]^{n\times n}$,
where $w_{ij}\in [0,1]$ denotes the influence weight from user $i$ to user
$j$.
Different candidates $c_q$ can have different matrices $W_q$.
Notice that {\em barring these weights, the graph structure and the nodes remain the same for all candidates}.
The set $E$ is the union of the edges with non-zero weights across all candidates. This setting is used in topic-aware IM \cite{CFLFTT15}.
We next present two widely used opinion diffusion models: DeGroot \cite{DeG74} and its extension FJ \cite{FJ90,FJ99}.

\vspace{-0.1mm}
\spara{The DeGroot Model} for a single candidate $c_q$ is given by:
\begin{small}
\begin{align}
  B_q^{(t)}=B_q^{(t-1)}W_q=B_q^{(t-2)}W_q^2=...=B_q^{(0)}W_q^t
  \vspace{-2mm}
  \label{eq:degroot}
\end{align}
\end{small}
At every timestamp, each user adopts the
weighted average of her in-neighbors' opinions from the previous timestamp.
Users without in-neighbors retain their initial opinions.
Since $W_q$ is column-stochastic, the opinion values remain in $[0,1]$.
We assume that the opinions about different candidates
diffuse independently. In multi-campaigner and multi-feature settings, independent propagation of opinions and influences
has been considered in \cite{
TuAG20,GuoCW21,GarimellaGPT17,KKC18}. Note that in our case, {\em while the opinion propagation for multiple
campaigns happens concurrently and independently, voting-based scores naturally incorporate competition among the campaigns
(\S~\ref{sec:voting})}.




\vspace{-0.1mm}
\spara{The Friedkin-Johnsen (FJ) Model} extends the DeGroot model by introducing the notion of {\em stubbornness}:
\begin{small}
\begin{align}
  B_q^{(t+1)} = B_q^{(t)} W_q \left( I - D_q \right) + B_q^{(0)} D_q
  \vspace{-2mm}
  \label{eq:fj}
\end{align}
\end{small}
$D_q=diag(d_{q1}, d_{q2}, ..., d_{qn})$ is a diagonal matrix: $d_{qi}$ represents the stubbornness
of user $i$ on retaining her initial opinion about candidate $c_q$. If $d_{qi}=1$, the user $i$ is {\em fully stubborn} and
sticks to her initial opinion about $c_q$. A {\em partially stubborn} user ($0<d_{qi}<1$)
aggregates the opinions from neighbors as well as her original opinion, while {\em non-stubborn} users
($d_{qi}=0$) follow the DeGroot model. Since the DeGroot model is a special case where all users are non-stubborn,
all our results with the FJ model also hold for the DeGroot model.

If the opinions of all users do not change after a specific timestamp, the diffusion reaches a state of {\em convergence}. The FJ model can reach convergence if and only if the edge weight matrix of the subgraph induced by all oblivious nodes is regular or there is no oblivious node \cite{PT17, PPTF17}. Oblivious nodes are (1) non-stubborn and (2) not reachable from any fully or partially stubborn node. One of our novel contributions is the seed selection for opinion maximization {\em at any given time horizon}, which introduces non-trivial additional hardness, as discussed in \S~\ref{sec:hard} and \S~\ref{sec:subm}.

\begin{figure}
  \centering
  \vspace{-3mm}
  \includegraphics[scale=0.28]{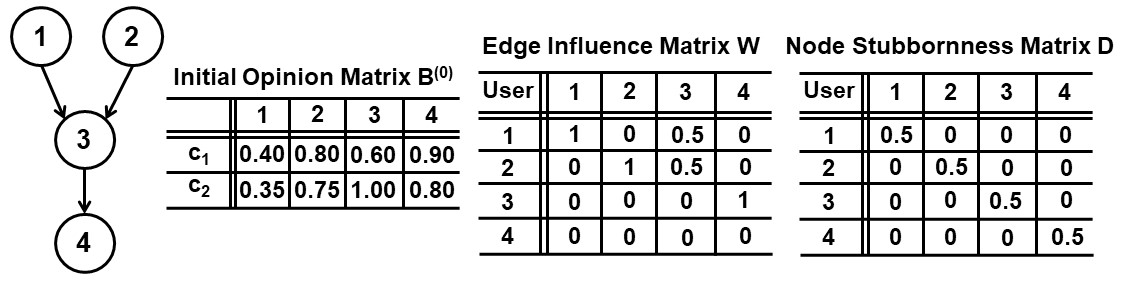}
  \vspace{-4mm}
  \caption{\small \revise{Running example. All users share the same influence weight and stubbornness matrices for both candidates.}}
  \label{fig:example_pre}
    \vspace{-4mm}
\end{figure}

\revise{
\begin{exam}
    \vspace{-2mm}
    The input graph in Figure~\ref{fig:example_pre} consists of 4 users and 3 edges. Suppose $c_1$ is our target candidate and $c_2$ is a competing candidate.
    Based on the FJ model, for any $x \in \{1, 2\}$, a user's opinion about candidate $c_x$ at any time horizon can be computed by taking the weighted average of her in-neighbors' opinions at the previous time horizon and then averaging with that of herself. Thus, users 1 and 2 will always keep their initial opinions, as they do not have any incoming edge. The opinion of user 3 at any time horizon $t$ can be computed as $b^{(t)}_{x3} = \frac{1}{2} \left[ b^{(t-1)}_{x3} + \frac{1}{2} \left( b^{(t-1)}_{x1} + b^{(t-1)}_{x2} \right) \right]$, which is the average opinion of users 1 and 2 at the previous time horizon, then averaged with that of user 3. For user 4, $b^{(t)}_{x4} = \frac{1}{2} \left[ b^{(t-1)}_{x3} + b^{(t-1)}_{x4} \right]$, which is the average of the opinions of users 3 and 4 at the previous time horizon.
    \vspace{-2mm}
    \label{ex:model}
\end{exam}}


%
\vspace{-1mm}
\subsection{Voting-based Scores}
\label{sec:voting}
\vspace{-0.8mm}
All campaigns start at timestamp $0$ and proceed concurrently (FJ model), independently of each other.
Given a time horizon $t$, we employ several voting-based scores \cite{VM19,Gae06,Fish74} to decide
the winning candidate. In particular, we compute a score $F ( B^{(t)}, c )$ for each candidate $c$.
The one with the maximum score is the winner at time $t$. We next define five major voting-based score functions that we study.

\spara{Cumulative Score.}  For a candidate $c_q$, the \textit{cumulative score} is the sum of all users' opinion values about her
at time $t$:
\vspace{-1mm}
\begin{small}
\begin{align}
  F\left(B^{(t)}, c_q\right)=\sum_{v\in V}b_{qv}^{(t)}
  \label{eq:cumu}
  \vspace{-4mm}
\end{align}
\end{small}

\vspace{-1.5mm}
\spara{\revise{Plurality} Score.} 
The \textit{\revise{plurality} score} counts the number of users who prefer $c_q$ to all other candidates
at time $t$:
\vspace{-1mm}
\begin{footnotesize}
\begin{align}
  F\left(B^{(t)}, c_q\right)=\sum_{v \in V}\mathbbm{1}\left[b_{qv}^{(t)} > \max_{c_x \in C \setminus \{c_q\}} b_{xv}^{(t)}\right] = \sum_{v \in V} \mathbbm{1} \left[ \beta \left( b_{qv}^{(t)} \right) \leq 1 \right]
  \label{eq:best-rank}
  \vspace{-4mm}
\end{align}
\end{footnotesize}
$\mathbbm{1}\left[\cdot\right]$ is an indicator that returns 1 if the condition inside is true, 0 otherwise; and $\beta \left( b_{qv}^{(t)} \right) = \sum_{c_x \in C} \mathbbm{1} \left[ b_{xv}^{(t)} \geq b_{qv}^{(t)} \right]$ is the rank of $c_q$ in the preference order for user $v$ at time $t$.
In practice, a user generally votes for only one politician, or has a limited budget to purchase one specific type of product. Intuitively, she selects the one with the highest opinion value in her mind -- the \revise{plurality} score captures this.

\vspace{-0.5mm}
\spara{$p$-Approval Score.}
Given an integer $p \in [1, r]$, the $p$-approval score of $c_q$ is defined as the number of users $v$ such that $c_q$ is among the top-$p$ preferred candidates for $v$ at time $t$, i.e.
\begin{small}
\begin{align}
  F\left(B^{(t)}, c_q\right)=\sum_{v \in V}\mathbbm{1} \left[ \beta \left( b_{qv}^{(t)} \right) \leq p \right]
  \label{eq:p-rank}
  \vspace{-4mm}
\end{align}
\end{small}
%

\vspace{-2mm}
\spara{Positional-$p$-Approval Score.}
Given an integer $p \in [1, r]$ and a sequence of position weights $(\omega[1], \omega[2], ..., \omega[r])$ such that $\omega[i] \in [0, 1] \, \forall i \in [1, r]$ and $\omega[i] \leq \omega[i - 1] \, \forall i \in [2, r]$, the positional-$p$-approval score of $c_q$ is the sum of the weights of the positions (up to $p$) of $c_q$ in the preference order of all users at time $t$. Formally,
\begin{small}
\begin{align}
  F\left(B^{(t)}, c_q\right) =
  \sum_{v \in V} \omega \left[ \beta \left( b_{qv}^{(t)} \right) \right] \times \mathbbm{1} \left[ \beta \left( b_{qv}^{(t)} \right) \leq p \right]
  \label{eq:weighted_p_rank}
  \vspace{-4mm}
\end{align}
\end{small}
Clearly, the $p$-approval and positional-$p$-approval scores are generalizations of the plurality score. In real-world applications like paid movie services, users can hold memberships of multiple platforms; the $p$-approval score accounts for this. Moreover, the service platforms usually provide multiple levels of membership having different prices and benefits. Thus, the platform still prefers a higher rank for itself by each user, since the user
may only purchase higher level memberships for her favorite ones. The positional-$p$-approval score captures this notion.
Both variants allow for ties and are more robust to small noises in the users' individual preference orders.

\spara{\revise{Copeland} Score.} We define an ordering $\succ_{M}$ on candidates: $c_q\succ_{M}c_p$ (i.e., $c_q$ wins over $c_p$), if more users have a higher opinion value for $c_q$ than for $c_p$, compared to the other way around, at time $t$. The score counts how many such one-on-one competitions a candidate $c_q$ wins:
\vspace{-2mm}
\begin{small}
\begin{align}
  &F\left(B^{(t)}, c_q\right) = \left|\{c_p \colon c_q\succ_M c_p\}\right| \nonumber \\
  &= \sum_{c_x \in C \setminus \{c_q\}} \mathbbm{1} \left[ \sum_{v \in V} \mathbbm{1} \left[ b^{(t)}_{qv} > b^{(t)}_{xv} \right] > \sum_{v \in V} \mathbbm{1} \left[ b^{(t)}_{qv} < b^{(t)}_{xv} \right] \right]
  \label{eq:condorcet}
  \vspace{-3mm}
\end{align}
\end{small}
The Condorcet winner \cite{Fish77} is the candidate that wins all such one-on-one competitions, i.e., has the maximum possible $F(B^{(t)},c_q)$ score, which is  $r-1$. In general, a Condorcet winner is not always guaranteed to exist \cite{Fish77}. However, maximizing the \textit{\revise{Copeland} score} boosts the target candidate to beat as many other candidates as possible, and to be as close to become a Condorcet winner as possible.
\vspace{-1.5mm}
\subsection{Problem Formulation}
\label{sec:prob}
\vspace{-1mm}
We study the novel problem of selecting $k$ seed nodes
for a target candidate that maximize one of the
voting-based scores discussed in \S~\ref{sec:voting} for the target candidate at a given time horizon. All our scoring functions are non-decreasing w.r.t.
seed sets (\S~\ref{sec:subm}). Maximizing the score boosts the target candidate's odds
of being as close as possible to winning.

For each node $s$ in the seed set $S$ for candidate $c_q$, we
increase $b_{qs}^{(0)}$ and $d_{qs}$ to 1 (i.e., node $s$ becomes fully stubborn
towards retaining the maximum opinion value about $c_q$).
We denote the modified initial opinion row vector $B_q$ and the stubbornness matrix $D_q$ as
$B_q[S]$ and $D_q[S]$, respectively. The problem is formulated as follows.
%
%
\vspace{-1mm}
\begin{problem} [\textsf{FJ-Vote}]
Given the initial opinion matrix $B^{(0)}$, a target candidate $c_q$,
influence matrix $W_q$, stubbornness matrix $D_q$, and a time horizon $t$,
find a set of $k$ seed nodes $S\subset V$ that maximizes the score for $c_q$ at timestamp $t$.
Formally,
\begin{small}
\vspace{-1mm}
\begin{align}
  S^*=\argmax_{S\subset V, |S|=k} F \left( B^{(t)}[S], c_q \right)
  \vspace{-5mm}
\end{align}
\end{small}
\vspace{-6mm}
\label{prob:fjvote}
\end{problem}
Here $B^{(t)}[S]$ is computed from $B^{(0)}[S]$ via the FJ model (Equation \ref{eq:fj}).
Note that $B^{(0)}[S]$ is obtained from the initial opinion matrix $B^{(0)}$
by updating its row vector $B_q$ to $B_q[S]$ according to the seed set $S$ for $c_q$. The function $F$ is based on one of the five voting scores 
 (\S~\ref{sec:voting}).
\begin{table}[t!]
	\scriptsize
	\vspace{-2mm}
	\centering\revise{
	\begin{center}
		\caption{\small Scores of candidate $c_1$ for various seed sets at $t=1$ in Figure \ref{fig:example_pre}. Assuming no seeds for $c_2$, the opinions of users 1, 2, 3, 4 about $c_2$ at $t = 1$ are resp. 0.35, 0.75, 0.78, 0.90.}
		\vspace{-2mm}
		\begin{tabular}{|c|c|c|c|c|c|c|c|}
		    \hline
			\multirow{2}{*}{\textbf{Seed Set}} &\multicolumn{4}{|c|}{\textbf{User}}  &\multicolumn{3}{|c|}{\textbf{Score}}  \\  \cline{2-8}
			& \textbf{1} & \textbf{2} & \textbf{3} & \textbf{4}  & \textbf{Cumu.} & \textbf{Plu.} & \textbf{Cope.}\\ \hline 
			$\{\}$ & {\bf 0.40} & {\bf 0.80} & 0.60 & 0.75 & 2.55 & 2 & 0 \\ \hline
			$\{1\}$ & {\bf 1.00} & {\bf 0.80} & 0.75 & 0.75 & {\bf 3.30} & 2 & 0 \\ \hline
			$\{2\}$ & {\bf 0.40} & {\bf 1.00} & 0.65 & 0.75 & 2.80 & 2 & 0\\ \hline
			$\{3\}$ & {\bf 0.40} & {\bf 0.80} & {\bf 1.00} & {\bf 0.95} & 3.15 & {\bf 4} & {\bf 1}\\ \hline
			$\{4\}$ & {\bf 0.40} & {\bf 0.80} & 0.60 & {\bf 1.00} & 2.80 & 3 & {\bf 1}\\ \hline
			$\{1,2\}$ & {\bf 1.00} & {\bf 1.00} & {\bf 0.80} & 0.75 & {\bf 3.55}  &  3 & {\bf 1}\\  \hline 
		\end{tabular}
		\vspace{-6mm}
		\label{tab:example}
	\end{center}}
\end{table}
\revise{
\begin{exam}
    \vspace{-2mm}
    Suppose we aim to choose one seed user to maximize the score for $c_1$ (i.e., improve $c_1$'s odds of winning against competitor $c_2$) at time horizon $t=1$. The optimal seed sets are quite different for various voting-based scores. As shown in Table~\ref{tab:example}, selecting user 1 as the seed leads to the maximum cumulative score; however, we still have only 2 users preferring our target candidate $c_1$ to $c_2$. Thus, the Copeland score of $c_1$ remains 0. Choosing user 3 as the seed will encourage all four users to favor $c_1$ over $c_2$, which results in the highest plurality score. Meanwhile, $c_1$ will become the Condorcet winner (Copeland score equals 1) when user 3 or 4 is selected as the seed, since more than half the users will have higher opinion values for $c_1$ than for $c_2$.
    \label{ex:seed}
    \vspace{-2mm}
\end{exam}}
\vspace{-0.2mm}
\spara{Remarks.}
We assume that the opinion diffusion for multiple candidates proceeds concurrently and independently, following \cite{TuAG20,GuoCW21,GarimellaGPT17,KKC18}.
{\bf (1)}
For the cumulative score, due to its aggregate nature, the top-$k$ seeds for the target candidate can be computed independent of the others, similar to the single-campaigner setting \cite{GionisTT13, AKPT18}. In contrast, {\em our other voting-based scores (\revise{plurality}, $p$-approval, positional-$p$-approval, and \revise{Copeland}) incorporate competition among the candidates via ranking-based formulations using each user's preference order.}
{\bf (2)}
As long as we know the seed sets for the non-target candidates at the beginning of the diffusion (i.e., at time $0$), our algorithm can compute their opinions at any time horizon, and we select the $k$ seed nodes for the target campaign (also at time $0$) so as to maximize the target's voting-based score at the time horizon, \textit{relative to the placement of seeds for non-target candidates} at time $0$.
Thus, {\em while our analyses and techniques apply for this general case where the competing candidates have seeds, for simplicity of notation and exposition, we assume w.l.o.g. that the non-target candidates have no seeds}.
\revise{{\bf (3)} Since we find the seed set of size at most $k$ that maximizes the score of the target candidate, winning is not always guaranteed, because even after selecting the $k$ optimal seed nodes for the target candidate, another candidate may still have a higher score than the target. In that case, the target candidate needs more seeds 
to win. The following variant of our problem can mitigate this issue.}
\revise{
\begin{problem} [\textsf{FJ-Vote-Win}]
Given the initial opinion matrix $B^{(0)}$, a target candidate $c_q$,
influence matrix $W_q$, stubbornness matrix $D_q$, and a time horizon $t$,
find a set of seed nodes $S^* \subset V$ of minimum size $k^*$ such that the score for $c_q$ at timestamp $t$ is the largest among all candidates.
Formally,
\vspace{-2mm}
\begin{scriptsize}
\begin{align}
  &\scriptsize S^*_k = \argmax_{S\subset V, |S|=k} F \left( B^{(t)}[S], c_q \right) \nonumber\\
  &\scriptsize k^* = \min\left\{k: \left[F\left(B^{(t)}\left[S^*_k \right], c_q \right) > \max_{c_x \in C \setminus \{c_q\}} F\left( B^{(t)}\left[S^*_k \right], c_x \right)\right]\right\} \nonumber\\
  &\scriptsize S^*=S^*_{k^*}
  \vspace{-4mm}
\end{align}
\end{scriptsize}
\label{prob:fjvotewin}
\end{problem}
}
\vspace{-6mm}
\revise{In \S~\ref{sec:overview}, we show that a solution to Problem~\ref{prob:fjvote} can be extended to solve this new problem. 
}
\vspace{-1mm}
\section{Basic Results \& Solution Framework}
\vspace{-1mm}

In this section, we discuss the hardness of our problem (\S~\ref{sec:hard}) and the submodularity of our scores (\S~\ref{sec:subm}), followed by a greedy solution to our problem (\S~\ref{sec:overview}). All of these are a part of our novel contributions. \revise{A summary of these properties for all our scores is given in Table \ref{tab:prop}.}

\begin{table}[t!]
    \centering
    \vspace{-2mm}
    \scriptsize
    \revise{
    \caption{Properties of our voting-based scores}
    \vspace{-2mm}
    \begin{tabular}{c||c|c|c|c}
        \hline
        {\bf Score} & {\bf \NP-hard} & {\bf Non-negative} & {\bf Non-decreasing} & {\bf Submodular} \\ \hline \hline
        {\bf Cumulative} & Yes & Yes & Yes & Yes \\ \hline
        {\bf Plurality} & Yes & Yes & Yes & No \\ \hline
        {\bf $p$-Approval} & Yes & Yes & Yes & No \\ \hline
        {\bf Pos.-$p$-Appr.} & Yes & Yes & Yes & No \\ \hline
        {\bf Copeland} & Open & Yes & Yes & No \\ \hline
    \end{tabular}}
    \label{tab:prop}
    \vspace{-5mm}
\end{table}

\vspace{-2mm}
\subsection{Hardness}
\label{sec:hard}
\vspace{-1mm}
We show that the decision version of Problem \ref{prob:fjvote} is \NP-hard for the cumulative and \revise{plurality} scores.
\vspace{-2mm}
\begin{theor}
The decision version of Problem \ref{prob:fjvote} is \NP-hard with the cumulative score.
\label{th:np_cum}
\end{theor}
\vspace{-4mm}
\begin{proof} 
We prove by a reduction from the \NP-hard VERTEX COVER problem \cite{K72}.
A vertex cover in an undirected graph $G = (V, E)$ is a subset of nodes such that every edge
in $E$ is incident to at least one of them. Given $G$ and an integer $k$, the decision version
of the problem asks if $G$ contains a vertex cover of size at most $k$.

Let $|V|=n$ and $|E|=m$. $G$ is transformed into a directed graph $\mathcal{G} = (V, E')$,
where $E'$ contains directed edges $(u, v)$ and $(v, u)$ for each undirected edge $(u, v) \in E$.
We create two candidates $c_q$ (our target) and $c_x$. For each $y \in \{q, x\}$, we set the following:
for each $i \in V$, $b^{(0)}_{yi} = 0$, $d_{yi} = 0$; and for each $(i, j) \in E'$, $w_{y;ij} = 1/deg(j)$,
where $deg(v)$ denotes the degree of node $v$ in $G$. This ensures that $W_y$ is column-stochastic.
The time horizon $t$ is set to $1$.
This reduction takes $\bigO(m+n)$ time.
We prove that a set $S$ of at most $k$ nodes is a vertex cover of $G$ if and only if
$F ( B^{(1)}[S], c_q ) \geq n$.

{\bf (1)} If $S$ is a vertex cover in $G$, then each node $v$ in $\mathcal{G}$ either belongs to $S$ or
has all of its incoming neighbors in $S$. In the former case, $b^{(1)}_{qv}[S] = 1$ by definition.
In the latter case, since $W_q$ is column-stochastic, it follows from Eq.~\ref{eq:fj} that $b^{(1)}_{qv}[S] = 1$.
This implies that $F ( B^{(1)}[S], c_q ) = n$.
{\bf (2)} If $S$ is not a vertex cover in $G$, then there exists at
least one edge $(u, v) \in E$ such that neither $u$ nor $v$ is in $S$.
This implies that $b^{(1)}_{qv}[S] \leq 1 - 1/deg(v) < 1$, which means that $F ( B^{(1)}[S], c_q ) < n$.
The theorem follows. 
\end{proof}
\vspace{-3mm}
\spara{Remark}: While  Problem~\ref{prob:fjvote} with the cumulative score is similar to \cite{GionisTT13}, a key difference is as follows. Unlike Problem~\ref{prob:fjvote}, \cite{GionisTT13} selects seeds to maximize the sum of the expressed opinions {\em at the Nash equilibrium}, instead of at \textit{a given finite time horizon}. The proofs of NP-hardness and submodularity in \cite{GionisTT13} rely on showing that an absorbing random walk is an unbiased estimate of the true equilibrium opinion. However, we cannot use absorbing random walks to estimate opinions at {\em a finite time horizon},  rendering their proofs inapplicable in our case. {\em Our \NP-hardness and submodularity proofs for the cumulative score are novel contributions}.

\vspace{-2mm}
\begin{theor}
  The decision version of Problem \ref{prob:fjvote} is \NP-hard with the \revise{plurality} score.
  \vspace{-1mm}
  \label{th:np_best}
\end{theor}
 \vspace{-3mm}
\begin{proof} 
The reduction remains the same as in the proof of Theorem \ref{th:np_cum}, except that $c_x$ satisfies $b^{(0)}_{xv} = 1 - \delta \; \forall v \in V$, where $0 < \delta < \min_{v \in V} 1/deg(v)$; this ensures that $b^{(1)}_{xv} = 1 - \delta$.
\end{proof}
The computational complexity of Problem \ref{prob:fjvote} with the \revise{Copeland} score is
open as of now. We, however, show in \S~\ref{sec:subm} that 
the \revise{Copeland} score is not submodular.
\vspace{-2mm}
\subsection{Submodularity}
\label{sec:subm}
\vspace{-1mm}
We show that the cumulative score used in Problem \ref{prob:fjvote} is submodular, while the \revise{plurality} and \revise{Copeland} scores are not. 
A set function $f:2^V\rightarrow \mathbb{R}^{\geq 0}$ over a ground set $V$ is
submodular if $f(X\cup\{i\})-f(X)\geq f(Y\cup\{i\})-f(Y), \, \forall X\subset Y\subset V, i\in V\setminus Y$.
The classic greedy algorithm returns a $(1-1/e)$-approximate solution for
maximizing a non-negative, non-decreasing, submodular function \cite{NWF78}. Including a user
$s$ into the seed set $S$ will increase her opinion value on $c_q$, which will in turn influence those of some other users.
Thus, after the inclusion
of $s$ into $S$, each user's opinion value and ranking of $c_q$ cannot decrease.
Hence, \textit{all our scoring functions are non-decreasing in seed sets for $c_q$}.

\vspace{-0.4mm}
\spara{Submodularity of the Cumulative Score.}
\vspace{-1mm}
\begin{theor}
The opinion value of any user $i$ about any candidate $c_q$ is submodular w.r.t. the seed set 
for that candidate. Formally, $\forall X\subseteq Y\subseteq V, s\in V\setminus Y$,
\begin{small}
\begin{align}
  b_{qi}^{(t)}[X\cup \{s\}]-b_{qi}^{(t)}[X]\geq b_{qi}^{(t)}[Y\cup \{s\}]-b_{qi}^{(t)}[Y]
\end{align}
\end{small}
\vspace{-5mm}
\label{lm:sub_opin}
\end{theor}
%
%
\begin{proof}

We prove by induction on $t$. First, we prove for the base case ($t=0$). There are two sub-cases:

\noindent (1) When $i=s$, the initial opinion of $s$ will increase to 1.
\begin{small}
\begin{align}
b_{qs}^{(0)}[S\cup \{s\}]-b_{qs}^{(0)}[S]=1-b_{qs}^{(0)}\geq 0, \quad S\in\{X,Y\}
\end{align}
\end{small}
\noindent (2) When $i\neq s$, the initial opinion of node $i$ will not be affected by the inclusion of $s$ into the seed set $S$. We have:
\begin{small}
\begin{align}
b_{qi}^{(0)}[S\cup \{s\}]-b_{qi}^{(0)}[S]=0, \quad S\in\{X,Y\}
\end{align}
\end{small}
In each sub-case, the marginal gain is non-negative and the same irrespective of whether the current seed set is $X$ or $Y$.
Thus, the submodularity holds for the base case.

Next, we prove for the induction step. Assuming that the submodularity holds at any time-stamp $t$, we prove that it also
holds at the next time-stamp $t+1$, by considering two sub-cases as below.

\noindent (1) When $i=s$, we increase the stubbornness value $d_{qs}$ of node $s$ to 1, which ensures that
its opinion value remains the same as the initial opinion value 1, in any future time-stamp. Thus, for $S \in \{X,Y\}$, we have:
\begin{small}
\begin{align}
b_{qs}^{(t+1)}[S\cup \{s\}]-b_{qs}^{(t+1)}[S]=1-b_{qs}^{(t+1)}[S]\geq 0
\label{eq:sub_cumu_case1}
\end{align}
\end{small}
Since the opinion values are non-decreasing with respect to the inclusion of seed nodes, and $X\subseteq Y$, we have:
\begin{small}
\begin{align}
b_{qs}^{(t+1)}[X]\leq b_{qs}^{(t+1)}[Y]
\label{eq:mon}
\end{align}
\end{small}
Based on Equations~\ref{eq:sub_cumu_case1} and \ref{eq:mon}, we derive:
\begin{small}
\begin{align}
b_{qs}^{(t+1)}[X\cup \{s\}]-b_{qs}^{(t+1)}[X]\geq b_{qs}^{(t+1)}[Y\cup \{s\}]-b_{qs}^{(t+1)}[Y]
\end{align}
\end{small}
(2) When $i\neq s$, following the FJ model (Equation~\ref{eq:fj}), we compute the marginal gain as follows,
where $S\in\{X,Y\}$.
\begin{small}
\begin{align}
&b_{qi}^{(t+1)}[S\cup \{s\}]-b_{qi}^{(t+1)}[S]\nonumber \\
&= \underbrace{\left( 1-d_{qi}[S] \right) \sum_{j=1}^{n}\left[\left(b_{qj}^{(t)}[S\cup \{s\}]-b_{qj}^{(t)}[S]\right)\cdot w_{ji}\right]}_{\text{1st term}}\nonumber \\
    & \qquad + \underbrace{b_{qi}^{(0)}[S\cup \{s\}]\cdot d_{qi}[S\cup \{s\}] - b_{qi}^{(0)}[S]\cdot d_{qi}[S]}_{\text{2nd term}}\nonumber \\
    & = \left( 1-d_{qi}[S] \right) \sum_{j=1}^{n}\left[\left(b_{qj}^{(t)}[S\cup \{s\}]-b_{qj}^{(t)}[S]\right)\cdot w_{ji}\right]
    \label{eq:marginal_submodular}
\end{align}
\end{small}
In the above, the second term vanishes because $d_{qi}[S]=d_{qi}[S\cup \{s\}]$
and $b_{qi}^{(0)}[S]=b_{qi}^{(0)}[S\cup \{s\}]$. Notice that in the first term, we also use
the fact that $d_{qi}[S]=d_{qi}[S\cup \{s\}]$.

Now, let us consider the seed set $S$ as $X$ and $Y$, respectively. By the induction hypothesis, we have:
\begin{small}
\begin{align}
b_{qj}^{(t)}[X\cup \{s\}]-b_{qj}^{(t)}[X]\geq b_{qj}^{(t)}[Y\cup \{s\}]-b_{qj}^{(t)}[Y]
\label{eq:induction}
\end{align}
\end{small}
Furthermore, by the definition of seed set, we get:
\begin{small}
\begin{align}
1-d_{qi}[X]&=1-d_{qi}[Y]=0\quad\forall i\in X \nonumber \\
1-d_{qi}[X]&=1-d_{qi}\geq 1-d_{qi}[Y]=0\quad\forall i\in Y\setminus X \nonumber \\
1-d_{qi}[X]&=1-d_{qi}[Y]=1-d_{qi}\quad\forall i\in V\setminus Y, \quad i\neq s \nonumber
\end{align}
\end{small}
To summarize,
\begin{small}
\begin{align}
1-d_{qi}[X]\geq 1-d_{qi}[Y]\quad\forall i\neq s
\label{eq:stubborn_submodular}
\end{align}
\end{small}
Therefore, we derive the following $\forall i\neq s$.
\begin{small}
\begin{align}
&b_{qi}^{(t+1)}[Y\cup \{s\}]-b_{qi}^{(t+1)}[Y] \nonumber \\
    & = \left( 1-d_{qi}[Y] \right) \sum_{j=1}^{n}\left[\left(b_{qj}^{(t)}[Y\cup \{s\}]-b_{qj}^{(t)}[Y]\right)\cdot w_{ji}\right] \quad {\text{by Eq.~\ref{eq:marginal_submodular}}}\nonumber \\
    &\leq \left( 1-d_{qi}[Y] \right) \sum_{j=1}^{n}\left[\left(b_{qj}^{(t)}[X\cup \{s\}]-b_{qj}^{(t)}[X]\right)\cdot w_{ji}\right] \quad {\text{by Eq.~\ref{eq:induction}}}\nonumber \\
    &\leq \left( 1-d_{qi}[X] \right) \sum_{j=1}^{n}\left[\left(b_{qj}^{(t)}[X\cup \{s\}]-b_{qj}^{(t)}[X]\right)\cdot w_{ji}\right] \quad {\text{by Eq.~\ref{eq:stubborn_submodular}}}\nonumber \\
    &=b_{qi}^{(t+1)}[X\cup \{s\}]-b_{qi}^{(t+1)}[X] \quad {\text{by Eq.~\ref{eq:marginal_submodular}}}
\end{align}
\end{small}
This completes the proof.
\end{proof}

The cumulative score is the sum of all users' opinion values (Equation~\ref{eq:cumu}).
As the sum of submodular functions is also submodular, the cumulative score is submodular.

%
%
\spara{Non-Submodularity of the Other Scoring Functions.}
We show the non-submodularity of the \revise{plurality} and \revise{Copeland} scores using the same running example (Figure~\ref{fig:example_pre} and Table~\ref{tab:example}).

\revise{
\begin{exam}
    \vspace{-2mm}
    As shown in Table~\ref{tab:example}, inserting node 2 into the empty seed set results in zero marginal gain for both the plurality and Copeland scores. However, inserting node 2 into seed set $\{1\}$ will make user 3 preferring the target candidate $c_1$ (resulting in marginal gain $1$ for the plurality score) and also the number of users preferring $c_1$ more than the same for $c_2$ (resulting in marginal gain $1$ for the Copeland score). Hence, submodularity is violated for both scores.
    \vspace{-2mm}
    \label{ex:sub}
\end{exam}
}

\vspace{-1mm}
\subsection{Solution Overview}
\label{sec:overview}
\vspace{-1mm}
Since the cumulative score is non-negative, non-decreasing, and submodular,
the greedy framework (Algorithm~\ref{alg:whole}), which identifies the node that maximizes the
marginal gain in score at each round, can provide a $(1-1/e)$-approximate solution.
\revise{We show in \S~\ref{sec:practical_effectiveness} that there is a problem instance for which the well-known submodularity ratio $\psi$ \cite{BianB0T17,DasK18} becomes 0 for our other non-submodular voting-based scores; thus their approximation factor $(1-e^{-\psi})$ degrades and goes to 0.
However, in \S~\ref{sec:sandwich}, with the help of {\em Sandwich Approximation} \cite{LCL15}, we prove that the greedy
framework can still generate good approximate solutions for these scores.}

\begin{algorithm}[tb!]
  \caption{\small \texttt{Greedy} Seed Selection to Maximize the Score}
  \scriptsize
  \begin{algorithmic}[1]
    \REQUIRE Graph $\mathcal{G}=(V,E)$, initial opinion matrix $B^{(0)}$, influence matrix $W_i$ and stubbornness matrix $D_i$ for each candidate $c_i$, target candidate $c_q$, seed set size budget $k$, time horizon $t$, and a scoring function $F$
    \ENSURE Seed set $S^*$ of size $k$
    \STATE $S^* \gets \emptyset$
    \FOR {$j=1$ \TO $k$}
    \STATE $u \gets \argmax_{v \in V\setminus S^*} \left[ F \left( B^{(t)}[S^*\cup \{v\}], c_q \right) - F \left( B^{(t)}[S^*], c_q \right) \right]$ \label{line:max_gain}
    \STATE $S^* \gets S^*\cup \{u\}$
    \ENDFOR
    \RETURN $S^*$
  \end{algorithmic}
  \label{alg:whole}
\end{algorithm}
\vspace{-0.6mm}
\spara{Time Complexity with the Cumulative Score.}
To find the node that maximizes the marginal gain at each round of Algorithm~\ref{alg:whole},
one can apply
Eq.~\ref{eq:fj} $t$ times (due to the input time horizon $t$).
Since every such matrix-vector multiplication has time complexity $\bigO(m)$ using a sparse matrix package, we have $k$ rounds
(to find the top-$k$ seed nodes), and $\bigO(n)$ candidate nodes from which a seed node is selected in each round,
the final time complexity is $\bigO(ktmn)$.
As the cumulative score is monotone and submodular, we also apply the CELF optimization \cite{LeskovecKGFVG07}.
In \S~\ref{sec:random_walk} and
\S~\ref{sec:sketching},
we propose {\em random walk}- and {\em sketching}-based
estimation, respectively, to further improve
the efficiency, with theoretical quality guarantees.


\begin{algorithm}[tb!]
  \caption{\small \texttt{Greedy} Seed Selection for Winning}
  \scriptsize
  \begin{algorithmic}[1]
    \REQUIRE Graph $\mathcal{G}=(V,E)$, initial opinion matrix $B^{(0)}$, influence matrix $W_i$ and stubbornness matrix $D_i$ for each candidate $c_i$, target candidate $c_q$, time horizon $t$ and a scoring function $F$
    \ENSURE Seed set $S^*$ of minimum size for $c_q$ to win
    \STATE $S^* \gets \emptyset$, $l \gets 0$, $u \gets n$
    \WHILE {$u - l > 1$}
    \STATE $k \gets \frac{1}{2} \left( l + u \right)$
    \STATE $S \gets$ Algorithm \ref{alg:whole} with seed set size $k$
    \IF {$F\left(B^{(t)}\left[S \right], c_q \right) > \max_{c_x \in C \setminus \{c_q\}} F\left( B^{(t)}\left[S \right], c_x \right)$}
    \STATE $u \gets k$, $S^* \gets S$
    \ELSE
    \STATE $l \gets k$
    \ENDIF
    \ENDWHILE
    \RETURN $S^*$
  \end{algorithmic}
  \label{alg:win}
\end{algorithm}

\vspace{-0.6mm}
\spara{\revise{Remark.}}
\revise{
{\bf (1)} This greedy solution can be extended to solve Problem~\ref{prob:fjvotewin} about finding the smallest seed set size 
$k^*$ such that the target candidate wins. Since $0 \leq k^* \leq n$ and our scoring functions are non-decreasing, we resort 
to a binary search for $k^*$, with the initial lower (resp. upper) bound as $0$ (resp. $n$). In each iteration, we compute 
the optimal seed set $S$ of size at most the value midway between the bounds. If the target wins (resp. loses) with the seed set $S$, 
the upper (resp. lower) bound is updated to the middle value and the process is repeated till the bounds converge. The overall pseudocode is shown in Algorithm \ref{alg:win}.
{\bf (2)} Due to the hardness of our problem (\S~\ref{sec:hard}), we find an ``approximately optimal'' seed set 
(e.g., using Algorithm \ref{alg:whole}). Since such a seed set will lead to a lower voting-based score than that 
for the optimal solution, the final seed set size obtained could be larger than the true minimal one to achieve the 
winning criterion.}

%% file: 3-algo.tex
\section{\revise{Plurality} Variants and \revise{Copeland} Scores: Sandwich Approximation}
\label{sec:sandwich}
\vspace{-1mm}
{\em Sandwich Approximation} \cite{LCL15} (\S \ref{sec:sandwich_approx}) is a powerful framework for providing  approximation guarantees (possibly lower than $(1-1/e)$) for non-submodular function maximization. 
\revise{Our novel contribution is to construct non-trivial upper and lower bound functions to enable sandwich approximation for
our \revise{plurality} (\S~\ref{sec:b_best}) and \revise{Copeland} (\S~\ref{sec:b_cond}) scores, as they must satisfy certain properties to admit good approximations. Furthermore, we empirically validate that the additional ratio introduced by sandwich approximation (which degrades the overall approximation) is reasonably high for our proposed bounding functions in all cases (\S~\ref{sec:practical_effectiveness}).} For simplicity, we re-write $F(B^{(t)}[S],c_q)$ as $F(S)$, since the target candidate $c_q$ is arbitrary but fixed.

\begin{algorithm}[tb!]
	\caption{\small Sandwich Approximation-Based Seed Selection}
	\scriptsize
	\begin{algorithmic}[1]
		\REQUIRE Graph $\mathcal{G}=(V,E)$, initial opinion matrix $B^{(0)}$, influence matrix $W_i$ and stubbornness matrix $D_i$ for each candidate $c_i$, target candidate $c_q$, seed set size budget $k$, time horizon $t$, and a (non-submodular) scoring function $F$ with a lower bound $LB$ and an upper bound $UB$
		\ENSURE Seed set $S^*$ of size $k$
		\STATE $S_U \gets \eta$-approximate solution to maximize $UB(\cdot)$
		\STATE $S_L \gets \tau$-approximate solution to maximize $LB(\cdot)$
		\STATE $S_F \gets$ Feasible solution (e.g. standard greedy) for $F(\cdot)$
		\RETURN $\argmax_{S\in \{S_{U},S_{L},S_F\}} F(S)$
	\end{algorithmic}
	\label{alg:sandwich}
\end{algorithm}

\vspace{-1mm}
\subsection{Sandwich Approximation}
\label{sec:sandwich_approx}
\vspace{-1mm}
For any non-submodular set function $F(S)$, $S\subseteq V$, suppose $UB(S)$ and $LB(S)$ are any set functions defined on the same ground set $V$, such that $LB(S)\leq F(S)\leq UB(S)$, $\forall S\subseteq V$. If we are able to compute approximate solutions for both $UB(S)$ and $LB(S)$,
then we can obtain the sandwich approximation for the targeted set function $F(S)$ as follows \revise{(pseudocode in Algorithm \ref{alg:sandwich})}.
{\bf (1)} Run the approximation algorithms to obtain an $\eta$-approximate solution $S_{U}$ to $UB(S)$ and a $\tau$-approximate solution $S_{L}$ to $LB(S)$, where $\eta$ (resp. $\tau$) is the approximation factor afforded by the algorithm for $UB(S)$ (resp. $LB(S)$).
{\bf (2)} Find a feasible solution $S_F$ to function $F(S)$, e.g., by applying the standard greedy algorithm.
{\bf (3)} Report the final solution $S^\#$:
$S^\#=\argmax_{S\in \{S_{U},S_{L},S_F\}} F(S)$.
\vspace{-4mm}
\begin{theor}[\cite{LCL15}]
  Sandwich approximation guarantees:
  \begin{small}
  \begin{align}
    \vspace{-3mm}
    F \left( S^\# \right) \geq \max \left\{ \eta \cdot \frac{F \left( S_{U} \right)}{UB \left( S_{U} \right)} \cdot F \left( S^*_F \right), \tau \cdot LB \left( S^*_F \right) \right\}
    \label{eq:sandwich}
  \end{align}
  \end{small}
  where $S^*_F$ 
  maximizes $F(S)$ subject to a constraint, e.g., a cardinality constraint $|S|\leq k$, or a matroid constraint.
  \label{th:sandwich}
\end{theor}
%
%
%
\vspace{-2mm}
\subsection{Bounds on the \revise{Plurality} Score Variants}
\label{sec:b_best}
\vspace{-1mm}
Motivated by 
this result, we design {\em non-negative, non-decreasing, submodular} lower and upper bounding functions $LB(S)$ and $UB(S)$ such that $0 \leq LB(S) \leq F(S) \leq UB(S)$ $\forall S \subseteq V$, thereby enabling sandwich approximation with $\eta = \tau = 1 - 1/e$ (Eq.~\ref{eq:sandwich}) via running the greedy algorithm (Algorithm~\ref{alg:whole}) on $LB(S)$, $F(S)$, and $UB(S)$, respectively. Note that ensuring the submodularity of $LB(\cdot)$ and $UB(\cdot)$ is one (not the only) way to enable sandwich approximation. This analysis is for $F(\cdot)$ denoting the positional-$p$-approval score; thus, it also holds for special cases, e.g.,  \revise{plurality} and $p$-approval scores. We first define two useful terms.
\begin{defn} [Favorable Users Set]
\label{def:fav}
The favorable users set, denoted by $V_q^{(t)}$, is the set of nodes (users) who would have the target candidate $c_q$ among their top-$p$ ranked candidates (according to their opinion values) at the time horizon $t$, even without introducing any seed for $c_q$. Formally,
\begin{small}
\begin{equation}
\label{eq:vq}
V_q^{(t)} = \left\{ v \in V : \beta \left( b_{qv}^{(t)} \right) \leq p \right\}
\vspace{-2mm}
\end{equation}
\end{small}
\end{defn}
Since the opinion of a user about $c_q$ increases with the seed set for $c_q$, and the users in $V_q^{(t)}$ have $c_q$ among their top-$p$ ranked candidates at the time horizon $t$ even without any seed for $c_q$, they will continue doing so on the addition of seed nodes for $c_q$.
Recall that the set of such users at the time horizon $t$ decides $c_q$'s positional-$p$-approval score.
Hence, we use $V_q^{(t)}$ to construct a lower bound for the positional-$p$-approval score (Definition~\ref{def:lb_best}).
\begin{defn} [Reachable Users Set]
\label{def:reachable}
The reachable users set, denoted by $N_S^{(t)}$, is the set of nodes (users) at most $t$ outgoing hops away from any node in a seed set $S$. Formally, denoting by $u \overset{h}{\rightsquigarrow} v$ the existence of a path with $h$ edges from $u$ to $v$,
\begin{small}
\begin{equation}
\label{eq:ns}
N_S^{(t)} = \bigcup_{s \in S} \bigcup_{h=0}^t \left\{ v \in V : s \overset{h}{\rightsquigarrow} v \right\}
\vspace{-2mm}
\end{equation}
\end{small}
\end{defn}
On adding seeds for $c_q$, along with the users in $V_q^{(t)}$, some additional users could also have higher opinions about $c_q$ at time $t$, who according to FJ model, can only be at most $t$ outgoing hops away from any seed node. Hence, $V_q^{(t)}$ and $N_S^{(t)}$ are used to construct an upper bound for the positional-$p$-approval score (Definition~\ref{def:ub_best}).
\begin{defn}
\label{def:lb_best}
The lower bounding function $LB(S)$ for the positional-$p$-approval score $F(S)$ is defined as the aggregated opinion value about $c_q$ at time $t$ for all users in the favorable users set, on the introduction of a seed set $S$ for $c_q$, times the weight $\omega[p]$ for position $p$. Formally,
\begin{small}
\begin{align}
LB(S) &= \omega[p]\sum_{v \in V_q^{(t)}} b_{qv}^{(t)}[S]
\vspace{-2mm}
\end{align}
\end{small}
\end{defn}
\vspace{-1mm}
\begin{defn}
\label{def:ub_best}
The upper bounding function $UB(S)$ for the positional-$p$-approval score $F(S)$ is defined as the total number of users either in the favorable users set or in the reachable users set, times the weight $\omega[1]$ for position $1$. Formally,
\begin{small}
\begin{align}
UB(S) &= \omega[1] \left| N_S^{(t)} \cup V_q^{(t)} \right|
\vspace{-2mm}
\end{align}
\end{small}
\end{defn}
\spara{Correctness Guarantee.}
We now have: 
%
\vspace{-1mm}
\begin{theor}
\label{lm:lb_br}
$LB(S)$ is {\bf (1)} non-negative, {\bf (2)} non-decreasing, {\bf (3)} submodular, and {\bf (4)} a lower bound for $F(S)$.
\end{theor}
\vspace{-3mm}
\begin{proof} 
{\bf (1)} Since $b_{qv}^{(t)}[S] \geq 0 \; \forall v \in V$ and $\omega[p] \geq 0$, $LB(S) \geq 0$.
{\bf (2)} $LB(S)$ is the sum of $b_{qv}^{(t)}[S]$ $\forall v \in V_q^{(t)}$ (multiplied by a non-negative constant $\omega[p]$), and each of them is non-decreasing
w.r.t. the inclusion of seeds in $S$.
{\bf (3)} From Theorem \ref{lm:sub_opin}, each $b_{qv}^{(t)}[S]$ is submodular, and hence so is $LB(S)$, which
is the sum of such functions  $\forall v \in V_q^{(t)}$ multiplied by a non-negative constant $\omega[p]$.
{\bf (4)} Notice that $b_{qv}^{(t)}[S] \geq b_{qv}^{(t)}$. Thus, $v \in V_q^{(t)}$ implies $\beta \left( b_{qv}^{(t)}[S] \right) \leq \beta \left( b_{qv}^{(t)} \right) \leq p$ or $\mathbbm{1} \left[ \beta \left( b_{qv}^{(t)}[S] \right) \leq p \right] = 1$; so $\omega \left[ \beta \left( b_{qv}^{(t)}[S] \right) \right] \geq \omega[p]$. Hence,
\begin{small}
\begin{align*}
LB(S) &= \omega[p] \sum_{v \in V_q^{(t)}} b_{qv}^{(t)}[S] \leq \sum_{v \in V_q^{(t)}} \omega[p] \\
&\leq \sum_{v \in V_q^{(t)}} \omega \left[ \beta \left( b_{qv}^{(t)}[S] \right) \right] \times \mathbbm{1} \left[ \beta \left( b_{qv}^{(t)}[S] \right) \leq p \right] \\
&\leq \sum_{v \in V} \omega \left[ \beta \left( b_{qv}^{(t)}[S] \right) \right] \times \mathbbm{1} \left[ \beta \left( b_{qv}^{(t)}[S] \right) \leq p \right] = F(S)
\vspace{-3mm}
\end{align*}
\end{small}
\end{proof}
\begin{lem}
\label{lm:nst}
If a user $v$ is not in the reachable users set, then the opinion of $v$ about $c_q$ does not change by virtue of the seed set. Formally, if $v \notin N_S^{(t)}$, then $b_{qv}^{(t)}[S] = b_{qv}^{(t)}$.
\end{lem}
\vspace{-1mm}
Intuitively, this follows from the FJ model; the influence of the seed set diffuses by one hop in each timestamp, and hence cannot spread beyond $t$ hops at timestamp $t$.
\begin{theor}
\vspace{-1mm}
$UB(S)$ is {\bf (1)} non-negative, {\bf (2)} non-decreasing, {\bf (3)} submodular, and {\bf (4)} an upper bound for $F(S)$.
\label{th:ub_br}
\end{theor}
%
\begin{proof} 
{\bf (1)} Since the size of any set is non-negative, $UB(S) \geq 0$.
{\bf (2)} $UB(S)$ is non-decreasing because, for any $X \subseteq Y$,
\begin{small}
\begin{align*}
UB(Y) &= UB(Y \cup X) = \omega[1] \left| N_{Y \cup X}^{(t)} \cup V_q^{(t)} \right| \\
&= \omega[1] \left| N_Y^{(t)} \cup N_X^{(t)} \cup V_q^{(t)} \right| \geq \omega[1] \left| N_X^{(t)} \cup V_q^{(t)} \right| = UB(X)
\end{align*}
\end{small}
{\bf (3)} $UB(S)$ is submodular because, for any $X \subset Y \subset V$ and $s \in V \setminus Y$, we have
\begin{footnotesize}
\begin{align*}
&UB(X \cup \{s\}) - UB(X) = \omega[1] \left( \left| N_{X \cup \{s\}}^{(t)} \cup V_q^{(t)} \right| - \left| N_X^{(t)} \cup V_q^{(t)} \right| \right) \\
&= \omega[1] \left| N_{\{s\}}^{(t)} \setminus \left( N_X^{(t)} \cup V_q^{(t)} \right) \right| \geq \omega[1] \left| N_{\{s\}}^{(t)} \setminus \left( N_Y^{(t)} \cup V_q^{(t)} \right) \right| \\
&= \omega[1] \left( \left| N_{Y \cup \{s\}}^{(t)} \cup V_q^{(t)} \right| - \left| N_Y^{(t)} \cup V_q^{(t)} \right| \right) = UB(Y \cup \{s\}) - UB(Y)
\end{align*}
\end{footnotesize}
{\bf (4)} Suppose $v \notin N_S^{(t)}$ and $v \notin V_q^{(t)}$. From Eq. \ref{eq:vq} and Lemma \ref{lm:nst}, $
\beta \left( b_{qv}^{(t)}[S] \right) = \beta \left( b_{qv}^{(t)} \right) > p$. Thus, $\beta \left( b_{qv}^{(t)}[S] \right) \leq p$ implies that $v \in N_S^{(t)} \cup V_q^{(t)}$. Also, $\omega \left[ \beta \left( b_{qv}^{(t)}[S] \right) \right] \leq \omega[1]$. Hence,
\begin{footnotesize}
\begin{align*}
F(S) &= \sum_{v \in V} \omega \left[ \beta \left( b_{qv}^{(t)}[S] \right) \right] \times \mathbbm{1} \left[ \beta \left( b_{qv}^{(t)}[S] \right) \leq p \right] \\
&\leq \sum_{v \in V} \omega[1] \times \mathbbm{1} \left[ v \in N_S^{(t)} \cup V_q^{(t)} \right] = \omega[1] \left| N_S^{(t)} \cup V_q^{(t)} \right| = UB(S)
\end{align*}
\end{footnotesize}
\vspace{-2mm}
\end{proof}
\vspace{-4mm}
\subsection{Upper Bound for the \revise{Copeland} Score}
\label{sec:b_cond}
We construct a non-negative, non-decreasing, submodular upper bounding function for the \revise{Copeland} score
in a similar way as in \S~\ref{sec:b_best}, under the constraint that no user has equal opinion values
about any two candidates at the time horizon. Notice that this constraint does not change the definition of
the \revise{Copeland} score (Equation \ref{eq:condorcet}) in any way; rather, whether this constraint holds or not depends
on the input dataset, the seed set, and the time horizon.
We enable sandwich approximation via running the
greedy algorithm (Algorithm~\ref{alg:whole}) on $F(S)$ and $UB(S)$ only,
and we get $\eta = 1 - 1/e$ in Equation \ref{eq:sandwich}.
As in \S~\ref{sec:b_best}, ensuring the submodularity of $UB(\cdot)$ is one (not the only) way to enable sandwich approximation.
The construction of a useful lower bound and the case when a user has equal preference to two candidates
at the time horizon are interesting open questions for future work.
\begin{defn}[Weakly Favorable Users Set]
  The weakly favorable users set, denoted by $U_{q}^{(t)}$, is the set of nodes (users) who prefer $c_q$ to at least one other candidate at the time horizon $t$, even without having any seed for $c_q$. Formally,
  \begin{small}
  \begin{equation}
    \vspace{-2mm}
    \label{eq:vq'}
    U_{q}^{(t)} = \left\{ v \in V : b_{qv}^{(t)} > \min_{c_x \in C \setminus \{c_q\}} b_{xv}^{(t)} \right\}
    \end{equation}
  \end{small}
  \vspace{-4mm}
  \label{def:max_fav}
\end{defn}
Since the \revise{Copeland} score computes the number of one-on-one competitions won by $c_q$, only those users who prefer $c_q$ to at least one other candidate, i.e., those in $U_{q}^{(t)}$, can contribute to this score, along with those users who could be influenced by the seed set, i.e., those in $N_S^{(t)}$. Thus, $U_{q}^{(t)}$ and $N_S^{(t)}$ are used to construct an upper bound as below.
\begin{defn}
The upper bounding function $UB(S)$ for the \revise{Copeland} score $F(S)$ is defined as the total number of users either in the weakly favorable users set or in the reachable users set, times the ratio of the number of non-target candidates to one more than half the total number of users.
\begin{small}
\vspace{-1mm}
\begin{align}
UB(S) = \frac{r-1}{\left\lfloor\frac{n}{2}\right\rfloor + 1} \left| N_S^{(t)} \cup U_{q}^{(t)} \right|
\end{align}
\end{small}
\vspace{-4mm}
\label{def:ub_cond}
\end{defn}
\vspace{-0.4mm}
\spara{Correctness Guarantee.} We show that $UB(S)$ is a non-negative, non-decreasing, submodular upper bound for $F(S)$.
\begin{theor}
    $UB(S)$ is {\bf (1)} non-negative, {\bf (2)} non-decreasing, {\bf (3)} submodular, and {\bf (4)} an upper bound for $F(S)$.
    \label{th:ub_cond}
\end{theor}
\begin{proof} 
{\bf (1)}, {\bf (2)} and {\bf (3)} can be proved by similar arguments as their counterparts in Theorem \ref{th:ub_br}.

  {\bf (4)} Suppose $v \notin N_S^{(t)}$ and $v \notin U_{q}^{(t)}$. From Equation \ref{eq:vq'} and Lemma \ref{lm:nst}, $\forall c_x \in C \setminus \{c_q\} \, : \, b_{xv}^{(t)} \geq b_{qv}^{(t)} = b_{qv}^{(t)}[S]$. Thus, $b_{qv}^{(t)}[S] > \min_{c_x \in C \setminus \{c_q\}} b_{xv}^{(t)}$ implies that $v \in N_S^{(t)} \cup U_{q}^{(t)}$. Hence, we have
  \begin{footnotesize}
    \begin{align*}
            &F(S) = \sum_{c_x \in C\setminus \{c_q\}} \mathbbm{1} \left[ \sum_{v \in V} \mathbbm{1} \left[ b^{(t)}_{qv}[S] > b^{(t)}_{xv} \right] > \sum_{v\in V} \mathbbm{1} \left[ b^{(t)}_{qv}[S] < b^{(t)}_{xv} \right] \right] \\
            &= \sum_{c_x \in C\setminus \{c_q\}} \mathbbm{1} \left[ \sum_{v \in V} \mathbbm{1} \left[ b^{(t)}_{qv}[S] > b^{(t)}_{xv} \right] \geq \left\lfloor\frac{n}{2}\right\rfloor + 1 \right] \\
            &\leq \sum_{c_x \in C\setminus \{c_q\}} \frac{1}{\left \lfloor \frac{n}{2} \right \rfloor + 1} \sum_{v \in V} \mathbbm{1} \left[b^{(t)}_{qv}[S] > b^{(t)}_{xv}\right]\\
            & \leq \frac{1}{\left\lfloor\frac{n}{2}\right\rfloor + 1} \sum_{c_x \in C\setminus \{c_q\}} \sum_{v\in V} \mathbbm{1} \left[b^{(t)}_{qv}[S] > \min_{c_y \in C \setminus \{c_q\}} b_{yv}^{(t)}\right] \\
            &\leq \frac{r-1}{\left\lfloor\frac{n}{2}\right\rfloor + 1} \sum_{v\in V} \mathbbm{1}\left[v \in N_S^{(t)} \cup U_{q}^{(t)}\right] = \frac{r-1}{\left\lfloor\frac{n}{2}\right\rfloor + 1} \left| N_S^{(t)} \cup U_{q}^{(t)} \right|=UB(S)
    \end{align*}
  \end{footnotesize}
  The second step above holds since no user has equal opinion values about any two candidates at time $t$ by our assumption.
\end{proof}
\vspace{-2mm}
\subsection{Practical Effectiveness of our Bounds}
\label{sec:practical_effectiveness}
We empirically compute the ratio $\frac{F\left(S_{U}\right)}{UB\left(S_{U}\right)}$ in Equation~\ref{eq:sandwich}, since sandwich approximation ensures an approximation factor of at least $\frac{F\left(S_{U}\right)}{UB\left(S_{U}\right)}\left(1-\frac{1}{e}\right)$.
We vary the major parameter, the number of seeds ($k$), from 100 to 1000 (with gap 100); each value corresponds to a \emph{trial}.
The ratio reaches $0.7$ in 90\% of the trials;
and in about 50\% of the trials it exceeds $0.8$ for both the \revise{plurality} and \revise{Copeland} scores. This results in an empirical
approximation factor of at least $0.8(1-1/e) \approx 0.51$ in more than half of our trials. It is only once that the ratio turns
out to be below 0.5 (0.46 for the \revise{plurality} score on the {\em Twitter\_Social\_Distancing} dataset), which is the worst case we observe
empirically.
In practice, our algorithm performs much better than several baselines (\S \ref{sec:exp}). 

\begin{figure}[t!]
	\vspace{-1mm}
	\centering
	\includegraphics[scale=0.2,angle=0]{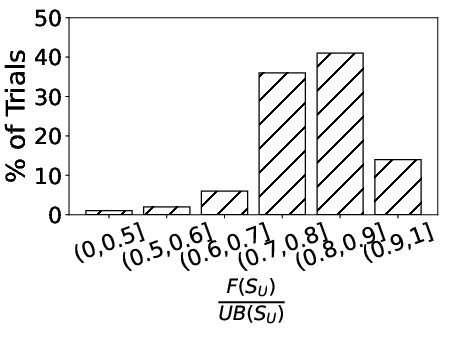}
	\includegraphics[scale=0.2,angle=0]{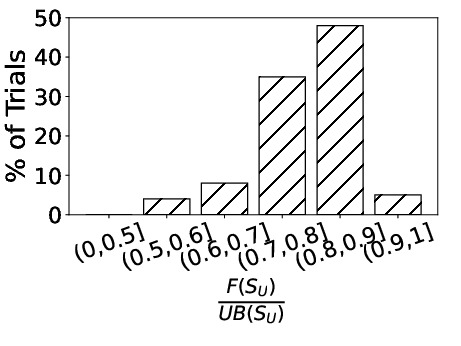}
 	\vspace{-3mm}
	\caption{\small Empirical study on the Sandwich approximation factor, with the \revise{plurality} score on the {\sf Twitter\_Social\_Distancing} dataset (left) and the \revise{Copeland} score on the {\sf Yelp} dataset (right). $100$ trials (runs) of the method were performed for each dataset-score pair.}
	\label{fig:bound}
 	\vspace{-4mm}
\end{figure}

The greedy algorithm for finding $S_U$ is much faster than that for computing $S_F$ (Algorithm \ref{alg:whole}), since it does not involve any expensive opinion computation. Meanwhile, $S_L$ is obtained via greedily maximizing the cumulative score on $V_q^{(t)}$ (Definition \ref{def:lb_best}), which is also much faster, since {\bf (1)} $\left| V_q^{(t)} \right| \ll |V|$ in practice, and {\bf (2)} the greedy algorithm for the cumulative score is much faster than that for the \revise{plurality} score (\S~\ref{sec:varyK}). Empirically, the running times for finding $S_U$ and $S_L$ are about 2\% and 5\%, respectively, of that for finding $S_F$.
\eat{\note[Laks]{I was not able to see where in \S~\ref{sec:exp} this finding is reported. Given that they all use greedy algorithms, what makes finding $S_U$ and $S_L$ so much faster than finding $S_F$? }}


\vspace{-0.4mm}
\spara{Remarks.} 
The notions of curvature, submodularity ratio, and submodularity index
have been exploited for establishing approximation guarantees of the greedy algorithm applied to the cardinality constrained maximization of non-submodular, non-decreasing
set functions \cite{BianB0T17,DasK18,ZhouS16}. For instance, when $F$ has submodularity ratio $\psi$, the greedy
algorithm for maximizing $F$ provides a $(1-e^{-\psi})$-approximation, where the submodularity ratio measures how
``close'' $F$ is to being submodular \cite{BianB0T17,DasK18}. Formally,
the submodularity ratio of $F$ is the largest scalar $\psi$ such that, for all $\Omega, S \subseteq V$,
\begin{small}
\begin{equation}
\label{eq:sr}
    \sum_{\omega \in \Omega \setminus S} [F(S \cup \{\omega\}) - F(S)] \geq \psi [F(S \cup \Omega) - F(S)]
\end{equation}
\end{small}
When the submodularity ratio of a function $F$ is 0, the approximation guarantee degrades and goes to 0 in limit. 
Unfortunately, as we show, there is a problem instance for which the submodularity ratio becomes 0 when $F$ denotes the \revise{plurality} score. 
\revise{
Consider the same running example as in Figure \ref{fig:example_pre}. From Table \ref{tab:example}, we have the following.
\begin{small}
\begin{align*}
    &B_2^{(1)} = [0.35, 0.75, 0.78, 0.9] \\
    &B_1^{(1)}[\emptyset] = [0.4, 0.8, 0.6, 0.75] \implies F(\emptyset) = 2 \\
    &B_1^{(1)}[\{1\}] = [1, 0.8, 0.75, 0.75] \implies F(\{1\}) = 2 \\
    &B_1^{(1)}[\{2\}] = [0.4, 1, 0.65, 0.75] \implies F(\{2\}) = 2 \\
    &B_1^{(1)}[\{1, 2\}] = [1, 1, 0.8, 0.75] \implies F(\{1, 2\}) = 3
\end{align*}
\end{small}
Clearly, in Equation \ref{eq:sr}, $\psi = 0$ for $S = \emptyset$ and $\Omega = \{1, 2\}$; hence, the submodularity ratio of $F$ is 0.}
The sandwich approximation method that we employ provides an alternative direction to derive an approximation guarantee
for the greedy algorithm applied to the cardinality constrained maximization of non-submodular, non-decreasing
set functions.

%% file: 4-random_walk.tex
\section{\revise{Efficient} Random Walk-based Estimation}
\label{sec:random_walk}
\vspace{-1mm}
The greedy framework (Algorithm~\ref{alg:whole})
has time complexity $\bigO(ktmn)$ via inefficient direct matrix-vector multiplication (\S~\ref{sec:overview} and 
\S~\ref{sec:sandwich}).
\revise{In this section, we first introduce a random walk interpretation for the opinion value of any node at any timestamp (\S~\ref{sec:rand_walk_interpret}). Next, as our novel contribution, an {\em efficient} random walk-based method with a {\em smart truncation strategy} is designed to estimate the marginal gain (\S~\ref{sec:random_walk_alg}). Finally, we establish novel {\em quality guarantees}
of the proposed method for {\em all} our voting-based scores (\S~\ref{sec:rw_quality}).}
\vspace{-2mm}
\subsection{Random Walk Interpretation}
\label{sec:rand_walk_interpret}
\vspace{-1mm}
%
As the influence matrix $W_q$ is column-stochastic for any candidate $c_q$, the probabilities on the outgoing edges of each node add up to $1$ in the reverse graph.\footnote{\scriptsize{The reverse graph has the same set of nodes and edges, but with edge directions reversed. The weights on the edges, now
interpreted as probabilities, remain the same.}}
This enables
the following {\em Direct Generation} of $t$-step random walks with seed set $S$.
{\bf (1)} Each node $v$ in the reverse graph has a termination probability $d_{qv}[S] \in [0,1]$ that is equivalent to
its stubbornness (recall that $d_{qv}[S] = 1$ if $v\in S$ and $d_{qv}[S]=d_{qv}$ otherwise),
and the probabilities on its outgoing edges add up to $1$.
{\bf (2)} If a random walk is at node $v$ in the current step, it terminates at $v$ with probability $d_{qv}[S]$.
Otherwise, it proceeds to an out-neighbor of $v$ chosen according to the edge probabilities.
{\bf (3)} From a start node $u$, we repeat step (2) to generate a random walk.
It terminates when step (2) has been conducted $t$ times, or the walk stops early (i.e., before reaching length $t$)
at a node due to the termination probability.
{\bf (4)} If the random walk terminates at node $v$, then the node $u$ at time $t$ adopts the
initial opinion of node $v$: $X_{qu}^{(t)}[S] = b_{qv}^{(0)}[S]$.
We show that the expected opinion value of any node $u$ at any time $t$ when serving as the start node of the above reverse
random walk is the same as the exact opinion value of $u$ at time $t$ computed by matrix-vector multiplication.\footnote{\scriptsize{Random walks for approximating matrix-vector multiplication are employed in \cite{CohenL99} and in PageRank \cite{ALNO07}, albeit with subtle differences from how they are applied in our work. While \cite{ALNO07,CohenL99} require a one-time estimation of the vector entries, we need the same for $k$ iterations of the greedy algorithm, and we do so in an efficient way. Also, the quality guarantees required are different from \cite{ALNO07, CohenL99} and specific to each voting-based score. For more details, we refer to Appendix~\ref{sec:rw_pr}.}}
\begin{algorithm}[tb!]
	\caption{\small Random Walk-Based \texttt{Greedy} Seed Selection}
	\scriptsize
	\begin{algorithmic}[1]
		\REQUIRE Graph $\mathcal{G}=(V,E)$, initial opinion matrix $B^{(0)}$, influence matrix $W_i$ and stubbornness matrix $D_i$ for each candidate $c_i$, target candidate $c_q$, seed set size budget $k$, time horizon $t$, and a scoring function $F$
		\ENSURE Seed set $S^*$ of size $k$
		\FORALL {$v \in V$}
		\FOR {$j = 1$ \TO $\lambda_v$}
		\STATE Generate a $t$-step reverse random walk starting from $v$
		\ENDFOR
		\ENDFOR
		\STATE $S^* \gets \emptyset$
		\FOR {$i = 1$ \TO $k$}
		\STATE $u \gets \argmax_{v \in V \setminus S^*} \left[ \widehat{F} \left( \widehat{B}^{(t)}[S^* \cup \{v\}], c_q \right) - \widehat{F} \left( \widehat{B}^{(t)}[S^*], c_q \right) \right]$
		\STATE $S^* \gets S^* \cup \{u\}$
		\STATE Truncate all walks containing $u$ at $u$
		\ENDFOR
		\RETURN $S^*$
	\end{algorithmic}
	\label{alg:rw}
\end{algorithm}

\begin{theor}
\label{th:expect_direct}
For any $t \geq 0$ and seed set $S$, the expected value of the estimated opinion $X^{(t)}_{qu}[S]$ of any user $u$ about any candidate $c_q$ at timestamp $t$ using a $t$-step reverse random walk by {\em Direct Generation} is equal to the exact opinion of $u$ about $c_q$ at timestamp $t$ according to the FJ model. Formally,
\begin{small}
\begin{align}
    \mathbbm{E} \left[ X^{(t)}_{qu}[S] \right] = b^{(t)}_{qu}[S]
\end{align}
\end{small}
\end{theor}
%
\begin{proof} 
We prove by induction on $t$. The base case ($t=0$) is trivial, since each node takes its initial opinion.
Next, assuming that the statement is true at timestamp $t$, we prove that it is true at timestamp $t+1$.
Let $\Pr \left( u \overset{t}{\rightsquigarrow} v \right)$ denote the probability that
a $t$-step reverse random walk starting from $u$ ends at $v$. Considering any $(t+1)$-step
reverse random walk from any node $u$, we have:
\begin{footnotesize}
\begin{align*}
    & \mathbbm{E}\left[X^{(t+1)}_{qu}[S]\right] = \sum_{v\in V} b^{(0)}_{qv}[S] \times \Pr \left( u \overset{t+1}{\rightsquigarrow} v \right) \\
    &= \left[\sum_{v\in V} b^{(0)}_{qv}[S] \left( 1 - d_{qu}[S] \right) \sum_{y\in V} w_{qyu} \times \Pr \left( y \overset{t}{\rightsquigarrow} v \right)\right] + b^{(0)}_{qu}[S] d_{qu}[S] \\
    &= \left[\left( 1 - d_{qu}[S] \right) \sum_{y\in V} w_{qyu} \sum_{v\in V} b^{(0)}_{qv}[S] \times \Pr \left( y \overset{t}{\rightsquigarrow} v \right) \right]+ b^{(0)}_{qu}[S] d_{qu}[S] \\
    &= \left( 1 - d_{qu}[S] \right) \sum_{y\in V} w_{qyu} \mathbbm{E}\left[X^{(t)}_{qy}[S]\right] + b^{(0)}_{qu}[S] d_{qu}[S] \\
    &= \left( 1 - d_{qu}[S] \right) \sum_{y\in V} b^{(t)}_{qy}[S] w_{qyu} + b^{(0)}_{qu}[S]d_{qu}[S] = b^{(t+1)}_{qu}[S]
\end{align*}
\end{footnotesize}
\end{proof}

\vspace{-4mm}
\subsection{The Algorithmic Workflow}
\label{sec:random_walk_alg}
%
We estimate the opinion of every user $v$
about any candidate $c_q$ at time $t$ by generating $\lambda_v$ independent $t$-step reverse
random walks starting from $v$.
The estimated opinion of node $v$
about candidate $c_q$ is computed as the average of the initial opinions of the
end nodes across all $\lambda_v$ random walks.
%
The seed set is generated greedily as in Algorithm~\ref{alg:whole}.
In Line~\ref{line:max_gain}, we select the best new seed based on the maximum \emph{estimated} marginal gain instead
of the maximum \emph{actual} marginal gain.
In each iteration, given the previously selected seed set $S^*$ for $c_q$, we need to compute
the marginal gain of including a candidate seed node $w$ into $S^*$, and hence the estimated opinions with the new seed set.
The {\em Direct Generation} approach would require the generation of new walks with the new seed set, which would be expensive.
Thus, we use an alternative {\em Post-Generation Truncation} technique as follows:
Before running Algorithm~\ref{alg:whole}, we generate (only once) $\lambda_v$ random walks from each node $v$ using the same approach as in
\S~\ref{sec:rand_walk_interpret} but with the empty seed set.
Thereafter, for any given seed set $S$, the estimated opinion $Y_{qv}^{(t)}[S]$ for a given walk
is the initial opinion of the end node of the walk truncated at the first occurrence of a node from $S$.
The overall estimated opinion $\widehat{b}_{qv}^{(t)}[S]$ of $v$ is the average of $Y_{qv}^{(t)}[S]$ across all $\lambda_v$ walks from $v$.
\revise{The overall pseudocode is given in Algorithm \ref{alg:rw}.}
%
%
The above approach is clearly more efficient since it does not involve regenerating random walks for each seed set.
It also does not introduce any further error, since the estimates $Y_{qv}^{(t)}[S]$ satisfy the same property as $X_{qv}^{(t)}[S]$ in Theorem \ref{th:expect_direct}, as shown below.

\eat{
The seed set is generated greedily as in Algorithm~\ref{alg:whole}.
In Line~\ref{line:max_gain}, we select the best new seed based on the maximum \textit{estimated} marginal gain instead
of the actual marginal gain. In each iteration, given the initial random walks for
candidate $c_q$ and the previously selected seed set $S$, the algorithm computes
the marginal gain of including a candidate seed node $w$ into $S$ as follows:
If a random walk starting at node $v$ and ending at node $u$ contains $w$,
it will now be truncated at $w$. Since the initial opinion of $w$ will be set to
$1$ if it is a new seed, the marginal gain in $b_{qv}^{(t)}$ due to this random walk is
$\frac{1-b_{qu}^{(0)}}{\lambda_v}$.
We find all the random walks containing $w$ and compute the estimated marginal gain
of including $w$ into $S$. The best node is determined by the maximum estimated marginal gain.
All the walks containing this best node (which will be selected as the seed in the current round)
will be truncated at that node for the subsequent iterations, since its stubbornness and initial opinion will be set to 1. The estimation obtained by the truncation approach\footnote{This technique is named after {\em Post-Generation Truncation}. The formal definition can be found in our extended version \cite{extended}.} satisfies the unbiased property, and hence does not introduce any additional error, as shown below.

We now discuss the efficiency of the aforementioned truncation method and the rationale behind it.
Intuitively, when a node $w$ is selected as a seed node, its initial opinion and stubbornness values
are both set to $1$.
By definition, any random walk which reaches $w$ will terminate there. Thus, we can obtain random walks
according to the new stubbornness values by simply truncating the existing walks at $w$, instead of
regenerating random walks for the current iteration.
We now formally show that the estimate obtained by the truncation approach also satisfies the unbiasedness property, and hence does not introduce any additional error.
}
\vspace{-2mm}
\begin{theor}
\label{th:expect_post}
For any $t \geq 0$, any node $u$ and any seed set $S$, let $Y^{(t)}_{qu}[S]$ denote the estimated opinion of $u$ about $c_q$ at time $t$ by the {\em Post-Generation Truncation} approach, i.e., the initial opinion of the end node of the resultant random walk after initially sampling a $t$-step reverse random walk starting from $u$ without any seed and then truncating the walk at the first occurrence of a node in $S$. Then
\vspace{-2mm}
\begin{equation}
    \small \mathbbm{E} \left[ Y^{(t)}_{qu}[S] \right] = b^{(t)}_{qu}[S]
\end{equation}
\end{theor}
\vspace{-3mm}
\begin{proof} 
We prove by induction on the time horizon $t$. When $t = 0$, the walk consists only of $u$, and hence $Y^{(t)}_{qu}[S] = b^{(0)}_{qu}[S] = b^{(t)}_{qu}[S]$.
Now assume that the statement is true for time horizon $t$. Consider an execution (random walk generation followed by truncation) with time horizon $t+1$. If $u \in S$, the resultant walk (after truncation) will consist only of $u$ with initial opinion $1$, and hence $\mathbbm{E} \left[ Y^{(t+1)}_{qu}[S] \right] = 1 = b^{(t+1)}_{qu}[S]$. Otherwise, we have the following. With probability $d_{qu}$, $u$ is stubborn, i.e., the walk terminates at $u$ during generation; thus, the end node of the resultant walk will be $u$. With probability $1 - d_{qu}$, $u$ is not stubborn, i.e., during generation, the walk transitions to a random node $y$ (with probability $w_{qyu}$) from which a $t$-step walk is generated; thus, the end node of the resultant walk will be the end node of the truncated $t$-step walk from $y$. Since $u \notin S$, $d_{qu} = d_{qu}[S]$ and $b^{(0)}_{qu} = b^{(0)}_{qu}[S]$. Then
\vspace{-2mm}
\begin{small}
\begin{align*}
    &\mathbbm{E} \left[ Y^{(t+1)}_{qu}[S] \right] = \left( 1 - d_{qu} \right) \mathbbm{E} \left[ Y^{(t+1)}_{qu}[S] \Big| u \text{ was not stubborn} \right] \\
    &\qquad\qquad\qquad\qquad + d_{qu} \mathbbm{E} \left[ Y^{(t+1)}_{qu}[S] \Big| u \text{ was stubborn} \right] \\
    &= \left( 1 - d_{qu} \right) \sum_{y \in V} w_{qyu} \mathbbm{E} \left[ Y^{(t+1)}_{qu}[S] \Big| u \text{ transitioned to } y \right] + d_{qu} b^{(0)}_{qu} \\
    &= \left( 1 - d_{qu} \right) \sum_{y \in V} w_{qyu} \mathbbm{E} \left[ Y^{(t)}_{qy}[S] \right] + d_{qu} b^{(0)}_{qu} \\
    &= \left( 1 - d_{qu}[S] \right) \sum_{y \in V} w_{qyu} b^{(t)}_{qy}[S] + d_{qu}[S] b^{(0)}_{qu}[S] = b^{(t+1)}_{qu}[S]
\end{align*}
\end{small}
\end{proof}

\vspace{-0.4mm}
\spara{Time Complexity.} For the target candidate,
the generation of $t$-step reverse random walks starting from all nodes
takes $\bigO \left( t \sum_{v \in V} \lambda_v \right)$ time. First,
we analyze the time complexity of finding the top-$k$ seed nodes for the cumulative score
via random walk-based estimation. In each iteration of the greedy algorithm,
we estimate, for every node in all generated random walks, the new score
that would result if that node is added as a seed. Then, we add the node $w$
resulting in the largest estimated marginal gain as the new seed node.
Since the candidate seeds are only those nodes which are present in the walks, we can compute the marginal gains for all of them with one scan over all walks as follows. We initialize the marginal gain for each node to $0$. During the scan, when we encounter a node in a walk, we increase the marginal gain for that node by computing the increase in the estimated opinion of the start node $v$ of the walk, which is 1 minus the initial opinion of the end node of the walk, divided by $\lambda_v$.
This part scans all the random walks once and takes
$\bigO \left( t \sum_{v \in V} \lambda_v \right)$ time.
Next, all walks containing $w$ are truncated at $w$ for the subsequent iterations.
This step also takes $\bigO \left( t\sum_{v \in V} \lambda_v \right)$ time.
As the entire process is repeated $k$ times (to find the
top-$k$ seeds), the running time of the seed selection phase is
$\bigO \left( k t \sum_{v \in V} \lambda_v \right)$.

For the \revise{plurality} score variants and the \revise{Copeland} score, we additionally need to compute the exact opinion
values of each user about all other candidates at time $t$
via direct matrix-vector multiplication,
taking an additional $\bigO \left( (r-1)tm \right)$ time.
Thus, the overall time complexity for these scores is $\bigO\left(kt\sum_{v \in V} \lambda_v + (r-1)tm \right)$. Practically, thanks to the sparseness of the matrices, the dominant term is the first one due to the seed selection phase.
%
\vspace{-1mm}
\subsection{Accuracy Guarantees}
\label{sec:rw_quality}
The quality of the estimated opinions
depends on $\lambda_v$, i.e., the number of reverse random walks from $v$.

\vspace{-0.6mm}
\spara{Cumulative Score.}
The cumulative score aggregates the opinion values of all users about a
target candidate $c_q$. We provide a probabilistic accuracy guarantee about the estimated opinion.
\begin{theor}
\label{th:cumu_rws}
Given $\delta, \rho > 0$, if, for any node $v$, $\lambda_v$ satisfies
\begin{small}
\begin{equation}
\label{eq:num_walks}
    \lambda_v \geq \frac{1}{2\delta^2}\ln\left(\frac{2}{1-\rho}\right)
\end{equation}
\end{small}
then the following holds with probability at least $\rho$:
\begin{small}
\begin{equation}
    \left|\widehat{b}^{(t)}_{qv}[S] - b^{(t)}_{qv}[S]\right| < \delta
\end{equation}
\end{small}
\end{theor}
\begin{proof} 
As mentioned in \S~\ref{sec:random_walk_alg}, $\widehat{b}_{qv}^{(t)}[S]$ is the average of $Y_{qv}^{(t)}[S]$ across all $\lambda_v$ walks from $v$. From Theorem \ref{th:expect_post} and the linearity of expectation, $\mathbbm{E} \left[ \widehat{b}^{(t)}_{qv} [S] \right] = b^{(t)}_{qv}[S]$. From Hoeffding's inequality,
\begin{small}
\begin{equation*}
    \Pr \left( \left|\widehat{b}^{(t)}_{qv}[S] - b^{(t)}_{qv}[S]\right| < \delta \right) \geq 1 - 2 \exp \left( - 2\lambda_v\delta^2 \right) \geq \rho
\end{equation*}
\end{small}
\end{proof}

\vspace{-0.6mm}
\spara{\revise{Plurality} Score Variants}. As in \S~\ref{sec:b_best}, the following analysis is shown for the positional-$p$-approval score, and hence also works for special cases, e.g., the \revise{plurality} and $p$-approval scores. Each user contributes a value which is equal to the weight of the rank of $c_q$ in her preference ordering if the rank is at most $p$, and $0$ otherwise (Equation \ref{eq:weighted_p_rank}).
Theorem~\ref{th:bound_br_rw} ensures that, with a high probability, our approach correctly estimates this contributed value.
\begin{theor}
\label{th:bound_br_rw}
Given a user $v$ and a seed set $S$ for candidate $c_q$, let $\gamma_v[S] = \min_{c_p \in C \setminus \{c_q\}} \left| b^{(t)}_{pv} - b^{(t)}_{qv}[S] \right|$, $\lambda_v \geq \frac{1}{2\left(\gamma_v[S]\right)^2} \ln \left( \frac{2}{1 - \rho} \right)$. Assume $\gamma_v[S]\neq 0$.
%
Then, with probability at least $\rho$, the following holds:
\begin{scriptsize}
\begin{equation}
\label{eq:br_rw_acc}
    \omega \left[ \beta \left( \widehat{b}_{qv}^{(t)}[S] \right) \right] \cdot \mathbbm{1} \left[ \beta \left( \widehat{b}_{qv}^{(t)}[S] \right) \leq p \right] = \omega \left[ \beta \left( b_{qv}^{(t)}[S] \right) \right] \cdot \mathbbm{1} \left[ \beta \left( b_{qv}^{(t)}[S] \right) \leq p \right]
    \vspace{-1mm}
\end{equation}
\end{scriptsize}
\end{theor}
\begin{proof} 
To satisfy Eq. \ref{eq:br_rw_acc}, it suffices to ensure that the estimated position of candidate $c_q$ in the preference ranking of user $v$ is correct.
Clearly, this is true if the following hold.

    $\bullet$ $\forall c_x$ s.t. $b^{(t)}_{qv}[S]-b^{(t)}_{xv}> 0$ (i.e., $\gamma_v[S] \leq b^{(t)}_{qv}[S]-b^{(t)}_{xv}$):
        \begin{footnotesize}
        \begin{align*}
            & -\gamma_v[S] < \widehat{b}^{(t)}_{qv}[S] - b^{(t)}_{qv}[S] = \left(\widehat{b}^{(t)}_{qv}[S] - b^{(t)}_{xv}\right) - \left(b^{(t)}_{qv}[S] - b^{(t)}_{xv}\right) \\
            \implies& -\gamma_v[S] < \left(\widehat{b}^{(t)}_{qv}[S] - b^{(t)}_{xv}\right) - \gamma_v[S] \implies \widehat{b}^{(t)}_{qv}[S] - b^{(t)}_{xv} > 0 \\
            \implies & \text{The estimated ordering between } c_q \text{ and } c_x \text{ is correct}
        \end{align*}
        \end{footnotesize}
    $\bullet$ $\forall c_x$ s.t. $b^{(t)}_{xv} - b^{(t)}_{qv}[S]> 0$ (i.e., $\gamma_v[S] \leq b^{(t)}_{xv} - b^{(t)}_{qv}[S]$):
        \begin{footnotesize}
        \begin{align*}
            & -\gamma_v[S] < b^{(t)}_{qv}[S] - \widehat{b}^{(t)}_{qv}[S] = \left(b^{(t)}_{xv} - \widehat{b}^{(t)}_{qv}[S]\right) - \left(b^{(t)}_{xv} - b^{(t)}_{qv}[S]\right) \\
            \implies & -\gamma_v[S] < \left(b^{(t)}_{xv} - \widehat{b}^{(t)}_{qv}[S]\right) - \gamma_v[S] \implies b^{(t)}_{xv} - \widehat{b}^{(t)}_{qv}[S] > 0 \\
            \implies & \text{The estimated ordering between } c_q \text{ and } c_x \text{ is correct}
        \end{align*}
        \end{footnotesize}
From Hoeffding's inequality,
both the above points hold with probability at least $ 1 - 2 \exp \left( -2\lambda_v[\gamma_v[S]]^2 \right) \geq \rho$.
\end{proof}

In each iteration of Algorithm \ref{alg:whole}, the estimation of the opinion of user $v$ about $c_q$ involves an average over $\lambda_v$ random walks. However, the quantity $\gamma_v[S]$ in Theorem \ref{th:bound_br_rw} 
depends on the seed set $S$ for candidate $c_q$. For a given $S$, $\gamma_v[S]$ can be computed exactly via matrix-vector multiplication. But since $S$ differs from iteration to iteration (specifically, one node is added in each iteration),
a value of $\gamma_v[S]$ (and hence $\lambda_v$) that works well in one iteration may not work well in another iteration. As we generate random walks right in the beginning and reuse them for the subsequent iterations, 
a value of $\gamma_v[S]$ that works well in all iterations is
\begin{small}
\begin{equation}
\label{eq:gamma}
    \gamma_v^* = \min_{S \subseteq V \, : \, |S| \leq k} \gamma_v[S]
\end{equation}
\end{small}
However, efficiently computing the minimum over all seed sets $S$ of size at most $k$ is challenging. Thus, we estimate it heuristically using a greedy approach. Starting with $S = \emptyset$, we first estimate the opinion of user $v$ about $c_q$ by averaging over $\alpha$ random walks; $\alpha$ could, for example, be set to $\frac{1}{2\delta^2} \ln \left(\frac{2}{1-\rho}\right)$ in order to guarantee that, with probability at least $\rho$, each estimate differs from the true value by at most $\delta$. Once these estimates are found, we can estimate $\gamma_v[S]$ as $\widehat{\gamma_v}[S]$. After this, we repeatedly add to $S$ that node which minimizes the new $\widehat{\gamma_v}[S]$ computed using the newly estimated opinion values. The repetition stops once $|S| = k$ or there is no decrease in $\widehat{\gamma_v}[S]$, at which point we return $\widehat{\gamma_v}[S]$ as our estimate of $\gamma_v^*$.

\spara{\revise{Copeland} Score.} This score denotes the number of candidates whom the target candidate defeats in one-on-one competitions. Thus, we need the one-on-one winner to be estimated correctly (with a high probability) using the estimated opinion values.
%
\begin{theor}
\label{th:bound_con_rw}
Given a user $v$ and a seed set $S$ for candidate $c_q$, let $\gamma_v[S] = \min_{c_p \in C \setminus \{c_q\}} \left| b^{(t)}_{pv} - b^{(t)}_{qv}[S] \right|$. Suppose $\gamma_v[S] \neq 0$ and $\lambda_v \geq \frac{1}{2\left(\gamma_v[S]\right)^2} \ln \left( \frac{1}{1 - \rho} \right)$.
Then the following holds with probability at least $\rho$ for any $c_x \neq c_q$.
\begin{small}
\begin{equation}
\label{eq:con_rw_acc}
    \mathbbm{1} \left[ \widehat{b}^{(t)}_{qv}[S] > b^{(t)}_{xv} \right] = \mathbbm{1} \left[ b^{(t)}_{qv}[S] > b^{(t)}_{xv} \right]
\end{equation}
\end{small}
\end{theor}
\begin{proof} 
Assume, without loss of generality, that $b^{(t)}_{qv}[S] > b^{(t)}_{xv}$. Then, $\gamma_v[S] \leq b^{(t)}_{qv}[S] - b^{(t)}_{xv}$
by definition. We want to ensure that $c_q$ is correctly predicted to be ranked higher than $c_x$ for user $v$, i.e., Equation \ref{eq:con_rw_acc} is satisfied. Clearly, this is true if the following holds.
\begin{footnotesize}
\begin{align*}
    & -\gamma_v[S] < \widehat{b}^{(t)}_{qv}[S] - b^{(t)}_{qv}[S] = \left(\widehat{b}^{(t)}_{qv}[S] - b^{(t)}_{xv}\right) - \left(b^{(t)}_{qv}[S] - b^{(t)}_{xv}\right) \\
    \implies & -\gamma_v[S] < \left(\widehat{b}^{(t)}_{qv}[S] - b^{(t)}_{xv}\right) - \gamma_v[S] \implies \widehat{b}^{(t)}_{qv}[S] - b^{(t)}_{xv} > 0 \\
    \implies & \text{The estimated ordering between } c_q \text{ and } c_x \text{ is correct}
\end{align*}
\end{footnotesize}
From Hoeffding's inequality, this holds with probability at least $ 1 - \exp \left( -2\lambda_v[\gamma_v[S]]^2 \right) \geq \rho$.
\end{proof}
We estimate $\gamma_v[S]$ the same way as with the \revise{plurality} score.

\eat{
\spara{Remark.} To apply Theorems \ref{th:bound_br_rw} and \ref{th:bound_con_rw}, one needs to 
know whether $\gamma_v[S] \neq 0$ or not. 
%
To determine this efficiently in a probabilistic manner, we design a statistical hypothesis test which returns false with a high probability if $\gamma_v[S] = 0$.
This is stated formally below.
\begin{lem}
Given a seed set $S$ for $c_q$, suppose we generate $\lambda$ $t$-step reverse random walks starting from $v$.
We then compute $\widehat{b}^{(t)}_{qv}[S]$, which is  the average of the initial opinions (towards $c_q$) from the end nodes 
across all such walks, where
\begin{small}
\begin{equation}
  \lambda \geq \frac{1}{2\delta^2} \ln \left(\frac{2}{1-\rho}\right)
\end{equation}
\end{small}
Define $\widehat{\gamma_v}[S] = \min_{c_p \in C \setminus \{c_q\}} \left| b^{(t)}_{pv} - \widehat{b}^{(t)}_{qv}[S] \right|$. If $\gamma_v[S] = 0$, then $\widehat{\gamma_v}[S] < \delta$ with probability at least $\rho$.
\end{lem}
\begin{proof}
From Hoeffding's inequality, with probability at least $\rho$, $\left| b^{(t)}_{qv}[S] - \widehat{b}^{(t)}_{qv}[S] \right| < \delta$. If $\gamma_v[S] = 0$, there exists a candidate $c_x \neq c_q$ such that $b^{(t)}_{xv} = b^{(t)}_{qv}[S]$. This means, with probability at least $\rho$, $\left| b^{(t)}_{xv} - \widehat{b}^{(t)}_{qv}[S] \right| < \delta$, and hence $\widehat{\gamma_v}[S] < \delta$.
\end{proof}
}

%% file: 5-sketching.tex
\section{Sketch-based Estimation}
\label{sec:sketching}
Random walk-based approximation (\S~\ref{sec:random_walk}) requires the generation of reverse random walks starting from {\em all} nodes,
which could still be expensive. In this section, we further propose a {\em more efficient} reverse sketching-based {\em approximation} technique.
Notice that reverse sketching was used earlier in influence maximization (IM) \cite{BBCL14,TXS14,TSX15}.
\revise{We are the first to prove that the real-valued opinions in the FJ model can be estimated via reverse sketching and use it for opinion maximization. Moreover, our sketches (i.e., walks) are simpler and less memory consuming than the ones based on RR-sets (i.e., BFS trees), used in the classic IM.}
\vspace{-1mm}
\subsection{The Algorithmic Workflow}
\label{sec:alg_sketch}
We repeat the following $\theta$ times independently: Generate $\lambda_v$ $t$-step reverse random walks starting from a node $v$ chosen uniformly at random.
We refer to the set of generated walks as the sketch set. These sketches are similar to the tree-structured sketches used in the classic IM \cite{BBCL14,TXS14,TSX15} (see below for an intuition).
However, our sketches are walks, which are simpler and less memory consuming.
The opinions and the corresponding voting-based scores are estimated with the sketch set,
as detailed in \S~\ref{sec:sketch_quality}. The greedy seed selection workflow remains the same as in Algorithm~\ref{alg:whole}. The overall pseudocode is shown in Algorithm \ref{alg:sketch}.

\begin{algorithm}[tb!]
	\caption{\small Sketch-Based \texttt{Greedy} Seed Selection}
	\scriptsize
	\begin{algorithmic}[1]
		\REQUIRE Graph $\mathcal{G}=(V,E)$, initial opinion matrix $B^{(0)}$, influence matrix $W_i$ and stubbornness matrix $D_i$ for each candidate $c_i$, target candidate $c_q$, seed set size budget $k$, time horizon $t$, number of sketches $\theta$, and a scoring function $F$
		\ENSURE Seed set $S^*$ of size $k$
		\FOR {$j = 1$ \TO $\theta$}
		\STATE Choose a start node $v_j \in V$ uniformly at random
		\STATE Generate a $t$-step reverse random walk starting from $v_j$
		\ENDFOR
		\STATE $S^* \gets \emptyset$
		\FOR {$i = 1$ \TO $k$}
		\STATE $u \gets \argmax_{v \in V \setminus S^*} \left[ \widehat{F} \left( \widehat{B}^{(t)}[S^* \cup \{v\}], c_q \right) - \widehat{F} \left( \widehat{B}^{(t)}[S^*], c_q \right) \right]$
		\STATE $S^* \gets S^* \cup \{u\}$
		\ENDFOR
		\RETURN $S^*$
	\end{algorithmic}
	\label{alg:sketch}
\end{algorithm}


The reverse reachable (RR) sets in \cite{BBCL14,TXS14,TSX15}
are constructed by randomized BFS or DFS (sampling the incoming edges with their probabilities when we reach new nodes) from a start node, whose final status is decided by the initial statuses of the nodes in the set.
If an RR set contains a seed, the start node is said to be influenced and hence set to ``activated''.
This suggests that an RR set can alternatively be viewed as a directed tree rooted at the start node (resulting from the BFS or DFS). When a node is made a seed, the tree is truncated by removing all descendants of the seed, and then the ``activated'' status of the seed is ``pushed'' up to the start node.
In our method, we adopt a similar
technique: sampling one incoming edge when we reach new nodes in a walk, leading to a path. The final opinion of the start node of a walk is decided
by the initial opinion of the end node. If a walk contains a seed,
it is truncated at the seed, whose initial opinion (set to $1$) is ``pushed'' up to the start node.

\spara{Time Complexity.} The main difference between the sketching-based
estimation method (\S~\ref{sec:sketching}) and the random walk-based estimation method (\S~\ref{sec:random_walk})
is the total number of nodes from which we need to generate random walks.
Therefore, the running time of random walk generation is reduced to $\bigO \left( t \frac{\theta}{n} \sum_{v \in V} \lambda_v \right)$,
and the running time of the seed selection phase is reduced to
$\bigO \left( k t \frac{\theta}{n} \sum_{v \in V} \lambda_v \right)$.

For the \revise{plurality} and \revise{Copeland} scores, the computation of the opinion values of each user about all other candidates
takes an additional $\bigO \left( (r - 1) t m \right)$ time. Thus, for these scores, the overall time complexity is $\bigO \left( k t \frac{\theta}{n} \sum_{v \in V} \lambda_v + (r - 1) t m \right)$.
\subsection{Accuracy Guarantee for the Cumulative Score}
\label{sec:sketch_quality}
We discuss the number of sketches ($\theta$) required to
ensure that $\widehat{F} ( \widehat{B}^{(t)}[S], c_q )$ is a good estimate of $F ( B^{(t)}[S],c_q )$.
Let $v_j$ denote the $j^{th}$ sampled node, i.e., the start node of sketch $j$, where $j \in [1, \theta]$.

Denoting by $\widehat{b}^{(t)}_{qv_j}[S]$ the average of $Y_{qv_j}^{(t)}[S]$ (\S~\ref{sec:random_walk_alg}) across all $\lambda_{v_j}$ random walks from $v_j$, the estimated cumulative score is defined as:
\vspace{-2mm}
\begin{small}
\begin{equation}
    \widehat{F} \left( \widehat{B}^{(t)}[S],c_q \right) = \frac{n}{\theta} \sum_{j=1}^\theta \widehat{b}^{(t)}_{qv_j} [S]
\end{equation}
\end{small}
\vspace{-2mm}

Inspired by \cite{TSX15}, we aim to find a value of $\theta$ such that the true cumulative score for the seed set returned by Algorithm \ref{alg:sketch} is very close to the optimal score with a high probability. This is shown in Theorem \ref{th:theta}. Before proving this, we first show some useful lemma in this regard.

\begin{lem}
\label{lem:conc}
Let $p[S] = \frac{1}{n}F\left( B^{(t)}[S],c_q \right)$. The following inequalities hold for all values of $\theta$ and for all $\beta > 0$:
\begin{footnotesize}
\begin{align}
    \Pr \left( \sum_{i=1}^\theta \widehat{b}_{qv_i}^{(t)}[S] - \theta \cdot p[S] \geq \beta \cdot \theta \cdot p[S] \right) &\leq \exp \left( - \frac{\beta^2}{2 + \frac{2}{3}\beta} \cdot \theta \cdot p[S] \right) \label{eq:ut} \\
    \Pr \left( \sum_{i=1}^\theta \widehat{b}_{qv_i}^{(t)}[S] - \theta \cdot p[S] \leq - \beta \cdot \theta \cdot p[S] \right) &\leq \exp \left( - \frac{\beta^2}{2} \cdot \theta \cdot p[S] \right) \label{eq:lt}
\end{align}
\end{footnotesize}
\end{lem}
\begin{proof} 
For the independent random variables $\widehat{b}^{(t)}_{qv_j}[S], j \in [1, \theta]$, using Theorem \ref{th:expect_post},
\begin{small}
\begin{align*}
    \mathbbm{E} \left[ \widehat{b}^{(t)}_{qv_j} [S] \right] &= \sum_{v=1}^n \frac{1}{n} \mathbbm{E} \left[ \widehat{b}^{(t)}_{qv} [S] \right]
    = \frac{1}{n} \sum_{v=1}^n b^{(t)}_{qv} [S] \\
    &= \frac{1}{n} F\left( B^{(t)}[S],c_q \right) = p[S]
\end{align*}
\end{small}
Since $\widehat{b}^{(t)}_{qv_j} [S] \in [0, 1]$ and $F\left( B^{(t)}[S],c_q \right) \in [0, n]$,
\begin{small}
\begin{gather*}
    \widehat{b}^{(t)}_{qv_j} [S] - \mathbbm{E} \left[ \widehat{b}^{(t)}_{qv_j} [S] \right] \leq 1 \\
    \mathbbm{E} \left[ \widehat{b}^{(t)}_{qv_j} [S]^2 \right] \leq \mathbbm{E} \left[ \widehat{b}^{(t)}_{qv_j} [S] \right] = p[S] \\
    Var \left[ \widehat{b}^{(t)}_{qv_j} [S] \right] = \mathbbm{E} \left[ \widehat{b}^{(t)}_{qv_j} [S]^2 \right] - \left( \mathbbm{E} \left[ \widehat{b}^{(t)}_{qv_j} [S] \right] \right)^2 \leq p[S] \left( 1 - p[S] \right)
\end{gather*}
\end{small}
Following the concentration inequalities in Theorem \ref{th:conc} in Appendix \ref{sec:ineq}, we obtain Inequalities \ref{eq:ut} and \ref{eq:lt} as follows:
\begin{small}
\begin{align*}
    & \Pr \left( \sum_{i=1}^\theta \widehat{b}_{qv_i}^{(t)}[S] - \theta \cdot p[S] \geq \beta \cdot \theta \cdot p[S] \right) \\
    & \leq \exp \left( - \frac{(\beta \cdot \theta \cdot p[S])^2}{2 \left( \theta \cdot p[S] \left( 1 - p[S] \right) + \frac{\beta}{3} \cdot \theta \cdot p[S] \right)} \right) \\
    & \leq \exp \left( - \frac{\beta^2}{2 + \frac{2}{3}\beta} \cdot \theta \cdot p[S] \right) \\
    & \Pr \left( \sum_{i=1}^\theta \widehat{b}_{qv_i}^{(t)}[S] - \theta \cdot p[S] \leq - \beta \cdot \theta \cdot p[S] \right) \\
    & \leq \exp \left( - \frac{(\beta \cdot \theta \cdot p[S])^2}{2 \sum_{i=1}^\theta p[S]} \right) = \exp \left( - \frac{\beta^2}{2} \cdot \theta \cdot p[S] \right)
\end{align*}
\end{small}
\end{proof}

\begin{lem}
\label{lem:submod}
Given any node $u$ and any $t$-step reverse random walk from $u$ without any seed, let $Y^{(t)}_{qu}[S]$ denote the initial opinion of the end node of the resultant random walk after truncating the walk at the first occurrence of a seed node in $S$. Then $Y^{(t)}_{qu}[S]$ is submodular with respect to $S$.
\end{lem}
\begin{proof} 
Consider $P \subset Q \subset V$ and $s \in V \setminus Q$. It is easy to see that $Y^{(t)}_{qu}[S]$ is non-decreasing in $S$. We have the cases below.
\begin{itemize}
    \item $s$ does not belong to the truncated walk w.r.t. $P$. Then the same holds for $Q$ also, since $P \subset Q$. In that case,
    $Y^{(t)}_{qu}[P \cup \{s\}] = Y^{(t)}_{qu}[P]$ and $Y^{(t)}_{qu}[Q \cup \{s\}] = Y^{(t)}_{qu}[Q]$. Thus, 
    $Y^{(t)}_{qu}[P \cup \{s\}] - Y^{(t)}_{qu}[P] = 0 = Y^{(t)}_{qu}[Q \cup \{s\}] - Y^{(t)}_{qu}[Q]$.
    \item $s$ belongs to the truncated walk w.r.t. $P$ but not $Q$. Then $Y^{(t)}_{qu}[Q \cup \{s\}] - Y^{(t)}_{qu}[Q] = 0 \leq Y^{(t)}_{qu}[P \cup \{s\}] - Y^{(t)}_{qu}[P]$.
    \item $s$ belongs to the truncated walks w.r.t. both $P$ and $Q$. In that case, $Y^{(t)}_{qu}[P \cup \{s\}] - Y^{(t)}_{qu}[P] = 1 - Y^{(t)}_{qu}[P] \geq 1 - Y^{(t)}_{qu}[Q] = Y^{(t)}_{qu}[Q \cup \{s\}] - Y^{(t)}_{qu}[Q]$.
\end{itemize}
\end{proof}

For simplicity of notation, in what follows, we denote $F ( B^{(t)}[S], c_q )$ and $\widehat{F} ( \widehat{B}^{(t)}[S], c_q )$ by $F(S)$ and $\widehat{F}(S)$, respectively. We also use the following notations throughout the remainder of the section:
\begin{itemize}
    \item $OPT$: The maximum cumulative score for any size-$k$ seed set
    \item $S^o$: The size-$k$ seed set maximizing $F(S)$; this means $F \left(S^o\right) = OPT$.
\end{itemize}

\begin{lem}
\label{th:sketch}
Let $\delta_1 \in (0,1)$, $\epsilon_1 > 0$, and
\begin{equation}
    \small \theta_1 = \frac{2n}{OPT \cdot \epsilon_1^2} \cdot \ln \left( \frac{1}{\delta_1} \right) \label{eq:theta_1}
\end{equation}
If $\theta \geq \theta_1$, then $\widehat{F} \left( S^o \right) \geq \left( 1 - \epsilon_1 \right) \cdot OPT$ holds with probability at least $1 - \delta_1$. Formally,
\begin{equation*}
    \small \Pr \left( \widehat{F}(S^o) \geq (1 - \epsilon_1) \cdot OPT \right) \geq 1 - \delta_1
\end{equation*}
\end{lem}
\begin{proof} 
Let $p[S^o] = \frac{F(S^o)}{n}$. Using Inequality \ref{eq:lt}, we have
\begin{small}
\begin{align*}
    & \Pr \left( \widehat{F}(S^o) \leq (1 - \epsilon_1) \cdot OPT \right) \\
    &= \Pr \left( \frac{1}{\theta} \sum_{i=1}^\theta \widehat{b}_{qv_i}^{(t)}[S^o] \leq (1 - \epsilon_1) \cdot \frac{F(S^o)}{n} \right) \\
    &= \Pr \left( \sum_{i=1}^\theta \widehat{b}_{qv_i}^{(t)}[S^o] - \theta \cdot p[S^o] \leq - \epsilon_1 \cdot \theta \cdot p[S^o] \right) \\
    &\leq \exp \left( -\frac{\epsilon_1^2}{2} \cdot \theta \cdot p[S^o] \right) \leq \exp \left( -\frac{\epsilon_1^2}{2} \cdot \theta_1 \cdot \frac{F(S^o)}{n} \right) = \delta_1
\end{align*}
\end{small}
\end{proof}

\begin{lem}
\label{th:greedy}
Given $\epsilon > 0$, let $\delta_2 \in (0,1)$, $\epsilon_1 \in \left( 0, \frac{\epsilon}{1 - \frac{1}{e}} \right)$, and 
\begin{equation}
    \small \theta_2 = \frac{2n \left( 1 - \frac{1}{e} \right)}{OPT \left( \epsilon - \left(1 - \frac{1}{e} \right) \epsilon_1 \right)^2} \cdot \ln \left( \frac{\binom{n}{k}}{\delta_2} \right) \label{eq:theta_2}
\end{equation}
Let $S^*$ denote the size-$k$ seed set returned by Algorithm \ref{alg:sketch}. If $\theta \geq \theta_2$ and $\widehat{F} \left( S^o \right) \geq \left( 1 - \epsilon_1 \right) \cdot OPT$, then $F(S^*) \geq \left( 1 - \frac{1}{e} - \epsilon \right) \cdot OPT$ holds with probability at least $1 - \delta_2$. Formally,
\begin{footnotesize}
\begin{equation*}
    \Pr \left( F(S^*) \geq \left( 1 - \frac{1}{e} - \epsilon \right) \cdot OPT \bigg| \widehat{F}(S^o) \geq (1 - \epsilon_1) \cdot OPT \right) \geq 1 - \delta_2
\end{equation*}
\end{footnotesize}
\end{lem}
\begin{proof} 
It suffices to show that any size-$k$ seed set $S$ satisfying $F(S) < \left( 1 - \frac{1}{e} - \epsilon \right) \cdot OPT$ is returned by Algorithm \ref{alg:sketch} with probability at most $\frac{\delta_2}{\binom{n}{k}}$. In that case, by the union bound, there is at least $1 - \delta_2$ probability that no such set is returned.\\
Consider a given seed set $S$ satisfying $F(S) < \left( 1 - \frac{1}{e} - \epsilon \right) \cdot OPT$. Let $p[S] = \frac{F(S)}{n}$. Notice that $\widehat{b}^{(t)}_{qv_j}[S]$ is the average of $Y_{qv_j}^{(t)}[S]$ across all $\lambda_{v_j}$ random walks from $v_j$, and $\widehat{F}(S) = \frac{n}{\theta} \sum_{j=1}^\theta \widehat{b}^{(t)}_{qv_j}[S]$. Thus, Lemma \ref{lem:submod} implies that $\widehat{F}(S)$ is submodular w.r.t. $S$. Let $S^+$ be the size-$k$ seed set maximizing $\widehat{F}(S)$. If $S$ is returned by Algorithm \ref{alg:sketch}, from the submodularity of $\widehat{F}(\cdot)$ and the assumption that $\widehat{F} \left( S^o \right) \geq \left( 1 - \epsilon_1 \right) \cdot OPT$, we have
\begin{footnotesize}
\begin{equation*}
    \widehat{F}(S) \geq \left( 1 - \frac{1}{e} \right) \cdot \widehat{F}(S^+) \geq \left( 1 - \frac{1}{e} \right) \cdot \widehat{F}(S^o) \geq \left( 1 - \frac{1}{e} \right) \cdot (1 - \epsilon_1) \cdot OPT
\end{equation*}
\end{footnotesize}
Combining with $F(S) < \left( 1 - \frac{1}{e} - \epsilon \right) \cdot OPT$,
\begin{small}
\begin{align*}
    \widehat{F}(S) - F(S) &\geq \left( 1 - \frac{1}{e} \right) \cdot (1 - \epsilon_1) \cdot OPT - \left( 1 - \frac{1}{e} - \epsilon \right) \cdot OPT \\
    &= \left( \epsilon - \left( 1 - \frac{1}{e} \right) \epsilon_1 \right) \cdot OPT = \epsilon_2 \cdot OPT
\end{align*}
\end{small}
where $\epsilon_2 = \epsilon - \left( 1 - \frac{1}{e} \right) \epsilon_1 > 0$. This means the probability of $S$ being returned is at most the probability that $\widehat{F}(S) - F(S) \geq \epsilon_2 \cdot OPT$. Using this, Inequality \ref{eq:ut} and the fact that $n \cdot p[S] = F(S) < \left( 1 - \frac{1}{e} - \epsilon \right) \cdot OPT$, we have
\begin{footnotesize}
\begin{align*}
    \Pr (S \text{ is returned}) &\leq \Pr \left( \widehat{F}(S) - F(S) \geq \epsilon_2 \cdot OPT \right) \\
    &= \Pr \left( \sum_{i=1}^\theta \widehat{b}_{qv_i}^{(t)} - \theta \cdot p[S] \geq \frac{\epsilon_2 \cdot OPT}{n \cdot p[S]} \cdot \theta \cdot p[S] \right) \\
    &\leq \exp \left( -\frac{\epsilon_2^2 \cdot OPT^2}{2n^2 \cdot p[S] + \frac{2}{3}\epsilon_2 n \cdot OPT} \cdot \theta \right) \\
    &< \exp \left( -\frac{\epsilon_2^2 \cdot OPT^2}{2n \left( 1 - \frac{1}{e} - \epsilon \right) \cdot OPT + \frac{2}{3}\epsilon_2 n \cdot OPT} \cdot \theta \right) \\
    &< \exp \left( -\frac{\left( \epsilon - \left( 1 - \frac{1}{e} \right) \epsilon_1 \right)^2 \cdot OPT}{2n \left( 1 - \frac{1}{e} \right)} \cdot \theta \right) \\
    &\leq \exp \left( -\frac{\left( \epsilon - \left( 1 - \frac{1}{e} \right) \epsilon_1 \right)^2 \cdot OPT}{2n \left( 1 - \frac{1}{e} \right)} \cdot \theta_2 \right) = \frac{\delta_2}{\binom{n}{k}}
\end{align*}
\end{footnotesize}
\end{proof}

\begin{theor}
\label{th:theta}
Given any $\epsilon, l > 0$, setting
\begin{footnotesize}
\begin{equation}
    \theta \geq \frac{2n}{OPT \cdot \epsilon^2}\left[ \left(1 - \frac{1}{e}\right) \sqrt{\ln \left( 2n^l \right)} + \sqrt{\left(1 - \frac{1}{e}\right) \left[ \ln \left( 2n^l \right) + \ln \binom{n}{k} \right]} \right]^2 \label{eq:theta_cum}
\end{equation}
\end{footnotesize}
ensures that Algorithm \ref{alg:sketch} returns a $(1-1/e-\epsilon)$-approximate solution $S^*$ with probability at least $1 - n^{-l}$. More formally,
\begin{small}
\begin{equation}
    \Pr \left( F \left( S^* \right) \geq \left(1-\frac{1}{e}-\epsilon\right)OPT \right) \geq 1 - \frac{1}{n^{l}}
\end{equation}
\end{small}
\end{theor}
\begin{proof} 
Define $\delta_1 = \delta_2 = 0.5 \times n^{-l}$ and
\begin{footnotesize}
\begin{equation*}
    \epsilon_1 = \epsilon \cdot \frac{\sqrt{l \ln n + \ln 2}}{\left(1 - \frac{1}{e} \right) \sqrt{l \ln n + \ln 2} + \sqrt{\left( 1 - \frac{1}{e} \right) \left( l \ln n + \ln \binom{n}{k} + \ln 2 \right)}}
\end{equation*}
\end{footnotesize}
It is easy to see that $\delta_1, \delta_2, \epsilon_1$ satisfy the conditions in Lemma \ref{th:sketch} and \ref{th:greedy}, and that $\theta_1$ (Equation \ref{eq:theta_1}) and $\theta_2$ (Equation \ref{eq:theta_2}) are both the same and given by the RHS of Equation \ref{eq:theta_cum}. Thus, if $\theta$ satisfies Equation \ref{eq:theta_cum}, i.e., $\theta \geq \theta_1$ and $\theta \geq \theta_2$, combining Lemmas \ref{th:sketch} and \ref{th:greedy},
\begin{small}
\begin{align*}
    &\Pr \left( F(S^*) \geq \left( 1 - \frac{1}{e} - \epsilon \right) \cdot OPT \right) \\
    &\geq \Pr \left( F(S^*) \geq \left( 1 - \frac{1}{e} - \epsilon \right) \cdot OPT \bigg| \widehat{F}(S^o) \geq (1 - \epsilon_1) \cdot OPT \right) \\
    &\qquad \Pr \left( \widehat{F}(S^o) \geq (1 - \epsilon_1) \cdot OPT \right) \\
    &\geq (1 - \delta_1) (1 - \delta_2) > 1 - (\delta_1 + \delta_2) = 1 - n^{-l}
\end{align*}
\end{small}
\end{proof}

Since the above results hold for any value of $\lambda_v$, we set $\lambda_v = 1 \; \forall v \in V$.\footnote{\scriptsize{Although $\lambda_v = 1$ could
result in a very inaccurate estimate $\widehat{b}^{(t)}_{qv_j}$, we sample $\theta$ start nodes uniformly at random, all of which
need not be distinct; thus, it is still likely that the number of walks from a particular start node is more than 1. By
ensuring that $\theta$ is large enough, our overall cumulative score estimate is very
accurate with a high probability.}} 
In order to estimate a lower bound on $OPT$ in Equation \ref{eq:theta_cum}, we design a statistical hypothesis test which, on an input $x$, returns false with a high probability if $OPT < x$. Since $OPT \in [k, n]$, we can easily identify a lower bound on $OPT$ by running the test for $x \in \left\{ \frac{n}{2}, \frac{n}{4}, \frac{n}{8}, \ldots, k \right\}$. Such a test is provided in {\em Algorithm 2} in \cite{TSX15}.

\subsection{Accuracy Guarantees for the \revise{Plurality} Score Variants}

As in \S~\ref{sec:b_best}, the following analysis is shown for the positional-$p$-approval score, and hence also works for special cases, e.g., the \revise{plurality} and $p$-approval scores.

\smallskip

The estimated positional-$p$-approval score is defined as:
\begin{footnotesize}
\begin{equation}
    \widehat{F} \left( \widehat{B}^{(t)}[S],c_q \right) = \frac{n}{\theta} \sum_{j=1}^\theta \omega \left[ \beta \left( \widehat{b}^{(t)}_{qv_j} [S] \right) \right] \times \mathbbm{1} \left[ \beta \left( \widehat{b}^{(t)}_{qv_j} [S] \right) \leq p \right] \label{eq:br}
\end{equation}
\end{footnotesize}

Let $OPT$ denote the maximum positional-$p$-approval score for any size-$k$ seed set of $c_q$. We have:

\begin{lem}
\label{lem:estimate_br}
Suppose, for any sampled start node $v_j$, $j \in [1,\theta]$, the following holds with probability at least $\rho$ (Theorem \ref{th:bound_br_rw}):
\begin{footnotesize}
\begin{align}
    &\omega \left[ \beta \left( \widehat{b}_{qv_j}^{(t)}[S] \right) \right] \times \mathbbm{1} \left[ \beta \left( \widehat{b}_{qv_j}^{(t)}[S] \right) \leq p \right] \nonumber\\
    &= \omega \left[ \beta \left( b_{qv_j}^{(t)}[S] \right) \right] \times \mathbbm{1} \left[ \beta \left( b_{qv_j}^{(t)}[S] \right) \leq p \right] \label{eq:rho_br}
\end{align}
\end{footnotesize}
If $\theta$ satisfies
\begin{small}
\begin{equation}
    \rho^{\theta} \left[ 1 - 2 \exp \left(-\frac{\epsilon^2 \cdot OPT}{(8 + 2\epsilon) n} \cdot \theta\right) \right] \geq 1 - \binom{n}{k}^{-1}n^{-l} \label{eq:theta_br}
\end{equation}
\end{small}
Then, for any size-$k$ seed set $S$ for candidate $c_q$, the following holds with probability at least $1 - \binom{n}{k}^{-1} n^{-l}$.
\begin{small}
\begin{equation}
    \left| \widehat{F} \left( \widehat{B}^{(t)}[S], c_q \right) - F\left( B^{(t)}[S],c_q \right) \right| < \frac{\epsilon}{2} \cdot OPT \label{eq:error_br}
\end{equation}
\end{small}
\end{lem}
\begin{proof} 
From Equations \ref{eq:br} and \ref{eq:rho_br}, the following holds with probability at least $\rho^\theta$:
\begin{small}
\begin{equation*}
    \widehat{F} \left( \widehat{B}^{(t)}[S],c_q \right) = \frac{n}{\theta} \sum_{j=1}^\theta \omega \left[ \beta \left( b^{(t)}_{qv_j} [S] \right) \right] \times \mathbbm{1} \left[ \beta \left( b^{(t)}_{qv_j} [S] \right) \leq p \right]
\end{equation*}
\end{small}
Also, for all $j \in [1, \theta]$, we have:
\begin{footnotesize}
\begin{align*}
    &\mathbbm{E} \left[ \omega \left[ \beta \left( b^{(t)}_{qv_j} [S] \right) \right] \times \mathbbm{1} \left[ \beta \left( b^{(t)}_{qv_j} [S] \right) \leq p \right] \right] \\
    &= \sum_{v \in V} \frac{1}{n} \times \omega \left[ \beta \left( b^{(t)}_{qv} [S] \right) \right] \times \mathbbm{1} \left[ \beta \left( b^{(t)}_{qv} [S] \right) \leq p \right] \\
    &= \frac{1}{n} \cdot F\left( B^{(t)}[S],c_q \right)
\end{align*}
\end{footnotesize}
Since $\omega[i] \in [0, 1] \, \forall i \in [1, r]$, using the concentration inequality in Theorem \ref{th:chernoff} in Appendix \ref{sec:ineq} and Inequality \ref{eq:theta_br}, we obtain:
\begin{footnotesize}
\begin{align*}
    &\Pr \left( \left| \widehat{F} \left( \widehat{B}^{(t)}[S], c_q \right) - F\left( B^{(t)}[S],c_q \right) \right| < \frac{\epsilon}{2} \cdot OPT \right) \\
    &\geq \rho^\theta \Pr \left( \left| \sum_{j=1}^\theta \omega \left[ \beta \left( b^{(t)}_{qv_j} [S] \right) \right] \times \mathbbm{1} \left[ \beta \left( b^{(t)}_{qv_j} [S] \right) \leq p \right] \right. \right. \\
    &\left. \left. \quad - \frac{\theta}{n} \cdot F\left( B^{(t)}[S],c_q \right) \right| < \frac{\epsilon}{2} \cdot \frac{OPT}{F\left( B^{(t)}[S],c_q \right)} \cdot \frac{\theta}{n} \cdot F\left( B^{(t)}[S],c_q \right) \right) \\
    &\geq \rho^\theta \left[ 1 - 2 \exp \left(-\frac{\epsilon^2 \cdot OPT^2}{8 \cdot F\left( B^{(t)}[S],c_q \right) + 2\epsilon  \cdot OPT} \cdot \frac{\theta}{n} \right) \right] \\
    &\geq \rho^{\theta} \left[ 1 - 2 \exp \left(-\frac{\epsilon^2 \cdot OPT}{(8 + 2\epsilon) n} \cdot \theta\right) \right] \geq 1 - \binom{n}{k}^{-1}n^{-l}
\end{align*}
\end{footnotesize}
\end{proof}

Notice that, given a value of $\rho$, the LHS of Equation \ref{eq:theta_br} is not a monotonic function of $\theta$. Instead, it increases up to
a certain value of $\theta$
and then decreases. Thus, if the RHS of Equation \ref{eq:theta_br} is greater than the maximum value of the LHS, there is no satisfying $\theta$. Otherwise, there are two satisfying values of $\theta$, and we can choose the smaller one. To find those values, we plot the LHS of Equation \ref{eq:theta_br} as a function of $\theta$. This process is illustrated in Figure \ref{fig:br_plot}.

\begin{figure}[tb!]
    \centering
    \includegraphics[scale=0.34]{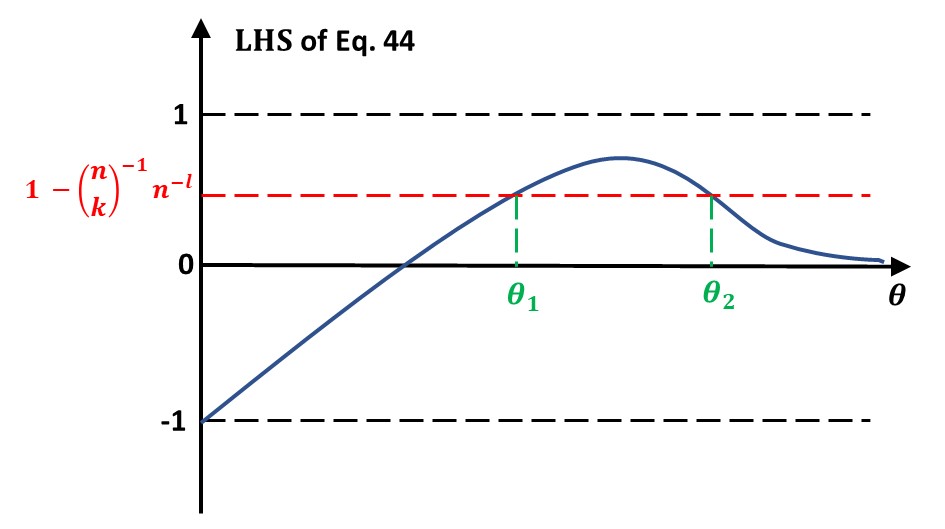}
    \caption{\small{Plot of the LHS of Equation \ref{eq:theta_br} as a function of $\theta$. We choose $\theta = \theta_1$, which is the smallest of all admissible values of $\theta$.}}
    \label{fig:br_plot}
\end{figure}

Suppose sandwich approximation ensures that the greedy algorithm returns a $\zeta$-approximate solution to maximizing the estimated positional-$p$-approval score. In that case, we have the following theorem.

\begin{theor}
\label{th:approx_ratio}
If $\theta$ satisfies Inequality \ref{eq:theta_br}, our algorithm returns a $(\zeta - \epsilon)$-approximate solution $S^*$ with probability at least $1 - n^{-l}$.
\begin{small}
\begin{equation}
    \Pr \left( F\left( B^{(t)}[S^*],c_q \right) > \left( \zeta - \epsilon \right)OPT \right) \geq 1 - n^{-l}
\end{equation}
\end{small}
\end{theor}
\begin{proof} 
Let $S^*$ be the seed set returned by our algorithm, and $S^+ = \argmax_{S \subseteq V, |S| = k} \widehat{F}( \widehat{B}^{(t)}[S],c_q )$.
Let $S^o$ denote the optimal solution to our Problem \ref{prob:fjvote}.
Assume that $\theta$ satisfies Inequality \ref{eq:theta_br}. By Lemma \ref{lem:estimate_br}, Equation \ref{eq:error_br} holds with probability at least $1 - \binom{n}{k}^{-1}n^{-l}$ for any size-$k$ seed set $S$. Then, by the union bound, Equation \ref{eq:error_br} should hold simultaneously for all size-$k$ seed sets with probability at least $1-n^{-l}$. In that case, we have
\begin{small}
\begin{align*}
    &F\left( B^{(t)}[S^*],c_q \right) \\
    &> \widehat{F} \left( \widehat{B}^{(t)}[S^*], c_q \right) - \frac{\epsilon}{2} \cdot OPT \qquad\qquad\qquad\qquad\quad (\text{Equation \ref{eq:error_br}}) \\
    &\geq \zeta \cdot \widehat{F} \left( \widehat{B}^{(t)}[S^+], c_q \right) - \frac{\epsilon}{2} \cdot OPT \qquad (\text{Sandwich Approximation}) \\
    &\geq \zeta \cdot \widehat{F} \left( \widehat{B}^{(t)}[S^o], c_q \right) - \frac{\epsilon}{2} \cdot OPT \qquad\qquad\quad (\text{Definition of }S^+) \\
    &> \zeta \cdot \left(F\left( B^{(t)}[S^o], c_q \right) - \frac{\epsilon}{2} \cdot OPT\right) - \frac{\epsilon}{2} \cdot OPT \quad (\text{Equation \ref{eq:error_br}}) \\
    &= \zeta \cdot OPT - \frac{1 + \zeta}{2} \cdot \epsilon \cdot OPT \qquad\quad \left( F \left( B^{(t)}[S^o],c_q \right) = OPT \right) \\
    &> \left( \zeta - \epsilon \right) OPT \qquad\qquad\qquad \qquad\qquad\qquad\qquad\qquad\qquad (\zeta < 1)
\end{align*}
\end{small}
\end{proof}

$OPT$ is estimated in a similar way as in \S~\ref{sec:sketch_quality}.

\subsection{Accuracy Guarantee for the \revise{Copeland} Score}
\label{sec:sketch_con}

The relation $\succ_{\widehat{M}}$ is defined as: $c_q \succ_{\widehat{M}} c_x$ if, among the $\theta$ samples, more users $v_j$ satisfy $\widehat{b}^{(t)}_{qv_j} [S] > b^{(t)}_{xv_j}$ than the other way round.
The estimated \revise{Copeland} score is then defined as
\begin{small}
\begin{align}
    &\widehat{F} \left( \widehat{B}^{(t)}[S], c_q \right) = \sum_{c_x \in C \setminus \{c_q\}} \mathbbm{1} \left[ c_q \succ_{\widehat{M}} c_x \right] \nonumber \\
    &= \sum_{c_x \in C \setminus \{c_q\}} \mathbbm{1} \left[ \sum_{v \in V} \mathbbm{1} \left[ \widehat{b}^{(t)}_{qv} > b^{(t)}_{xv} \right] > \sum_{v \in V} \mathbbm{1} \left[ \widehat{b}^{(t)}_{qv} < b^{(t)}_{xv} \right] \right]
\end{align}
\end{small}

Following Theorem \ref{th:bound_con_rw}, the preference between $c_q$ and any other candidate $c_x$ can be estimated correctly for any user with probability at least $\rho$. We define $\mu[S]$ as the minimum (across all candidates $c_x$) of the difference between the fraction of users who rank $c_q$ above $c_x$ and those who do not. More formally,
\begin{small}
    \begin{equation*}
        \mu[S] = \min_{c_x \in C \setminus \{c_q\}} \frac{1}{n} \left| \sum_{v=1}^n \mathbbm{1} \left[ b^{(t)}_{qv}[S] > b^{(t)}_{xv} \right] - \sum_{v=1}^n \mathbbm{1} \left[ b^{(t)}_{qv}[S] < b^{(t)}_{xv} \right] \right|
    \end{equation*}
\end{small}

\begin{lem}
\label{lem:estimate_con}
Assume $\gamma_v[S] = \min_{c_p \in C \setminus \{c_q\}} \left| b^{(t)}_{pv} - b^{(t)}_{qv}[S] \right| \neq 0 \; \forall v \in V$. Suppose, for any $c_x \neq c_q$ and any sampled start node $v_j$, $j \in [1,\theta]$, the following holds with probability at least $\rho$ (Theorem \ref{th:bound_con_rw}), i.e.,
\begin{small}
\begin{equation*}
    \Pr \left( \mathbbm{1} \left[ \widehat{b}^{(t)}_{qv}[S] > b^{(t)}_{xv} \right] = \mathbbm{1} \left[ b^{(t)}_{qv}[S] > b^{(t)}_{xv} \right] \right) \geq \rho
\end{equation*}
\end{small}
If $\theta$ satisfies
\begin{small}
\begin{equation}
    \rho^\theta \left[ 1 - \left( 1 - [\mu[S]]^2 \right)^{\frac{\theta}{2}} \right] \geq 1 - \binom{n}{k}^{-1} n^{-l} (r - 1)^{-1} \label{eq:theta_con}
\end{equation}
\end{small}
Then, for any size-$k$ seed set $S$ of candidate $c_q$, the following holds with probability at least $1 - \binom{n}{k}^{-1} n^{-l}$.
\begin{small}
\begin{equation}
    \widehat{F} \left( \widehat{B}^{(t)}[S],c_q \right) = F\left( B^{(t)}[S],c_q \right) \label{eq:error_con}
\end{equation}
\end{small}
\end{lem}
\begin{proof}
We derive:
\begin{small}
\begin{align}
    &\Pr \left( \widehat{F} \left( \widehat{B}^{(t)}[S],c_q \right) = F\left( B^{(t)}[S],c_q \right) \right) \nonumber \\
    &\geq \Pr \left( \forall c_x \in C \setminus \{c_q\} : \mathbbm{1} \left[ c_q \succ_{\widehat{M}} c_x \right] = \mathbbm{1} \left[ c_q \succ_M c_x \right] \right) \nonumber \\
    &\geq 1 - \sum_{c_x \in C \setminus \{c_q\}} \Pr \left( \mathbbm{1} \left[ c_q \succ_{\widehat{M}} c_x \right] \neq \mathbbm{1} \left[ c_q \succ_M c_x \right] \right) \label{eq:con_correct}
\end{align}
\end{small}
For any candidate $c_x \neq c_q$, assume (without loss of generality) that $c_q \succ_M c_x$, i.e. $\mathbbm{1} \left[ c_q \succ_M c_x \right] = 1$. Define the following:
\begin{align*}
    V_1 &= \left\{ v \in V : b^{(t)}_{qv}[S] > b^{(t)}_{xv} \right\} \\
    V_2 &= \left\{ v \in V : b^{(t)}_{qv}[S] < b^{(t)}_{xv} \right\}
\end{align*}
Let us define:
\begin{footnotesize}
\begin{equation*}
    \mu_x[S] = \frac{|V_1| - |V_2|}{n} = \frac{1}{n} \left[ \sum_{v=1}^n \left[ \mathbbm{1} \left[ b^{(t)}_{qv}[S] > b^{(t)}_{xv} \right] - \mathbbm{1} \left[ b^{(t)}_{qv}[S] < b^{(t)}_{xv} \right] \right] \right] 
\end{equation*}
\end{footnotesize}
Since $\gamma_v[S] \neq 0 \; \forall v \in V$, we have $|V_1| + |V_2| = n$. Also, by definition, $|V_1| - |V_2| = n \cdot \mu_x[S]$. Thus, we obtain $|V_1| = \frac{n}{2} \left[ 1 + \mu_x[S] \right]$ and $|V_2| = \frac{n}{2} \left[ 1 - \mu_x[S] \right]$. \\
Now, for $j \in [1, \theta]$, define $Z_j = \mathbbm{1} \left[ v_j \in V_2 \right] = \mathbbm{1} \left[ b^{(t)}_{qv_j}[S] < b^{(t)}_{xv_j} \right]$. This means
\begin{footnotesize}
\begin{align}
    &\Pr \left( \mathbbm{1} \left[ c_q \succ_{\widehat{M}} c_x \right] = \mathbbm{1} \left[ c_q \succ_M c_x \right] \right) = \Pr \left( \mathbbm{1} \left[ c_q \succ_{\widehat{M}} c_x \right] = 1 \right) \nonumber \\
    &= \Pr \left( \sum_{j=1}^\theta \mathbbm{1} \left[ \widehat{b}^{(t)}_{qv_j}[S] < b^{(t)}_{xv_j} \right] < \frac{\theta}{2} \right)\geq \rho^\theta \left[ 1 - \Pr \left( \sum_{j=1}^\theta Z_j \geq \frac{\theta}{2} \right) \right] \label{eq:con_cq}
\end{align}
\end{footnotesize}
Clearly $\mathbbm{E} \left[Z_j\right] = \Pr \left( Z_j = 1 \right) = \frac{|V_2|}{|V|} = \frac{1 - \mu_x[S]}{2} \; \forall j$. Using the concentration inequality in Theorem \ref{th:hoeffding} in Appendix \ref{sec:ineq}, and noticing that $\mu_x[S] \geq \mu[S]$, we have
\begin{small}
\begin{align*}
    \Pr \left( \sum_{j=1}^\theta Z_j \geq \frac{\theta}{2} \right) &\leq \left( \left( 1 - \mu_x[S] \right)^{\frac{1}{2}} \left( 1 + \mu_x[S] \right)^{\frac{1}{2}} \right)^\theta \\
    &= \left( 1 - [\mu_x[S]]^2 \right)^{\frac{\theta}{2}} \leq \left( 1 - [\mu[S]]^2 \right)^{\frac{\theta}{2}}
\end{align*}
\end{small}
Substituting the above into Inequality~\ref{eq:con_cq}, we have:
\begin{small}
\begin{align*}
    \Pr \left( \mathbbm{1} \left[ c_q \succ_{\widehat{M}} c_x \right] = \mathbbm{1} \left[ c_q \succ_M c_x \right] \right) 
    \geq \rho^\theta \left[ 1 - \left( 1 - [\mu[S]]^2 \right)^{\frac{\theta}{2}} \right]
\end{align*}
\end{small}
Substituting into Inequality~\ref{eq:con_correct} and using Inequality \ref{eq:theta_con},
\begin{small}
\begin{align*}
    &\Pr \left( \widehat{F} \left( \widehat{B}^{(t)}[S],c_q \right) = F\left( B^{(t)}[S],c_q \right) \right) \\
    &\geq 1 - \sum_{c_x \in C \setminus \{c_q\}} \Pr \left( \mathbbm{1} \left[ c_q \succ_{\widehat{M}} c_x \right] \neq \mathbbm{1} \left[ c_q \succ_M c_x \right] \right) \\
    &\geq 1 - (r - 1) \left[ 1 - \rho^\theta \left[ 1 - \left( 1 - [\mu[S]]^2 \right)^{\frac{\theta}{2}} \right] \right] \geq 1 - \binom{n}{k}^{-1} n^{-l}
\end{align*}
\end{small}
Note that even though we proved for the case when $c_q \succ_M c_x$, we can prove for the case when $c_x \succ_M c_q$ by simply reversing the definitions of $V_1$ and $V_2$ above and following similar steps. Since these two are mutually exclusive and exhaustive cases, the result is proved in general.
\end{proof}

To compute $\theta$ using Equation \ref{eq:theta_con}, we use a similar plotting method as with the plurality score variants (Figure \ref{fig:br_plot}).

Suppose sandwich approximation ensures that the greedy algorithm returns a $\zeta$-approximate solution to maximizing the estimated \revise{Copeland} score. In that case, denoting by $OPT$ the maximum \revise{Copeland} score for any size-$k$ seed set of $c_q$, we have the following.

\begin{theor}
If $\theta$ follows Lemma~\ref{lem:estimate_con}, our algorithm returns a $\zeta$-approximate solution $S^*$ with probability at least $1 - n^{-l}$.
\begin{small}
\begin{equation}
    F \left( B^{(t)}[S^*],c_q \right) \geq \zeta \cdot OPT
\end{equation}
\end{small}
\end{theor}
\begin{proof}
Let $S^*$ be the seed set returned by our algorithm, and $S^+ = \argmax_{S \subseteq V, |S| = k} \widehat{F}( \widehat{B}^{(t)}[S],c_q )$.
Let $S^o$ denote the optimal solution to our Problem \ref{prob:fjvote}.
Assume that $\theta$ satisfies Inequality \ref{eq:theta_con}. By Lemma \ref{lem:estimate_con}, Equation \ref{eq:error_con} holds with probability at least $1 - \binom{n}{k}^{-1}n^{-l}$ for any size-$k$ seed set $S$. Then, by the union bound, Equation \ref{eq:error_con} should hold simultaneously for all size-$k$ seed sets with probability at least $1-n^{-l}$. In that case, we have
\begin{small}
\begin{align*}
    &F\left( B^{(t)}[S^*],c_q \right) \\
    &> \widehat{F} \left( \widehat{B}^{(t)}[S^*], c_q \right) - \frac{\epsilon}{2} \cdot OPT \qquad\qquad\qquad\qquad\quad (\text{Equation \ref{eq:error_con}}) \\
    &\geq \zeta \cdot \widehat{F} \left( \widehat{B}^{(t)}[S^+], c_q \right) - \frac{\epsilon}{2} \cdot OPT \qquad (\text{Sandwich Approximation}) \\
    &\geq \zeta \cdot \widehat{F} \left( \widehat{B}^{(t)}[S^o], c_q \right) - \frac{\epsilon}{2} \cdot OPT \qquad\qquad\quad (\text{Definition of }S^+) \\
    &> \zeta \cdot \left(F\left( B^{(t)}[S^o], c_q \right) - \frac{\epsilon}{2} \cdot OPT\right) - \frac{\epsilon}{2} \cdot OPT \quad (\text{Equation \ref{eq:error_con}}) \\
    &= \zeta \cdot OPT - \frac{1 + \zeta}{2} \cdot \epsilon \cdot OPT \qquad\quad \left( F \left( B^{(t)}[S^o],c_q \right) = OPT \right) \\
    &> \left( \zeta - \epsilon \right) OPT \qquad\qquad\qquad\qquad \qquad\qquad\qquad\qquad\qquad (\zeta < 1)
\end{align*}
\end{small}
\end{proof}

Note that the LHS of Equation \ref{eq:theta_con} requires the value of $\mu[S]$, which is not monotonic with $S$. A value which works for all seed sets is given by
\begin{small}
\begin{equation}
    \mu^* = \min_{S \subseteq V : |S| \leq k} \mu[S]
\end{equation}
\end{small}
This value can be estimated akin to $\gamma^*$ in Equation \ref{eq:gamma}.

\vspace{-1mm}
\subsection{Heuristic Estimation of $\theta$ for the \revise{Plurality} Score Variants and the \revise{Copeland} Score}
\label{sec:theta_exp}


While theoretical bounds on $\theta$ for the \revise{plurality} score variants and the \revise{Copeland} score can be derived as above,
we find them to be not so effective: {\bf (1)} From the inequalities obtained in the theoretical guarantees, it is difficult to compute a closed-form expression for $\theta$; 
{\bf (2)} The sandwich approximation factor is smaller than $(1-1/e)$ (\S~\ref{sec:practical_effectiveness}); coupled with the approximation via sketches, the overall approximation factor is even smaller.
Instead, we use a heuristic method to compute the optimal value of $\theta$.
Note that our sketch-based method is more efficient than our
random walk-based approach only when $\theta < n$. For a given dataset
and score, we empirically find the smallest $\theta$ when that score
converges (for some $k$ and $t$). This one-time estimate of $\theta$
can be re-used on the same dataset and score, even with different number of seeds ($k$)
and time horizon ($t$) as inputs, since we find such an estimate to be less sensitive
to $k$ and $t$. In \S~\ref{sec:sensitivity}, we demonstrate that the above mentioned heuristic estimation
of $\theta$ produces good-quality results.


%% file: 6-related.tex
\vspace{-1mm}
\section{Related Works}
\label{sec:related}
\vspace{-1mm}
%
%
\spara{Opinion Manipulation.} \cite{DDMH17, PB15, BC18} consider
network modification
to enable (or prevent) opinion consensus (or convergence).
\cite{GMS20} proposes strategies for manipulating users' opinions with the voter model.
Opinion maximization with the voter model is considered in \cite{MMLB20,KK19,Li0WZ13}.
Conformity, an opposite notion of stubbornness (used in the FJ model),
measures the likelihood of a user adopting the opinions of her neighbors. Conformity-based
opinion maximization has been studied in \cite{DasGKL14,0005BSC15}, albeit in a {\em single-campaign setting}. \cite{GionisTT13, AKPT18} study seed selection for opinion maximization in a single-campaign and
without a given finite time horizon
 (details in Appendices \ref{sec:compare_om} and \ref{sec:time_horizon}).
To the best of our knowledge, {\bf (a)} {\em we are the first to bridge two different disciplines:
{\bf (1)} seed selection for opinion maximization at a finite time horizon and {\bf (2)} voting-based winning criteria with multiple campaigners}.
Moreover, {\bf (b)} {\em we are the first to design random walk and sketch-based efficient algorithms, with theoretical guarantees,
for DeGroot and FJ model-based opinion maximization}.

Recall that the cumulative score, due to its aggregate nature, is independent of the other campaigns; thus
it is similar to opinion maximization in a single-campaigner setting \cite{GionisTT13}. Hence, the greedy algorithm
in \cite{GionisTT13}, with proper modifications (e.g., adapted for a finite time horizon), would become similar to our
Algorithm~\ref{alg:whole} via direct matrix-vector multiplication for the cumulative score. Regarding this score, however, we make the following {\em novel} contributions: {\bf (a)} our \NP-hardness and submodularity proofs for the
cumulative score (those in \cite{GionisTT13} cannot be trivially extended to our case with any finite time horizon);
{\bf (b)} our random walk and sketch-based {\em efficient} algorithms, {\em with theoretical guarantees}, for the cumulative score ({\em more efficient} than
the greedy algorithm in \cite{GionisTT13}).

\vspace{-0.6mm}
\spara{Other Opinion Diffusion Models.}
Opinion diffusion has been investigated both from network science and statistical physics \cite{Weidlich1971THESD,Galam1982SociophysicsAN} perspectives,
and via discrete and continuous models.
In discrete models, an individual opinion is confined to be one of several integers; examples include the voter model \cite{HL75},
Axelrod model \cite{A97}, Sznajd model \cite{SS00}, majority rule models \cite{PhysRevLett.90.238701,Lambiotte_2008},
and social impact theory \cite{Bordogna_2007}. For instance, in the voter model, at each time stamp,
a node chooses a random neighbor and adopts the state (i.e., preference
for a certain campaigner) of this neighbor.
In contrast, continuous models, including
DeGroot \cite{DeG74} (the classic model) and its extensions ---
FJ \cite{FJ90,FJ99}, Deffuant \cite{Deffuant2000}, bounded confidence (BC) \cite{DNAW00}
and HK \cite{HK02} models, permit opinions to be represented by real numbers.
As such, these models are well-suited to be integrated with voting-based winning criteria in a multi-campaign setting.

%% file: 7-exp.tex
\begin{table}[t]
	\scriptsize
	\centering
	\begin{center}
	\vspace{-2mm}
		\caption{\small Characteristics of our datasets}
		\vspace{-3mm}
		\begin{tabular}{c||ccc}
			\textbf{Name}   & \textbf{\#Nodes} & \textbf{\#Edges} & \textbf{\#Candidates} \\ \hline \hline
			\revise{\textbf{DBLP}}   & \revise{63\,910} & \revise{2\,847\,120}    & \revise{2} \\
			\textbf{Yelp}   & 966\,240 & 8\,815\,788     & 10 \\ 
			\textbf{Twitter\_US\_Election} & 2\,246\,604 & 4\,270\,918 & 4 \\ 
			\textbf{Twitter\_Social\_Distancing} & 3\,244\,762 & 4\,202\,083 & 2 \\ 
			\textbf{Twitter\_Mask} & 2\,341\,769 & 3\,241\,153 & 2 \\	
		\end{tabular}
		\vspace{-6mm}
		\label{tab:dataset}
	\end{center}
\end{table}
\vspace{-1mm}
\section{Experimental Results}
\label{sec:exp}
\vspace{-1mm}
We perform experiments to demonstrate the accuracy, efficiency, scalability, and memory usage of our
methods\eat{, along with the robustness of our scores and the characteristics of the seeds}. Our code (available at \cite{code})
is executed on a single core, 512GB, 2.4GHz Xeon server.
\vspace{-1.5mm}
\subsection{Experimental Setup}
\label{sec:setup}
\vspace{-1mm}
\spara{Datasets.} We obtain \revise{five} directed graphs from \revise{three} real sources (Table \ref{tab:dataset}).
\revise{{\bf (1)} {\sf DBLP} \cite{dblp} is a well-known collaboration network. Nodes are users and edges are co-author relations. We only consider senior researchers who have published at least 50 papers.}
{\bf (2)} {\sf Yelp} \cite{yelp}
is a network of users who review businesses. Nodes are users and edges are friendships. We generate a graph based on restaurant-related records.
{\bf (3)} {\sf Twitter} is a social network. Nodes are users and edges are re-tweet relationships. We generate graphs from 24M tweets (Jul. 1 to Nov. 11, 2020) related to US elections \cite{25te-j338-20},
and 75M tweets (Mar. 19 to Oct. 5, 2020) related to two topics (``Social distancing'' and ``Wear a mask'') about COVID-19 \cite{781w-ef42-20}.

\vspace{-0.8mm}
\spara{Candidates.} \revise{{\bf (1)} {\sf DBLP.} We consider the candidates for the post of President in the ACM general election 2022, i.e., Yannis E. Ioannidis and Joseph A. Konstan.} {\bf (1)} {\sf Yelp.} We use the restaurant categories as candidates, e.g., American, Chinese, Italian, etc.
{\bf (2)} {\sf Twitter}. The political parties (Democratic, Republican, Green, Libertarian) are the candidates in {\em Twitter\_US\_Election}.
For each of the topics related to COVID-19, people may tweet for or against it. These two standpoints are the candidates in the respective Twitter COVID-19 datasets.
Without loss of generality, we
consider the following default target candidates for the respective datasets: \revise{``Joseph A. Konstan''}, ``Chinese Restaurant'', ``Democratic Party'', ``For Wearing a Mask'', and ``For Social Distancing''.

\vspace{-0.6mm}
\spara{Edge Weights.} Intuitively, for each category in {\sf Yelp}, if user $v$ visits a restaurant within one month of her friend $u$ (called a common visit), we say that $u$ influences $v$. Also, more common visits implies higher influence, and hence a larger edge weight. Thus, the edge $(u, v)$ is assigned a weight of $1 - e^{- a / \mu}$ \cite{PBGK10}, where $a$ is the number of common visits.
We set $\mu=10$ by default (details given in Appendix~\ref{sec:mu}). Similarly, we obtain edge weights \revise{(1) using the co-authorship counts for {\sf DBLP};} and (2) using
the number of retweets of a user pair for the {\sf Twitter} datasets. Finally, we normalize the edge weights such that the incoming weights of each node add up to 1.

\vspace{-0.8mm}
\spara{Initial Opinion Values.} \revise{{\bf (1)} {\sf DBLP.} A user's initial opinion is computed as the cosine similarity between the embeddings (obtained using SpaCy \cite{vasiliev2020natural}) of her papers to those of a candidate.} {\bf (2)} {\sf Yelp.} We use the average rating of a user towards a category as the initial opinion value.
{\bf (3)} {\sf Twitter.} 
We set the average sentiment score (computed using VADER \cite{HG14}) of each user about each candidate as her initial opinion. All the initial opinion values are normalized to $[0,1]$.

\vspace{-0.8mm}
\spara{Stubbornness Values.} \revise{{\bf (1)} {\sf DBLP} (resp. {\bf (2)} {\sf Yelp}). We set the stubbornness value of a user to 1 minus the variance of her yearly (resp. monthly) average opinions (as above)}, since a stubborn user is less likely to change her opinion about a candidate.
{\bf (3)} {\sf Twitter.} Since most users have only 1 tweet, we assign stubborness values uniformly at random in $[0,1]$.

\vspace{-0.6mm}
\spara{Methods Compared.} We find the best seed set by {\bf (1)} {\sf Direct Matrix Multiplication (DM)} via the greedy framework, coupled with CELF optimization \cite{LeskovecKGFVG07}. {\bf (2)} {\sf Random Walk Simulation (RW)} and {\bf (3)} {\sf Reverse Sketching (RS)} methods are implemented for better efficiency, with accuracy guarantees. We compare them with {\bf (4)} {\sf Independent Cascade (IC)} and {\bf (5)} {\sf Linear Threshold (LT)} models-based seed selection, both coupled with {\sf IMM} \cite{TSX15}, considering only the edge weights, and assuming that a user has only one chance to accept or reject a candidate. Multi-campaign versions {\sf MCIC} and {\sf MCLT} \cite{BAA11,OCC16} also exist. However, in our problem setting, the opinions diffuse independently for different candidates, and our algorithm selects seeds for the target candidate. With this setting, {\sf MCIC} and {\sf MCLT} reduce to {\sf IC} and {\sf LT}, respectively. Thus, we do not include them in our experiments. In addition, we also compare against the {\bf (6)} {\sf Greedy} algorithm in \cite{GionisTT13} for opinion maximization, adapted for a finite time horizon, which is denoted by {\sf GED-T}. Other baselines include seed selection via {\bf (7)} {\sf PageRank score (PR)} (based on the intuition that more frequently reached nodes in a random graph traversal are more likely to influence other users), {\bf (8)} {\sf Random Walk with Restart (RWR)} \cite{GionisTT13} and {\bf (9)} {\sf Degree Centrality (DC)}. All baselines differ only in the seed selection methods. Once the seeds are selected, all of them are evaluated in the same multi-campaign setting with the same diffusion model and scores as in \S~\ref{sec:preliminaries}. We could not compare against \cite{AKPT18} since their algorithms only work for small graphs and require more than 512GB memory on our datasets.

\vspace{-0.8mm}
\spara{Parameters.} {\bf (1)} {\sf Seed set size (k).} We vary $k$ from 100 to 2000. In \S~\ref{sec:sensitivity}, $k$ is set to 100 by default. {\bf (2)} {\sf Time horizon (t).} We vary $t$ from 0 to 30 steps (default: 20 steps). {\bf (3)} {\sf Random Walk Simulation.} We vary $\rho$ from 0.75 to 0.95 (default: 0.9). $\delta$ is set to 0.1. {\bf (4)} {\sf Sketches.} We vary $\epsilon$ from 0.05 to 0.3 (default: 0.1). $l$ is set to 1 following \cite{TSX15}.

\vspace{-0.8mm}
\spara{Performance Metrics.} {\bf (1)} {\sf Accuracy.} We report the {\em cumulative}, {\em \revise{plurality}}, and {\em \revise{Copeland}} scores (\S~\ref{sec:voting}) of the seed sets returned by the above methods.
{\bf (2)} {\sf Efficiency.} We report the running time of each method for finding the best seed set.

\begin{table*}[t!]
	\scriptsize
	\centering
	\begin{center}
	    \revise{
		\caption{\small Case study: ACM General Election (\#Users=63910, \#Seeds=100, Time Horizon $t=20$)}
		\vspace{-3mm}
		\begin{tabular}{c|l|c|c|c}
		    \hline	\multirow{2}{*}{\textbf{Domain}} &	\textbf{Top-10 seeds and their distribution across domains} &
		    \multirow{2}{*}{\textbf{Total \#users}} &\multicolumn{2}{c}{\textbf{\# Users voting for target candidate}}  \\  \cline{4-5}
			& \textbf{in which they influence the most} & & Without seeds & With seeds\\ \hline \hline
			\multirow{3}{*}{Data Management (DM)} & \{Jiawei Han, Victor C. M. Leung, Philip S. Yu, & \multirow{3}{*}{5056} & \multirow{3}{*}{\ \ 1138 (22.5\%)\ \ } & \multirow{3}{*}{4060 (80.3\%)} \\
            & Lei Zhang, Athanasios V. Vasilakos, Dusit Niyato & & \\
            &  Witold Pedrycz\} \quad\qquad\qquad\qquad\qquad\qquad\qquad\qquad & & \\ \hline
            Human Computer & \{Yoshua Bengio, H. Vincent Poor, Lei Zhang, & \multirow{2}{*}{4688} & \multirow{2}{*}{360 (7.7\%)} & \multirow{2}{*}{3345 (71.4\%)} \\
            Interaction (HCI) &  Dusit Niyato\}\qquad\qquad\qquad\qquad\qquad\qquad\qquad& & \\ \hline
            \multirow{2}{*}{Machine Learning (ML)} & \{Yoshua Bengio, Philip S. Yu, Witold Pedrycz,  & \multirow{2}{*}{4263} & \multirow{2}{*}{161 (3.8\%)} & \multirow{2}{*}{3125 (73.3\%)} \\
            & Jiawei Han\}\qquad\qquad\qquad\qquad\qquad\qquad\qquad\qquad& & \\ \hline
            \multirow{2}{*}{Computer Networks (CN)} & \{ H. Vincent Poor, Dusit Niyato, Luca Benini, & \multirow{2}{*}{4969} & \multirow{2}{*}{1241 (25.0\%)} & \multirow{2}{*}{4620 (93.0\%)} \\
            & Victor C. M. Leung, Lei Zhang\} \qquad\qquad\qquad\qquad & & \\ \hline
            Algorithms (AL) & \{Athanasios V. Vasilakos, Witold Pedrycz\}  & 2641 & 136 (5.1\%) & 1382 (52.3\%) \\  \hline
            Software (SW) & \{Luca Benini\}  & 1729 & 936 (54.1\%) & 1528 (88.4\%) \\  \hline
            Hardware (HW) & \{Luca Benini, H. Vincent Poor\}  & 4113 & 780 (19.0\%) & 3486 (84.8\%) \\  \hline
		\end{tabular}}
		\label{tab:case}
	\end{center}
\end{table*}

\begin{table*}[t!]
	\scriptsize
	\vspace{-2mm}
	\centering
	\begin{center}
	    \revise{
		\caption{\small Topics constituting the domains}
		\vspace{-3mm}
		\begin{tabular}{p{0.05\linewidth}|p{0.9\linewidth}}
	        \hline	
		    Domain & Topics \\ \hline
		    DM & data management, database systems, data mining, query processing, indexing, graphs, knowledge bases, clustering, social networks, recommender systems, data analysis, data streams, anomaly detection, information flow, semantic web, information retrieval, association rules, ranking, schema, relational, XML, joins  \\ \hline
		    HCI & recognition systems, detection systems, multimedia applications, image processing, signal processing, adaptive filtering, digital filtering, FIR filtering, language models, pose estimation, motion estimation, face recognition, speech recognition, natural languages, image sensors, image annotation, computer graphics, human actions, 3D reconstruction, moving objects, user interfaces  \\ \hline
		    ML & neural networks, Bayesian networks, Gaussian processes, reinforcement learning, machine learning, active learning, probabilistic models, Markov model, particle filtering, collaborative filtering, recommender systems, decision trees, time series, recurrent neural, feature selection, random fields, regression, classification, pattern matching \\ \hline
		    CN & distributed networks, cellular networks, ad-hoc networks, overlay networks, area networks, mobile networks, peer-to-peer networks, wireless networks, signal processing, adaptive filtering, digital filtering, FIR filtering, congestion control, routing protocols, wireless communications, fading channels, wireless sensors  \\ \hline
		    AL & linear systems, non-linear systems, graphs, approximation algorithms, data structures, programming languages, linear programming, dynamic programming, shortest paths, proofs, theorems, algebra, polynomial, quantum \\ \hline
		    SW & software systems, mobile applications, web applications, source code, programming languages, web services, web sites, software engineering, software development, user interfaces, software architecture \\ \hline
		    HW & real-time systems, embedded systems, control systems, distributed systems, scheduling, virtual machines, state machines, access control, power control, VLSI, FPGA, integrated circuits, digital circuits, analog circuits, power amplifiers, shared memory, synthesis tools, on-chip, caches, clocks, CMOS, mobile devices  \\ \hline
		\end{tabular}}
		\vspace{-6mm}
		\label{tab:topics}
	\end{center}
\end{table*}

\begin{figure}[t!]
	\vspace{-1mm}
	\centering
	\revise{
	\subfigure[\small {\em Without Seeds}]
	{\includegraphics[scale=0.3]{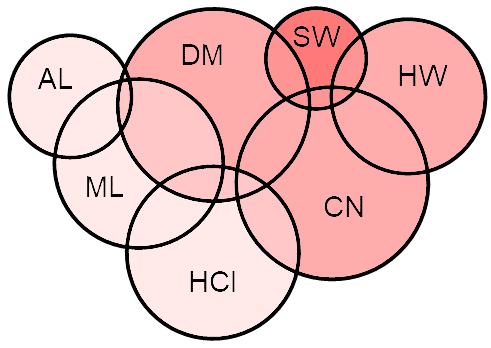}
		\label{fig:case_without}}
    \subfigure[\small {\em With Seeds}]
	{\includegraphics[scale=0.3]{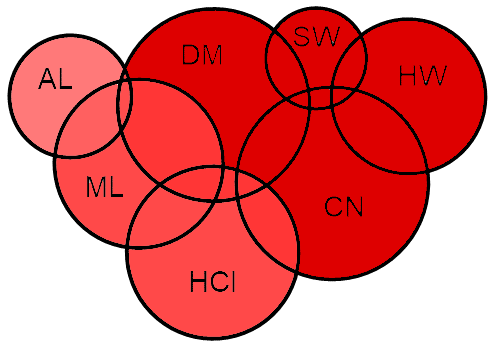}
		\label{fig:case_with}}	
	\vspace{-3mm}
	\caption{\small Case study: ACM general election  (\#Users=63910, \#Seeds =100, Time Horizon $t=20$). The size of each circle denotes the population of users in each domain, while the color captures the percentage of users who vote for the target candidate (Joseph A. Konstan). Darker color corresponds to higher percentage.}
	\label{fig:case}
	\vspace{-2mm}}
\end{figure}

\vspace{-2mm}
\subsection{\revise{Case Study: ACM General Election 2022; {\sf DBLP} Dataset}}
\label{sec:case}
\vspace{-1mm}
\revise{We observe that after including only the top-100 seeds, the number of users favoring our target candidate {\em Joseph A. Konstan} will significantly increase from 13\,990 (21.8\%) to 46\,433 (72.7\%), which might have reversed the election result. We select 7 frequent domains\footnote{\revise{We assume that a user may belong to at most 3 domains based on the frequencies of several keywords in the titles of their publications.
The selected keywords for each domain can be found in Table \ref{tab:topics}.}}
for the users who change their preferred candidates, and show the top-10 seeds and the domains in which these seeds influence the most 
(Table~\ref{tab:case}). Figure~\ref{fig:case} visualizes the domain overlaps 
and the percentage of users voting for our target candidate {\em Joseph A. Konstan}. Notice that a seed user may influence users from several domains. As DM is a common domain of both candidates, 7 out of the top-10 seeds are also active in the DM domain.
Only 1-2 seeds are from the SW and HW domains, since (1) the users in the SW domain already favor our target candidate more based on their initial opinions (thus introducing seeds who can influence users in this domain is not that useful); (2) the HW domain does not overlap with the DM domain. The number of seeds who influence the HCI, ML, and CN domains are higher, because (1) these domains have larger populations; (2) these domains have large overlaps with DM; and (3) the users in these domains initially prefer the competitor ({\em Yannis E. Ioannidis}) more, thus introducing seed nodes who can influence users in these domains is more helpful. Furthermore, we investigate the average distance 
between the candidates and those users who change minds after introducing the seeds. 14.5\% of them are closer to the target candidate, and 10.2\% of them are closer to the competitors (about 2 hops away). The majority of these users (75.3\%) are almost equidistant from both candidates (more than 3 hops away). This demonstrates that our solution focuses more on affecting the neutral users whose preferences are usually easier to switch.}

\begin{figure*}
  \centering
  \vspace{-4mm}
  \includegraphics[scale=0.47,angle=270]{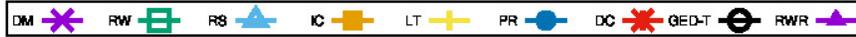}
  \vspace{-7mm}
  \caption{\small \revise{Legends for the methods compared in Figures \ref{fig:varyK_rank}-\ref{fig:varyK_cumu}}}
  \label{fig:legends}
    \vspace{-2mm}
\end{figure*}
\begin{figure*}[t!]
	\vspace{-4mm}
	\centering
	\subfigure[\small {\em \revise{Yelp}}]
	{\includegraphics[scale=0.154,angle=270]{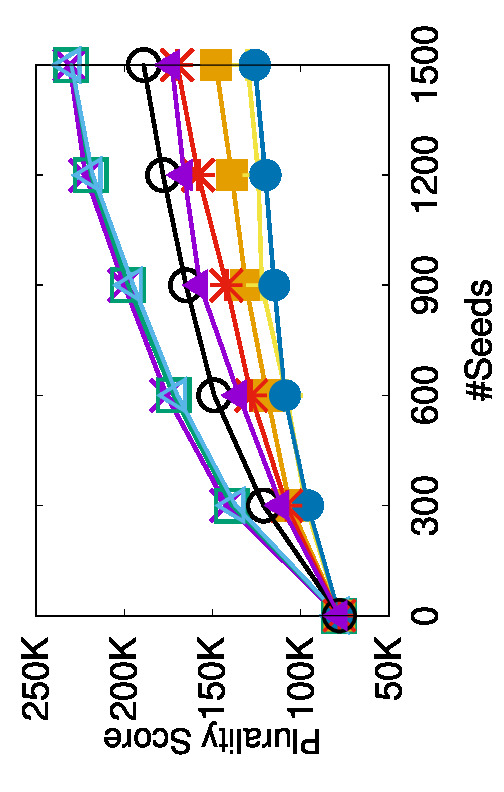}
		\label{fig:yelp_rank}}
	\subfigure[\small \revise{{\em Twitter\_US\_Election}}]
	{\includegraphics[scale=0.154,angle=270]{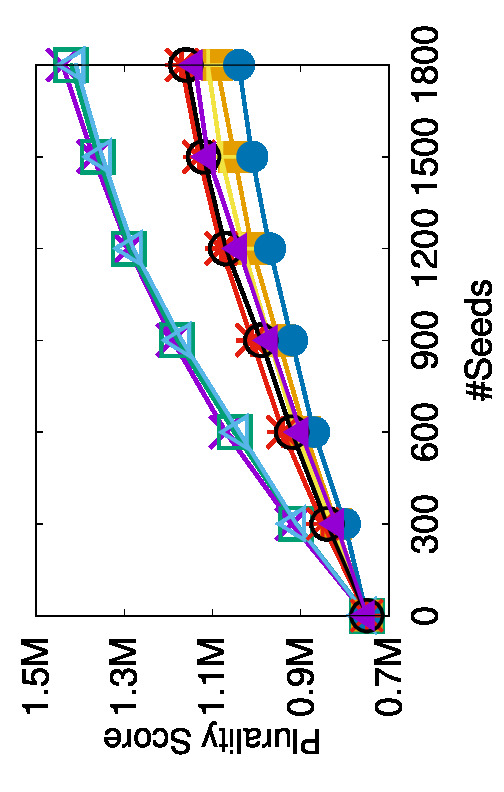}
	    \label{fig:us_rank}}
    \subfigure[\small {\em Twitter\_Mask}]
	{\includegraphics[scale=0.154,angle=270]{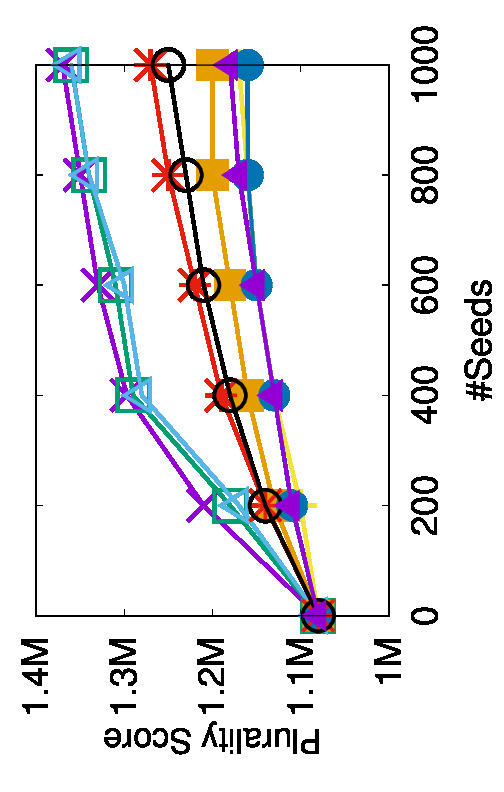}
		\label{fig:mask_rank}}	
    \subfigure[\small {\em Twitter\_Mask}]
	{\includegraphics[scale=0.154,angle=270]{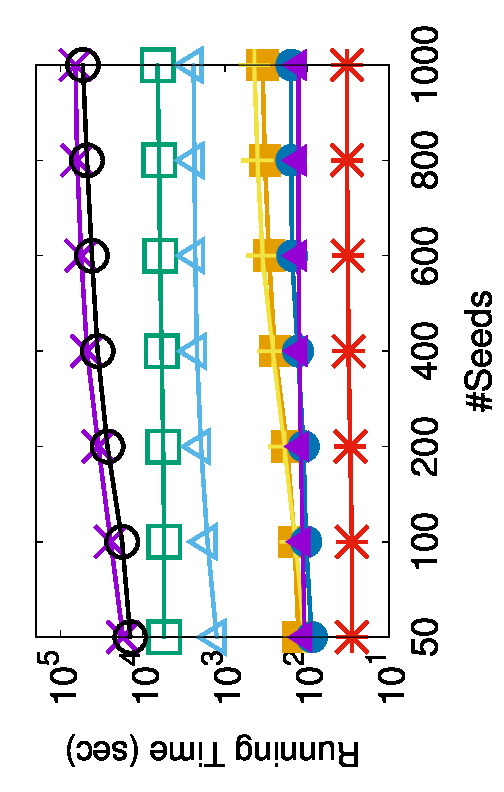}
		\label{fig:t_rank}}
	\vspace{-5mm}
	\caption{\small \revise{Plurality} score vs. seed set size $k$: (a-c) effectiveness, (d) efficiency}
	\label{fig:varyK_rank}
	\vspace{-2.5mm}
\end{figure*}
\begin{figure*}[t!]
	\vspace{-4mm}
	\centering
	\subfigure[\small {\em Yelp}]
	{\includegraphics[scale=0.154,angle=270]{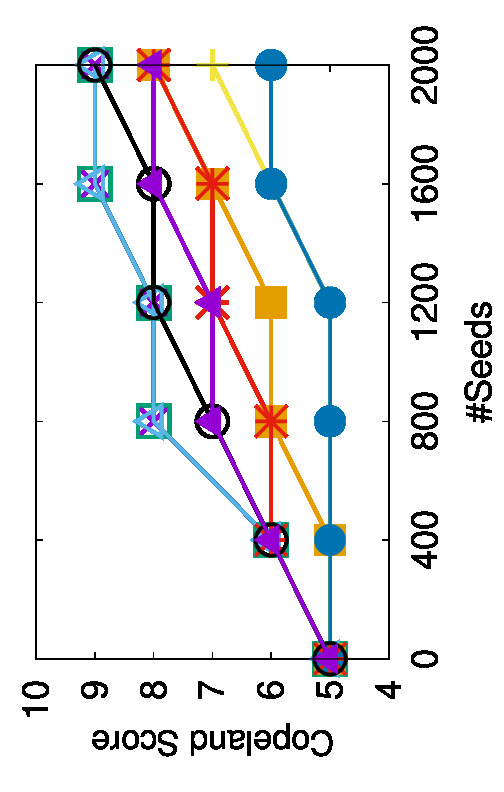}
		\label{fig:yelp_cond}}
	\subfigure[\small {\em \revise{Twitter\_US\_Election}}]
	{\includegraphics[scale=0.154,angle=270]{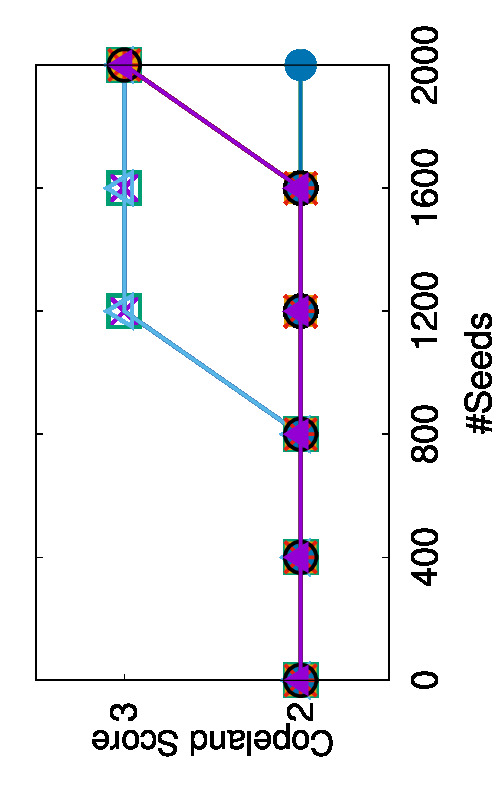}
		\label{fig:us_cond}}
	\subfigure[\small {\em \revise{Twitter\_Mask}}]
	{\includegraphics[scale=0.154,angle=270]{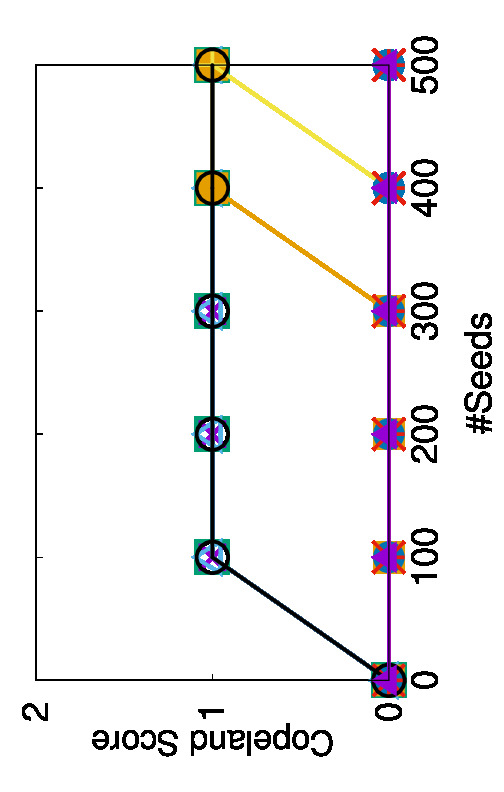}
		\label{fig:mask_cond}}
	\subfigure[\small {\em Twitter\_Mask}]
	{\includegraphics[scale=0.154,angle=270]{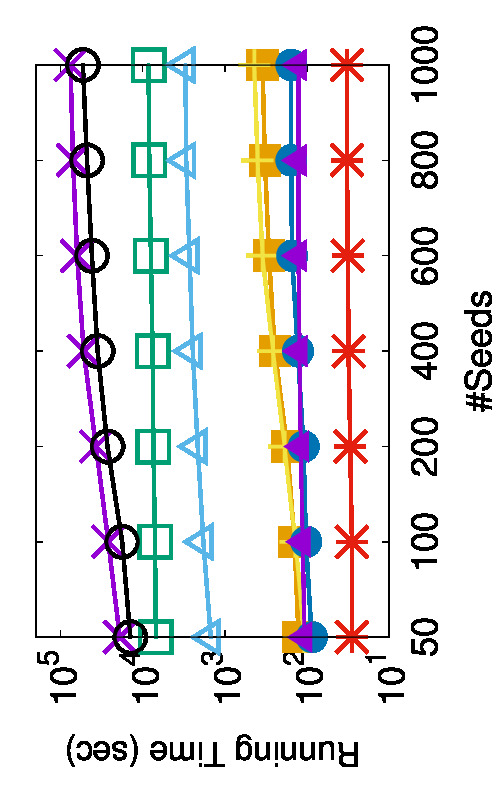}
		\label{fig:t_cond}}
	\vspace{-5mm}
	\caption{\small \revise{Copeland} score vs. seed set size $k$: (a-c) effectiveness, (d) efficiency}
	\label{fig:varyK_cond}
	\vspace{-2.5mm}
\end{figure*}
\begin{figure*}[t!]
	\vspace{-4mm}
	\centering
	\subfigure[\small {\em \revise{Yelp}}]
	{\includegraphics[scale=0.154,angle=270]{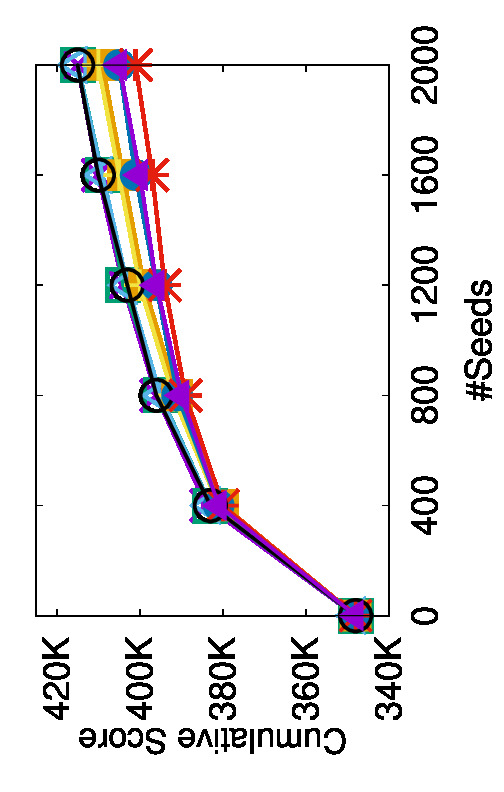}
		\label{fig:yelp_cumu}}
	\subfigure[\small {\emph{ Twitter\_US\_Election}}]
	{\includegraphics[scale=0.154,angle=270]{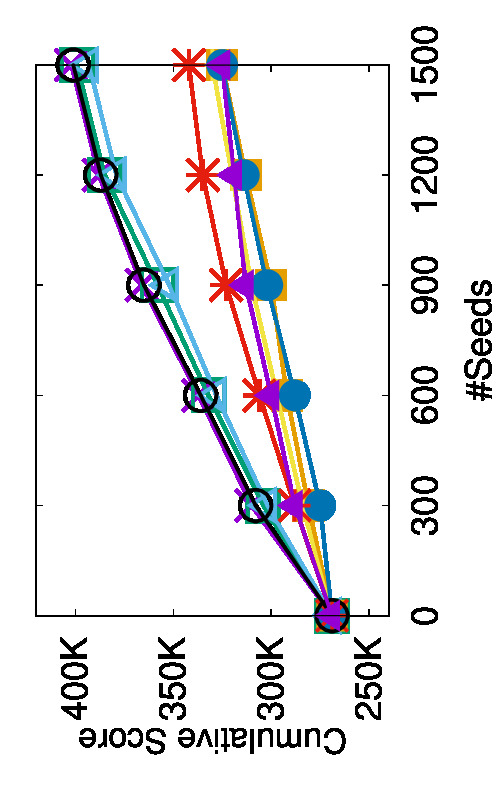}
		\label{fig:us_cumu}}
	\subfigure[\small {\em \revise{Twitter\_Mask}}]
	{\includegraphics[scale=0.154,angle=270]{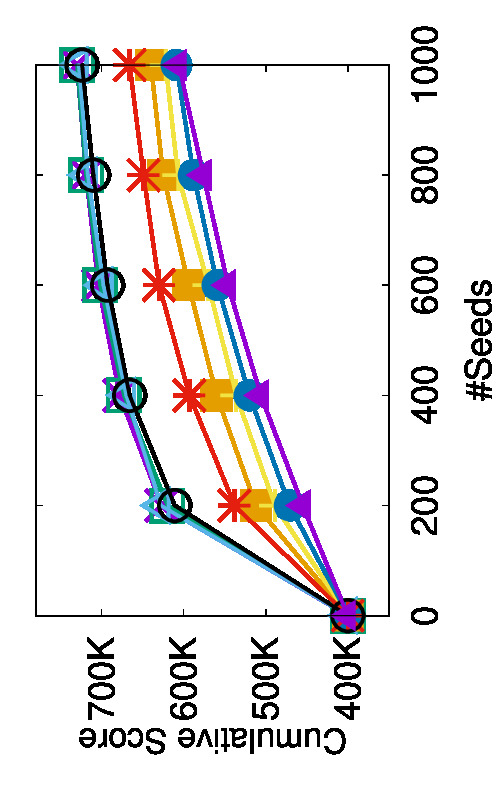}
		\label{fig:mask_cumu}}
	\subfigure[\small {\em Twitter\_Mask}]
	{\includegraphics[scale=0.154,angle=270]{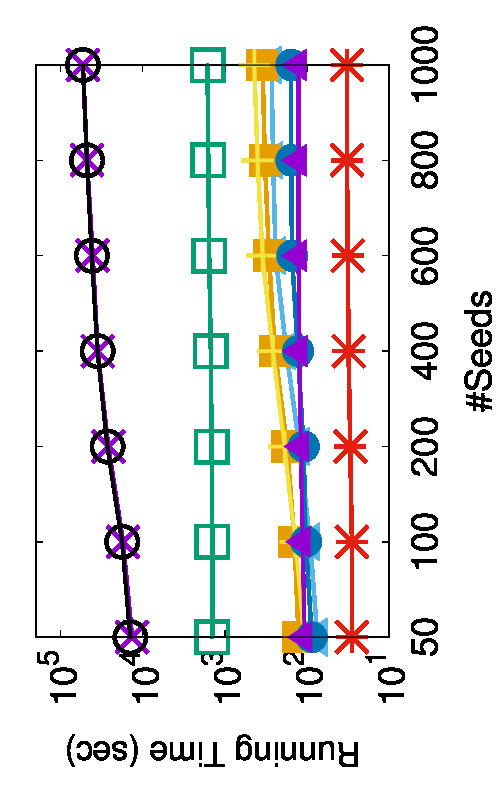}
		\label{fig:t_cumu}}
	\vspace{-5mm}
	\caption{\small Cumulative score vs. seed set size $k$: (a-c) effectiveness, (d) efficiency}
	\label{fig:varyK_cumu}
	\vspace{-4.5mm}
\end{figure*}

\vspace{-1mm}
\subsection{Performance Analysis}
\label{sec:varyK}
\vspace{-2mm}
\spara{Accuracy.} Our proposed methods outperform the baselines in all voting-based scores (Figures~\ref{fig:varyK_rank}-\ref{fig:varyK_cumu} (a-c)),
with the exception of our {\sf DM} vs. baseline {\sf GED-T} for the cumulative score.
The scores increase with the number of seeds $k$, and the growth rates are higher when $k$ is small.
For the \revise{plurality} and \revise{Copeland} scores, the proposed methods outperform the baselines more significantly.
For example, in {\em Twitter\_Social\_Distancing}, the best baseline {\sf DC} reaches up to 70\% of {\sf RW} with the cumulative score, while it
attains only 50\% of {\sf RW} with the \revise{plurality} score (the actual score difference is nearly 100K users, which can lead to a significant
impact in, e.g., an election's outcome). The classic IMM algorithm coupled with the IC and LT models performs poorly with voting-based scores,
as does {\sf GED-T}, since their seeds maximize different objective functions.
Recall that {\sf GED-T} is the greedy algorithm for opinion maximization \cite{GionisTT13}, adapted for a finite time horizon.
The cumulative score, due to its aggregate nature, is similar to opinion maximization in the single campaign setting, and therefore
our {\sf DM} and baseline {\sf GED-T} perform the same for the cumulative score ({\em only}).

\vspace{-1mm}
\spara{Efficiency.} The running time of {\sf RW} remains nearly the same for different $k$
(Figures~\ref{fig:varyK_rank}-\ref{fig:varyK_cumu} (d)),
while that of {\sf RS} increases slightly with $k$.
For {\sf RW}, we generate a fixed number (independent of $k$) of random walks starting from
each node (Theorem~\ref{th:cumu_rws}); while for {\sf RS}, we generate
one random walk starting from $\theta$ randomly sampled nodes (Theorem~\ref{th:approx_ratio}).
A larger $k$ does not necessarily increase $\theta$ as {\bf (1)} $OPT$ in the
denominator increases with $k$; {\bf (2)} $\binom{n}{k}$ in the numerator also increases
with $k$. Moreover, the random walk generation dominates the running time of both {\sf RW} and
{\sf RS}. The running time of {\sf DM} increases linearly with $k$, since it applies matrix-vector
multiplication in each of $k$ iterations.
The running times for the \revise{plurality} and \revise{Copeland} scores are higher than those of the cumulative
score, but follow the same trend. We also find that, among our proposed algorithms, {\sf RS} is the most efficient, and has accuracy
comparable to the others. {\em Therefore, we recommend {\sf RS} as our ultimately proposed method}.
Notice that {\sf RS} is about two orders of magnitude faster than {\sf GED-T}, even for the cumulative score.
\begin{figure}[t!]
	\vspace{-2mm}
	\centering
	\subfigure[\small {Positional-\revise{2-approval}}]
	{\includegraphics[scale=0.147,angle=270]{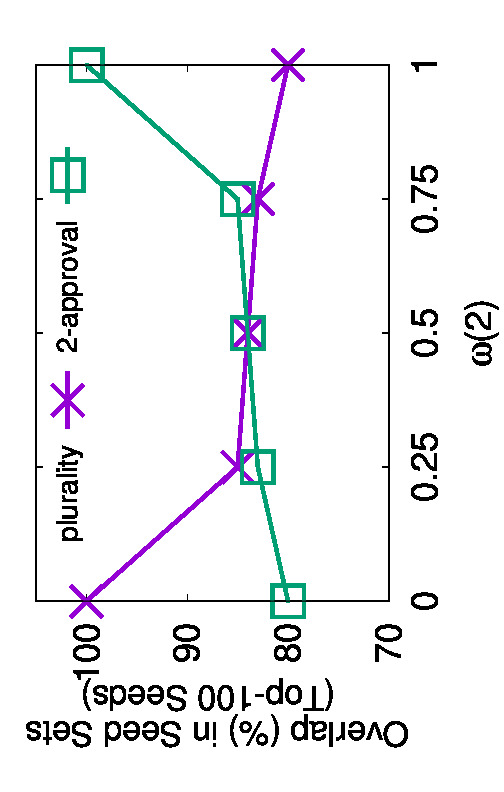}
		\label{fig:w-2-best}}
	\subfigure[\small {Positional-\revise{3-approval}}]
	{\includegraphics[scale=0.147,angle=270]{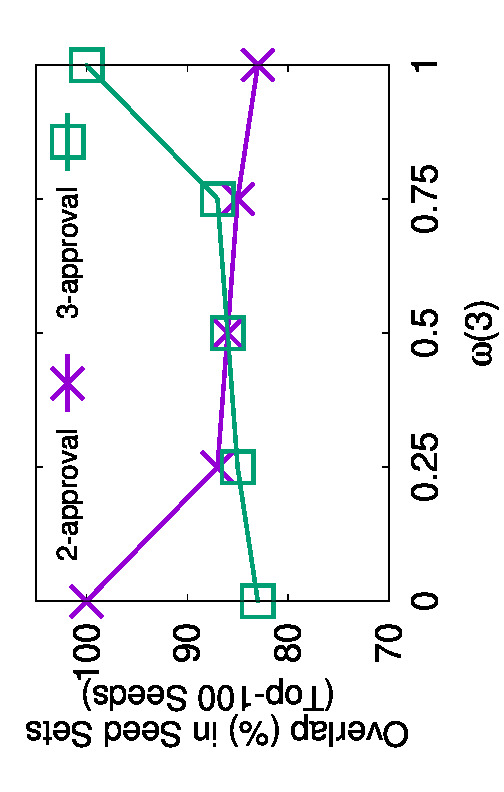}
		\label{fig:w-3-best}}
	\vspace{-3mm}
	\caption{\small Overlap of the seed set for the positional-$p$-approval score with respect to those for the \revise{plurality} and $p$-approval scores; {\em Yelp}}
	\label{fig:varyW}
	\vspace{-3mm}
\end{figure}
\begin{figure}[t!]
	\vspace{-2mm}
	\centering
	\includegraphics[scale=0.29]{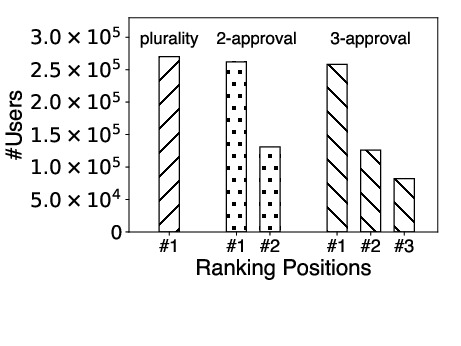}
	\vspace{-7mm}
	\caption{\small Number of users who rank the target candidate at the specified positions at the time horizon $t$; {\em Yelp}}
	\label{fig:pos}
	\vspace{-3mm}
\end{figure}

\vspace{-0.8mm}
\spara{Comparison among the \revise{plurality} score variants.} Figure~\ref{fig:varyW} shows the overlap of the seed sets ($k = 100$) 
returned for the \revise{plurality} score variants. For positional-$p$-\revise{approval}, we vary $\omega[p] \in [0, 1]$, while we keep $\omega[i]=1 \, \forall i < p$.
Thus, it becomes $p$-approval when $\omega[p]=1$ and $(p-1)$-approval when $\omega[p]=0$. 
The seed sets returned for \revise{plurality} and \revise{2-approval} have 80\% overlap. 
The seeds for \revise{plurality} help to improve the target candidate's first-position ranking for as many users as possible. 
However, once the ranking constraint is relaxed to also include the second-position ranking 
(e.g., \revise{2-approval}, positional-\revise{2-approval}), some seeds are changed to incorporate more users.
Similar results hold for the \revise{3-approval} variants. Figure~\ref{fig:pos} presents the ranking position distributions for various $p$. 
We also notice that all \revise{plurality} variants share similar running times.

\begin{table}[t]
    \centering
    \revise{
    \caption{Minimum seed set sizes achieved by our proposed methods for the target candidate to win w.r.t. the plurality score}
    \vspace{-2mm}
    \begin{tabular}{c||c|c|c}
        \hline
        {\bf Dataset} & {\sf DM} & {\sf RW} & {\sf RS} \\ \hline\hline
        {\em Twitter\_Mask} & 17 & 21 & 24 \\ \hline
        {\em Twitter\_Social\_Distancing} & 69 & 71 & 74 \\ \hline
    \end{tabular}}
    \vspace{-4mm}
    \label{tab:win}
\end{table}

\vspace{-0.8mm}
\spara{\revise{Minimum number of seeds for the target to win.}}
\revise{As discussed in \S~\ref{sec:overview}, we can adapt our methods to find the minimum number of seeds for the target to win. Table~\ref{tab:win} shows these values for our three proposed methods. For a ``more approximate'' method, the seed sets are ``less optimal'', and hence the minimum number of seeds required is larger.}

\begin{figure}[t!]
	\centering
	\subfigure[\small {IC Model}]
	{\includegraphics[scale=0.147,angle=270]{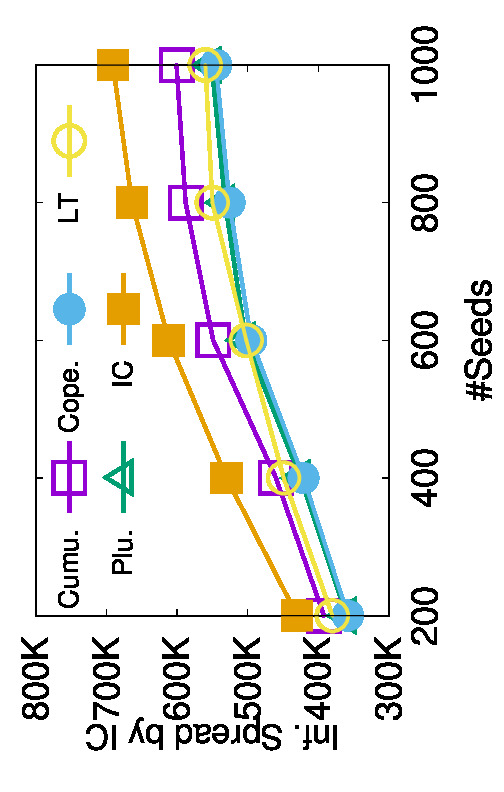}
		\label{fig:ic}}
	\subfigure[\small {LT Model}]
	{\includegraphics[scale=0.147,angle=270]{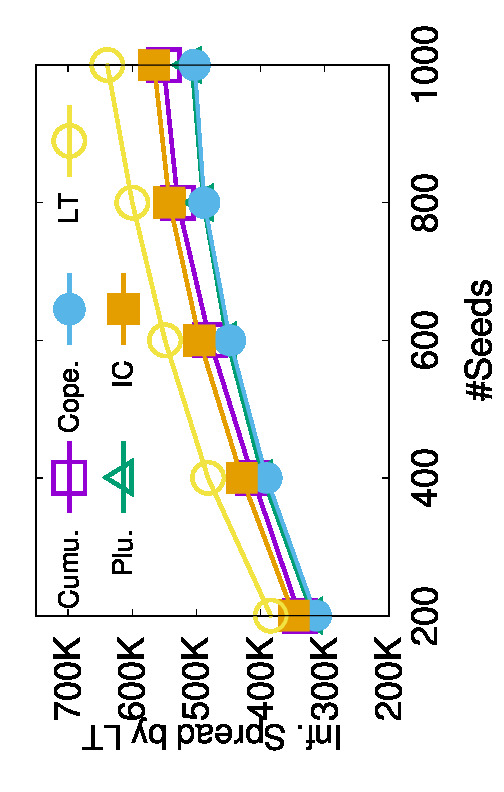}
		\label{fig:lt}}
	\vspace{-3mm}
	\caption{\small Expected influence spread over {\em Twitter\_Mask}.
		The seeds for the voting-based scores are selected by {\sf RW}.
	}
	\label{fig:inf}
	\vspace{-3mm}
\end{figure}
\vspace{-0.6mm}
\spara{Expected Influence Spread (EIS) Measurement.} EIS is the expected number of activated nodes from a given seed set when diffusion takes place following the IC or LT models \cite{KKT03}. For fairness, we compare the EIS of the seeds selected by {\sf RW} according to our three scores, with those of the seeds selected by {\sf IMM} \cite{TSX15} following the IC and LT models. This is done to demonstrate that the chosen seed set based on our proposed models and scores is not a bad solution with respect to the EIS. As shown in Figure \ref{fig:inf}, the performances of {\sf RW} and {\sf IMM} are comparable. The seeds given by {\sf RW} with the cumulative score can achieve over 80\% of the EIS of {\sf IMM} following both IC and LT. Thus, our seeds for the cumulative score work well even in the context of EIS following the IC and LT models.

\begin{figure}[t!]
	\centering
	\subfigure[\small {Score}]
	{\includegraphics[scale=0.147,angle=270]{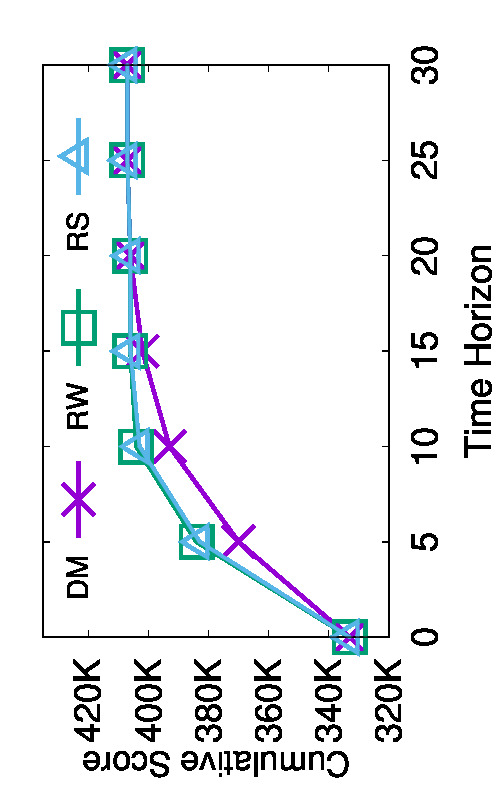}
		\label{fig:sens_t_score}}
	\subfigure[\small {Seed set finding time}]
	{\includegraphics[scale=0.147,angle=270]{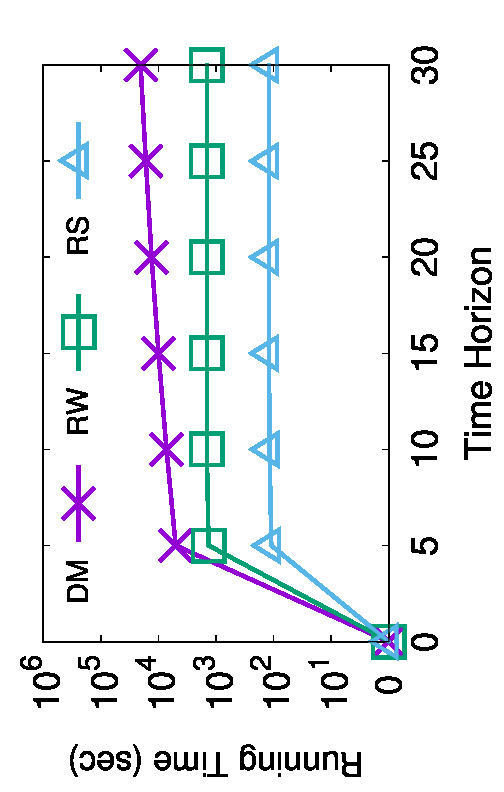}
		\label{fig:sens_t_time}}
	\vspace{-3mm}
	\caption{\small Cumulative score, seeds finding time vs. time horizon $t$; {\em Yelp}}
	\label{fig:varyT}
	\vspace{-3mm}
\end{figure}
\begin{figure}[t!]
	\vspace{-3mm}
	\centering
	\subfigure[\small {Varying seed set size, $k$}]
	{\includegraphics[scale=0.147,angle=270]{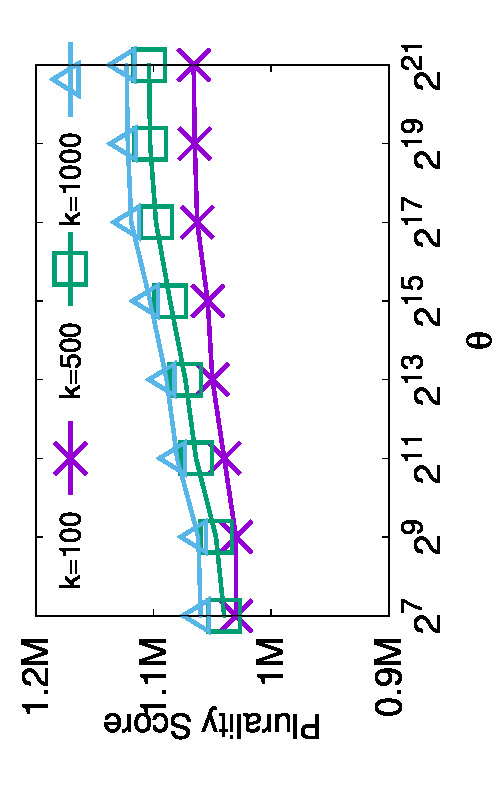}
		\label{fig:sens_theta_mask_k}}
	\subfigure[\small {Varying time horizon, $t$}]
	{\includegraphics[scale=0.147,angle=270]{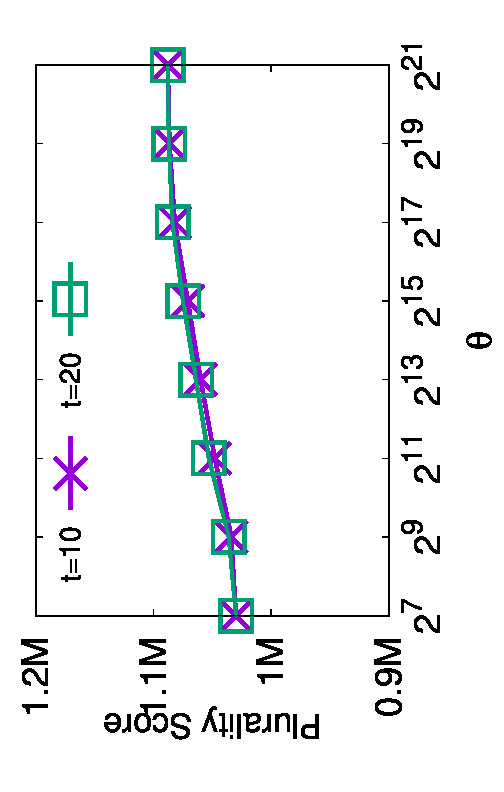}
		\label{fig:sens_theta_mask_t}}
	\vspace{-3mm}
	\caption{\small \revise{Plurality} score vs. $\theta$; {\em Twitter\_Mask}}
	\label{fig:varyTheta_mask}
	\vspace{-3mm}
\end{figure}
\begin{figure}[t!]
	\vspace{-3mm}
	\centering
	\subfigure[\small {Varying seed set size, $k$}]
	{\includegraphics[scale=0.147,angle=270]{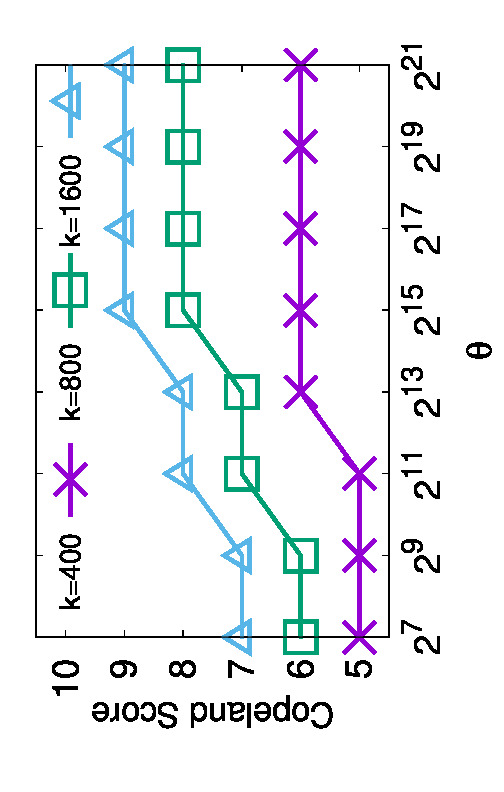}
		\label{fig:sens_theta_yelp_k}}
	\subfigure[\small {Varying time horizon, $t$}]
	{\includegraphics[scale=0.147,angle=270]{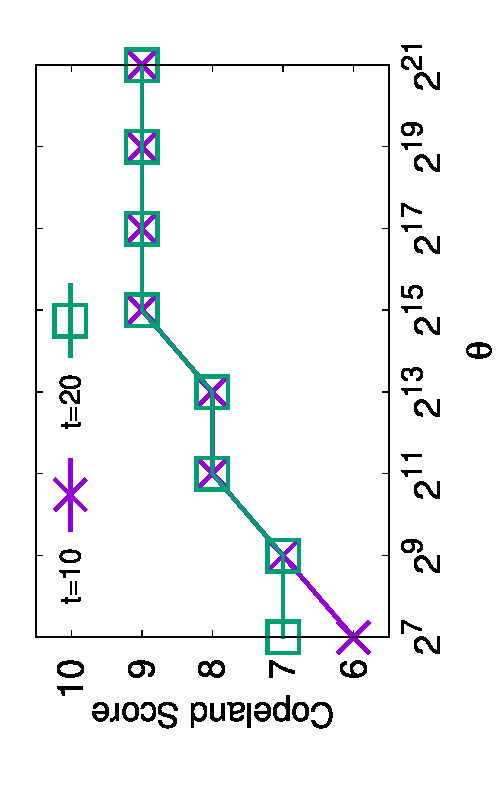}
		\label{fig:sens_theta_yelp_t}}
	\vspace{-3mm}
	\caption{\small \revise{Copeland} score vs. $\theta$; {\em Yelp}}
	\label{fig:varyTheta_yelp}
	\vspace{-4mm}
\end{figure}

\vspace{-1.5mm}
\vspace{-2mm}
\subsection{Parameter Sensitivity Analysis}
\label{sec:sensitivity}
\vspace{-1mm}
\spara{Impact of $t$.}
Figure~\ref{fig:varyT} shows that the cumulative score remains nearly the same after timestamp 20 for all the proposed methods. This happens slightly quicker for {\sf RW} and {\sf RS} than for {\sf DM}.
Thus, we set time horizon $t = 20$ as default in the rest of the experiments. The running time of {\sf DM} is more sensitive to $t$ than those of {\sf RW} and {\sf RS} because we need to conduct exactly $t$ rounds of matrix-vector multiplication in {\sf DM}, while random walks are often of length less than $t$ for {\sf RW} and {\sf RS}.

\vspace{-1mm}
\spara{Impact of $\theta$ for the \revise{plurality} and \revise{Copeland} scores.} We heuristically analyze the variation of these scores with $\theta$ (\S~\ref{sec:theta_exp}). Recall that {\sf RS} is more efficient than {\sf RW} only when $\theta < n$.
For a specific dataset and score, we empirically find the smallest $\theta$ when that score converges (for some $k$ and $t$),
which is $2^{19}$ for {\em Twitter\_Mask} with the \revise{plurality} score (Figure \ref{fig:varyTheta_mask}), and $2^{15}$ for {\em Yelp} with the \revise{Copeland} score (Figure \ref{fig:varyTheta_yelp}). Both values are smaller than the respective $n$. Moreover, this estimate
can be re-used on the same dataset and score, even with different $k$ and $t$ as inputs, since it is less sensitive to $k$ and $t$, as shown in Figures \ref{fig:varyTheta_mask} and \ref{fig:varyTheta_yelp}.

\vspace{-0.6mm}
\spara{Impact of $\epsilon$.} The parameter $\epsilon$ (for the cumulative score) controls how close the estimated score is to the true score in {\sf RS}, and affects the number of random walks to be generated.  Figure~\ref{fig:varyE} shows that the cumulative score suffers a drastic decrease from $\epsilon=0.1$ to $\epsilon=0.2$.
The running time decreases more sharply when $\epsilon$ is smaller. Thus, we select $\epsilon=0.1$ as the default value. 

\vspace{-0.6mm}
\spara{Impact of $\rho$.} The parameter $\rho$ controls the probability that the estimated score is the same as or close to the true score, and affects the number of random walks to be generated. As shown in Figure~\ref{fig:varyRho}, the \revise{plurality} score increases sharply when $\rho$ is small, while there is almost no difference from $\rho=0.9$ onward. The running time increases significantly with larger $\rho$.
Thus, we set $\rho=0.9$ as default in the rest of the experiments. 

\begin{figure}[t!]
	\vspace{-3mm}
	\centering
	\subfigure[\small {Score}]
	{\includegraphics[scale=0.147,angle=270]{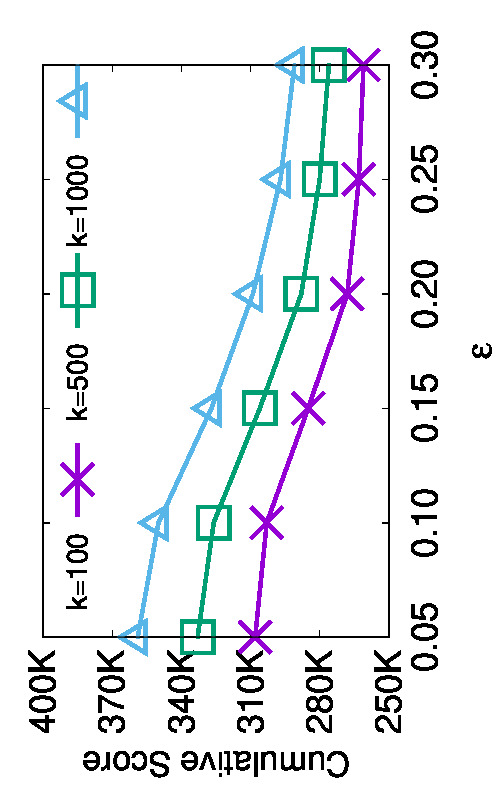}
		\label{fig:sens_e_score}}
	\subfigure[\small {Seed set finding Time}]
	{\includegraphics[scale=0.147,angle=270]{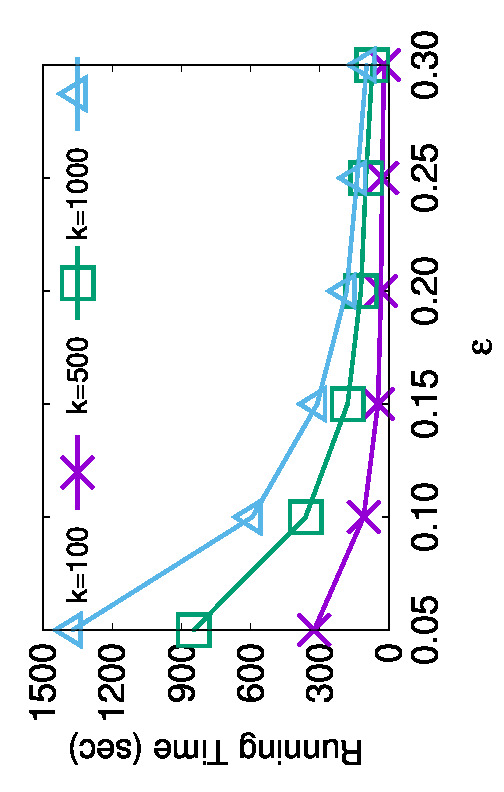}
		\label{fig:sens_e_time}}
	\vspace{-3mm}
	\caption{\small Cumulative score vs. $\epsilon$; {\sf RS} Method;  {\em Twitter\_US\_Election}} 
	\label{fig:varyE}
	\vspace{-6.5mm}
\end{figure}
\begin{figure}[t!]
	\centering
	\subfigure[\small {Score}]
	{\includegraphics[scale=0.147,angle=270]{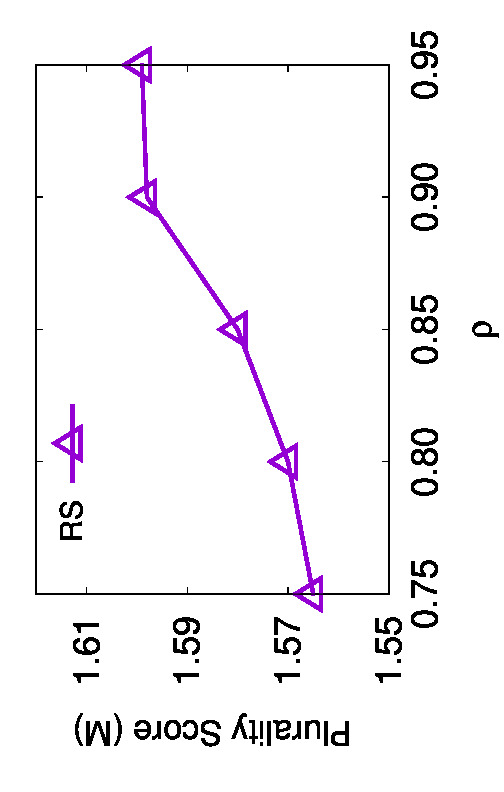}
		\label{fig:sens_r_score}}
	\subfigure[\small {Seed set finding Time}]
	{\includegraphics[scale=0.147,angle=270]{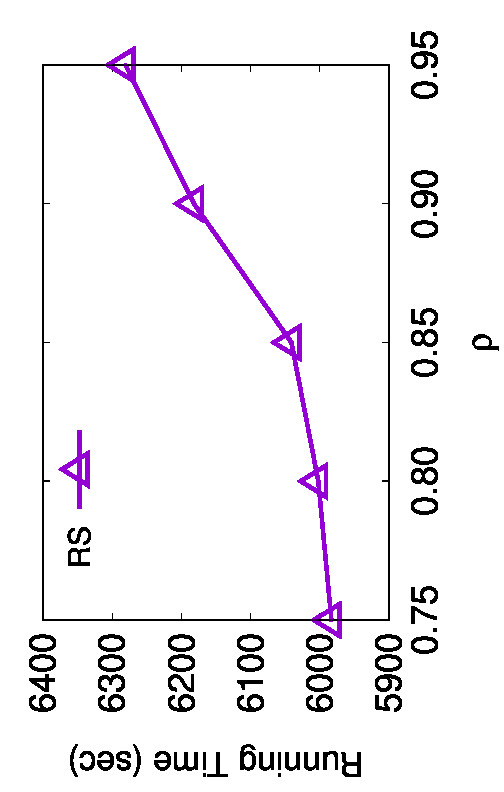}
		\label{fig:sens_r_time}}
	\vspace{-3mm}
	\caption{\small \revise{Plurality} score vs. $\rho$; {\em  Twitter\_Social\_Distancing}} 
	\label{fig:varyRho}
	\vspace{-6mm}
\end{figure}
\begin{figure}[t!]
	\centering
	\subfigure[\small {Seed set finding Time}]
	{\includegraphics[scale=0.147,angle=270]{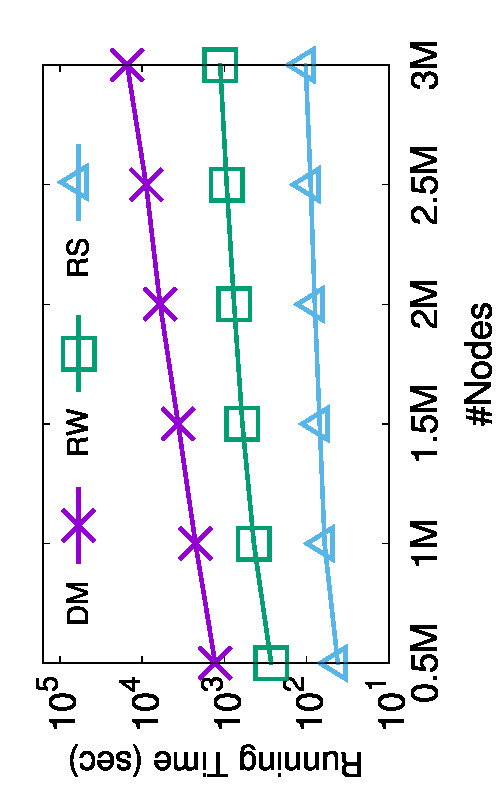}
		\label{fig:scala_time}}
	\subfigure[\small {Memory Usage}]
	{\includegraphics[scale=0.147,angle=270]{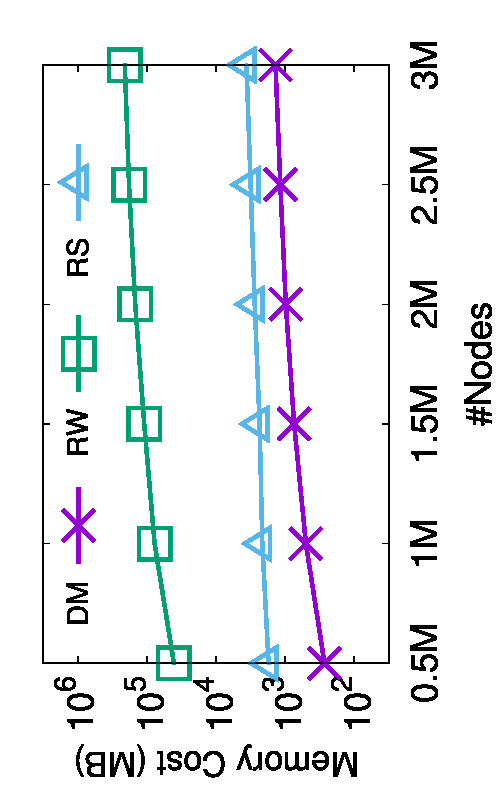}
		\label{fig:scala_mem}}
	\vspace{-3mm}
	\caption{\small Seed set finding time and memory usage for the cumulative score vs. graph size; {\em Twitter\_Social\_Distancing}}
	\label{fig:scala}
	\vspace{-5mm}
\end{figure}
\vspace{-2mm}
\subsection{Scalability and Memory Usage}
\label{sec:scalability}
\vspace{-1mm}
We test the scalability and memory usage of our algorithms with different graph sizes. The {\em Twitter\_Social\_Distancing} graph has about 3.2M nodes; we generate six graphs by
selecting 0.5M, 1M, 1.5M, 1M, 2.5M, 3M nodes uniformly at random, and apply our algorithms on the subgraphs induced by them. Figure~\ref{fig:scala_time} demonstrates that the running
times of {\sf RW} and {\sf RS} increase almost linearly with the number of nodes (the y-axis is logarithmic), which confirms good
scalability of our algorithms. The running time of {\sf DM} increases polynomially -- it has cubic growth with $n$ (\S~\ref{sec:overview}).

As shown in Figure \ref{fig:scala_mem}, {\sf DM} consumes the least memory since it only needs to store the edge weights, initial opinions, and stubbornness values. {\sf RW} and {\sf RS} further store random walks. ({\sf RW} far more than {\sf RS}). Our ultimately proposed method, {\sf RS}, consumes only a few GB for the {\em Twitter\_Social\_Distancing}
dataset.
\eat{
\subsection{Robustness Analysis}
\label{sec:roubstness}
We analyze the robustness of our scores for a given seed set by adding gaussian noise with mean $0$ and standard deviation $1$
to each of {\bf (1)} initial opinions, {\bf (2)} influence weights, and {\bf (3)} structure of the influence graph, separately,
corresponding to the target candidate.
For the case of initial opinions, we uniformly at random sample $x\%$ of the entries of the initial opinion vector, 
and add noise therein. For influence weights, we uniformly at random sample 
$x\%$ of the columns of the influence matrix, and add noise to all non-zero entries in those columns.
In case of structure of the influence graph, we also select $x\%$ of the columns of the influence matrix (uniformly
at random sampled), then add noise to all non-zero entries, as well as to 0.001\% of the zero entries (randomly chosen).
We constrain the values to be in the range $[0,1]$ after adding noise, and also normalize those columns in the influence matrix.
We vary $x\%$ as $0\%, 10\%, 20\%, 30\%$, and $40\%$, indicating the percentage of noise added.
\begin{figure}
	\centering
	\includegraphics[scale=0.18,angle=270]{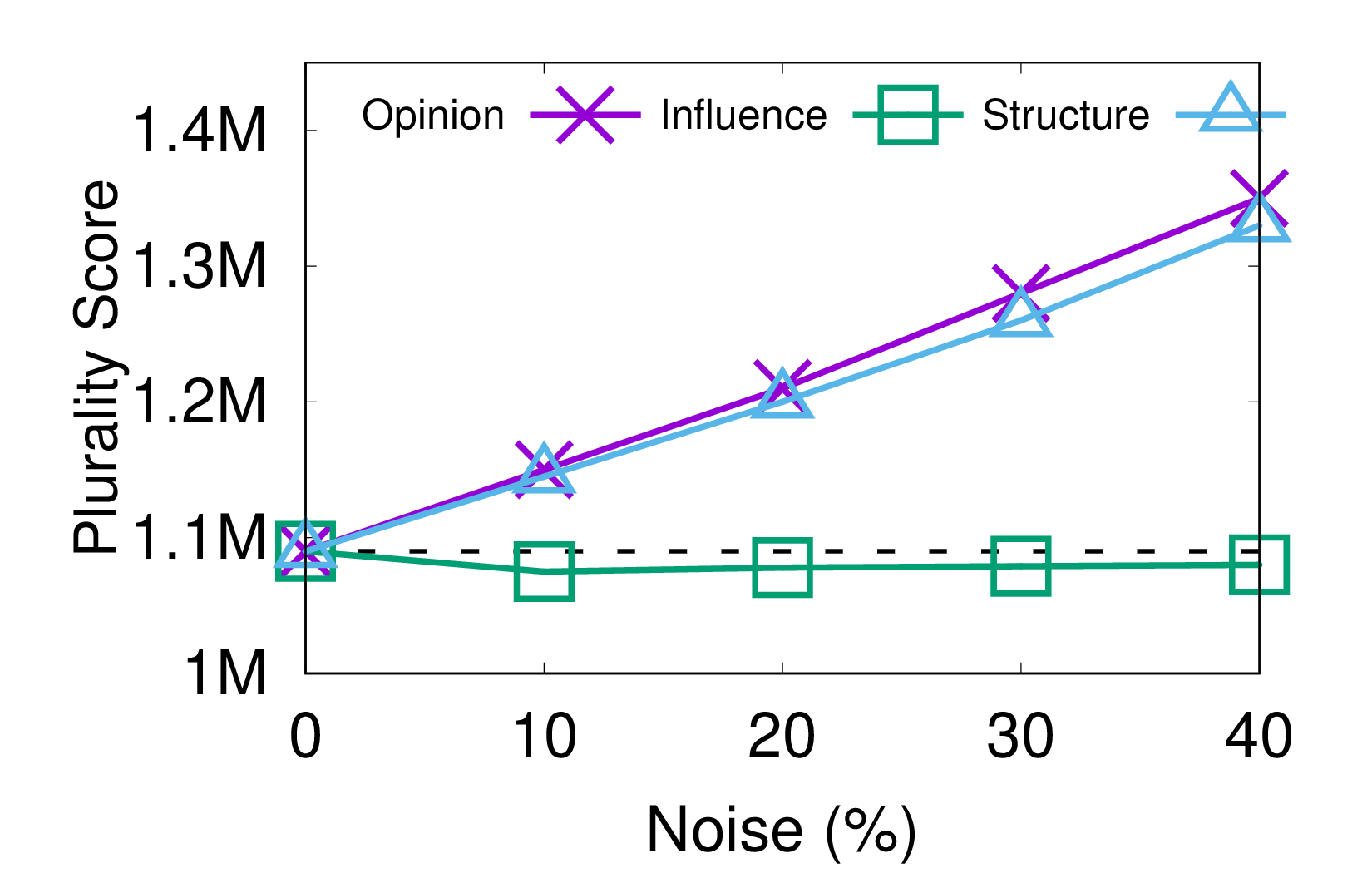}
	\vspace{-2mm}
	\caption{\small {\em \revise{Plurality}} score vs. noise; {\em  Twitter\_Mask}. The horizontal dotted line denotes the \revise{plurality} score for the top-100 greedy seed nodes over the input graph.}
	\vspace{-3mm}
	\label{fig:robust}
\end{figure}

Figure~\ref{fig:robust} demonstrates the performance of the seed set returned by the greedy algorithm on the input dataset, under each noisy setting. In general, the performance of our selected seeds is more sensitive to noise added to initial opinions and to graph structure. Imparting noise to the influence weights results in nearly no change in the \revise{plurality} score, since the noise will be almost eliminated by the normalization. With 40\% noise level, the performance of our seeds is affected by up to 20\% of that on the input dataset. We omit the plots for the other scores, since they follow similar trends.
\subsection{Seed Set Characteristics}
\label{sec:char}
In order to study the effects of network
properties on the seed set returned by our algorithms, 
we generate 500 scale-free synthetic graphs\footnote{\scriptsize{A scale-free network has a power-law degree distribution. 
Our scale-free networks are generated with the Barab\'{a}si-Albert preferential attachment model \cite{barabasi1999emergence}, where a graph of $n$ nodes is grown by attaching each new nodes with $e$ edges that are preferentially attached to existing nodes having high degrees. $n=1000$ and $e=10$ in our experiments.}} using the NetworkX \cite{NX} library. We assume that there are two candidates, and without loss of generality, the first one is our target. The edge weights, initial opinion values, and stubbornness values are assigned uniformly at random in $[0,1]$.

We sort the nodes in each network in descending order of {\bf (1)} {\sf PageRank score}, {\bf (2)} {\sf degree-centrality}, {\bf (3)} {\sf initial opinion} value about the target candidate, {\bf (4)} {\sf difference between initial opinions} (i.e., initial opinion about the target $-$ initial opinion about the non-target), and {\bf (5)} {\sf stubbornness} value for the target candidate. Moreover, we consider the reverse order (i.e., ascending) of {\bf (3)} and {\bf (4)}, as {\bf (6)} and {\bf (7)}, respectively. We retrieve the top-100 nodes $S$ based on each criterion, and measure their qualities according to the following metrics, considering the top-100 seeds $S^*$ returned by our {\sf DM} method for the \revise{plurality} score as the ground truth ranking.

$\bullet$ {\sf Precision} is the fraction of relevant nodes among the retrieved ones; formally: 
\begin{small}
\begin{align}
	Precision=\frac{|S\cap S^*|}{|S|}
\end{align}
\end{small}

$\bullet$ {\sf Normalized Discounted Cumulative Gain (NDCG)} is a measure of ranking quality. First, it computes the {\sf Cumulative Gain (CG)}, which is the sum of the {\em graded relevance values} of all nodes in a ranked list $S$. We assume that the graded relevance value $Rel(i)$ of a  retrieved node $i$ is: $|S^*|-pos(i,S^*)+1$, where $pos(i,S^*)$ is the position of node $i$ in the ground truth ranking $S^*$. (If $i$ does not belong to $S^*$, $pos(i,S^*) = |S^*|+1$.) We have:
\begin{small}
\begin{align}
	CG(S)=\sum_{i=1}^{|S|} Rel(i)
\end{align}
\end{small}
The {\sf Discounted Cumulative Gain (DCG)} is based on the intuition that highly relevant nodes appearing lower in a ranked list $S$ should be penalized by reducing the relevance value logarithmically proportional to the position of the node, and is defined below.
\begin{small}
\begin{align}
	DCG(S)=\sum_{i=1}^{|S|} \frac{Rel(i)}{\log_2(i+1)}
\end{align}
\end{small}
Finally, we define the {\sf Normalized Discounted Cumulative Gain} as
\begin{small}
\begin{align}
	NDCG(S)=\frac{DCG(S)}{DCG(S^*)}
\end{align}
\end{small}

$\bullet$ {\sf Kendall Tau Rank Correlation Coefficient (KTRCC)} measures the ordinal association between two rankings. 
It explores all $\binom{|S|}{2}$ pairs of nodes in a length-$|S|$ ranking list.
After that, it checks whether the order of each pair is the same 
as that in the ground truth ranking (i.e., concordant). We assume that if both nodes in a pair do not belong 
to the ground ranking, there is a tie between them, which will be neglected by KTRCC. 
We determine the number of concordant pairs and discordant pairs, denoted by $CP$ and $DP$, 
respectively. The KTRCC score is computed as:
\begin{small}
\begin{align}
	KT=\frac{CP-DP}{\binom{|S|}{2}}
\end{align}
\end{small}
\begin{table}[t]
	\footnotesize
	\centering
	\begin{center}
		\vspace{-2mm}
		\caption{\small Characteristics of seed sets}
		\vspace{-3mm}
		\begin{tabular}{l||c|c|c}
			\textbf{Ranking Criteria}   & \textbf{Precision} & \textbf{NDCG} & \textbf{KTRCC} \\ \hline \hline
			\textbf{PageRank}   & 0.718 & 0.777     & 0.492 \\ \hline
			\textbf{Degree-Centrality} & {\bf 0.723} & {\bf 0.781} & {\bf 0.498} \\ \hline	
			\textbf{Initial Opinion} & 0.059 & 0.059 & -0.090 \\ \hline
			\textbf{Opinion Difference} & 0.060 & 0.058 & -0.094 \\  \hline	
			\textbf{Stubornness} & 0.010 & 0.099 & -0.053 \\ \hline
			\textbf{Initial Opinion} &\multirow{2}{*}{0.157} & \multirow{2}{*}{0.157} & \multirow{2}{*}{-0.042} \\
			\textbf{(Ascending)} & & &\\ \hline
			\textbf{Opinion Difference} & \multirow{2}{*}{0.150} & \multirow{2}{*}{0.154} & \multirow{2}{*}{-0.047} \\
			\textbf{(Ascending)} & & &\\
		\end{tabular}
		\vspace{-6mm}
		\label{tab:seedset_char}
	\end{center}
\end{table}

Table~\ref{tab:seedset_char} demonstrates the scores of the top-100 retrieved nodes returned by each ranking criterion with respect to our top-100 seeds for \revise{plurality} score. As expected, no single criterion is sufficient alone to generate our seed set. We observe that:
{\bf (1)} The seeds maximizing the \revise{plurality} score tend to have high PageRank score and degree-centrality score. Moreover, the ranking of our selected seeds is more similar to the PageRank and degree-centrality based rankings.
{\bf (2)} The selected seeds have less correlation to the initial opinion values, opinion difference, and stubbornness values. Combined with (1), structural properties of a network are more important for identifying the seeds.
{\bf (3)} Our selected seeds for the \revise{plurality} score have higher preferences to those nodes with smaller initial opinion values towards the target candidate. Intuitively, including such nodes as seeds will increase their opinions drastically, then influence their neighbors as well.
{\bf (4)} Our selected seeds for the \revise{plurality} score also have higher preferences to those nodes having similar initial opinion values towards both candidates. By setting them as seeds, they will instead have clear preferences towards our target candidate.

The performances of these ranking criteria with respect to the seeds maximizing our other voting-based scores 
have similar trends to that with the \revise{plurality} score. Thus, we omit the details here.
}

%% file: 8-conclusions.tex
\vspace{-1mm}
\section{Conclusions}
\vspace{-1mm}
\label{sec:concl}
We formulated and investigated the novel problem of
opinion maximization in a social network, coupled with voting-based
scores. We proved that our problem is \NP-hard and non-submodular under various scores.
To solve the problem, we employed the well-known Sandwich Approximation, under which we proved
that the greedy algorithm can still provide approximation guarantees to our objectives.
Since exact opinion computation via iterative matrix-vector multiplications is inefficient,
we proposed random walk and sketching-based opinion computations, with theoretical approximation guarantees.
Experimental results validated the effectiveness and efficiency of our proposed algorithms.
Considering both accuracy and efficiency results, we recommend the sketching-based approach {\sf RS} as our ultimately proposed method. In future works, it would be interesting to consider more opinion diffusion models and voting scores. 

%% file: appendix.tex
\vspace{-2mm}

\appendices

\section{Comparison with Existing Works on Opinion Maximization}
\label{sec:compare_om}

Our problem setting is similar in some ways to \cite{GionisTT13, AKPT18};
however, there are important differences too, as shown below.

--- \cite{GionisTT13} is the first work on seed selection for opinion maximization in social networks.
For a given user $i$, the diffusion model in \cite{GionisTT13} (Equation 3.2) is given by
\begin{scriptsize}
\begin{align*}
    & \displaystyle z_i = \frac{s_i + \sum_{j \in N(i)} w_{ij} z_j}{1 + \sum_{l \in N(i)} w_{il}} \\
    & \displaystyle = \frac{1}{1 + \sum_{l \in N(i)} w_{il}} \cdot s_i + \left( 1 - \frac{1}{1 + \sum_{l \in N(i)} w_{il}} \right) \sum_{j \in N(i)} \frac{w_{ij}}{\sum_{l \in N(i)} w_{il}} \cdot z_j
\end{align*}
\end{scriptsize}
while, from Equation \ref{eq:fj}, our diffusion model is given by
\begin{align}
    \label{eq:fj_single}
    \small b_{qi}^{(t + 1)} = d_{qi} b_{qi}^{(0)} + \left( 1 - d_{qi} \right) \sum_{j \in V} w_{qji} b_{qj}^{(t)} 
\end{align}
Thus, the diffusion models are similar in following aspects: The intrinsic opinions $s_i$ (resp. expressed opinions $z_i$) in \cite{GionisTT13}
are analogous to our initial opinions $b_{qi}^{(0)}$ (resp. opinions $b_{qi}^{(t)}$ at any time $t$),
and the weights $\frac{w_{ij}}{\sum_{l \in N(i)} w_{il}}$ in \cite{GionisTT13} are the same as the weights $w_{qji}$ in our work. In \cite{GionisTT13},
each node has a preference towards its intrinsic opinion, which is equal to $\frac{1}{1 + \sum_{l \in N(i)} w_{il}}$.
Similarly, in our work, each node has a weight or preference to its initial opinion that is equal
to its stubbornness $d_{qi}$ (which can be any real number in $[0, 1]$, e.g., learnt from real data).
In \cite{GionisTT13}, when a node $i$ is made a seed, its expressed opinion $z_i$ is fixed to $1$. In our work, the same is achieved for the opinion $b_{qi}^{(t)}$ by setting both the initial opinion $b_{qi}^{(0)}$ and the stubbornness $d_{qi}$ to $1$, according to Equation \ref{eq:fj_single}.
However, there is one key difference: The problem in \cite{GionisTT13} involves choosing seeds that maximize the sum of the expressed
opinions at the Nash equilibrium, whereas our problem with the cumulative score involves the sum of the opinions at {\em any} given time horizon.

For the above reason,
the proofs of $\NP$-hardness and submodularity in \cite{GionisTT13} {\em cannot be extended trivially} to
our cumulative score for any finite time horizon. More specifically, \cite{GionisTT13} uses results from the theory of absorbing random walks (those that continue till an absorbing node is reached) to prove that the opinion computed by an absorbing random walk is an unbiased estimate of the true opinion at the Nash equilibrium, a property which is central to the proofs of $\NP$-hardness and submodularity in \cite{GionisTT13}. But in our work, we cannot use absorbing random walks to estimate the opinions at any finite time horizon, which renders the aforementioned proofs invalid (in our case), and hence the extension of the results in \cite{GionisTT13} to ours is non-trivial.


--- In our work, we provide accuracy guarantees for all our three methods (direct matrix-vector multiplication, random walks and sketches). However, \cite{GionisTT13} only provides a $(1 - 1/e)$-approximation guarantee for
the \textsf{Greedy} method (via direct matrix-vector multiplication) which is inefficient, and thus it proposes other heuristic methods without any accuracy guarantee.

--- For a given user $i$, the diffusion model in Equation 1 of \cite{AKPT18} is given by
\begin{equation*}
    \small x_i(t + 1) = \alpha_i s_i + \left( 1 - \alpha_i \right) \sum_{j \in N(i)} \frac{1}{deg(i)} \cdot x_j(t)
\end{equation*}
Thus, it is similar to ours (Equation \ref{eq:fj_single}) in the following aspects:
The innate opinions $s_i$ (resp. expressed opinions $x_i(t)$) in \cite{AKPT18} are analogous to our initial opinions $b_{qi}^{(0)}$ (resp. opinions $b_{qi}^{(t)}$),
and the weights $\frac{1}{deg(i)}$ in \cite{AKPT18} are similar to our weights $w_{qji}$.
In both works, each node has a weight or preference to its initial opinion equal to its stubbornness $d_{qi}$ or resistance $\alpha_i$
(which can be any real number in $[0, 1]$). However, the problem in \cite{AKPT18} requires maximizing
the sum of the expressed opinions at equilibrium, whereas our problem with the cumulative score
involves the sum of the opinions but at any given finite time horizon. In addition, the changes made
when a user is chosen to be a seed are different in our work from \cite{AKPT18}. We set both the initial opinion and stubbornness values to $1$,
whereas \cite{AKPT18} sets only the resistance value within a given interval $[l, u]$. Thus, the objective function in \cite{AKPT18} (under the budgeted setting) is
neither submodular nor supermodular, which is why \cite{AKPT18} does not provide any accuracy guarantee on even the greedy method for budgeted opinion maximization.
On the other hand, since our cumulative score is submodular, the greedy method provides a $(1 - 1/e)$-approximation guarantee.

--- In addition to the cumulative score, our work also involves the \revise{plurality} and Condorcet winner scores which are not used in prior works on opinion maximization \cite{GionisTT13, AKPT18}, and hence constitute one of our novel contributions (as rightly pointed out by the reviewer). Moreover, we design non-negative, non-decreasing, submodular upper and lower bound functions for the \revise{plurality} and Condorcet winner scores, and apply the {\em Sandwich Approximation} technique to achieve empirically good approximation guarantees (refer to \S~\ref{sec:practical_effectiveness} for more details).

--- The random walk interpretations in \cite{GionisTT13, AKPT18} are similar to that of ours, in as much as
the fact that the expressed opinion of a node $v$ is the expected innate opinion of the end node
of a random walk starting from $v$. However, \cite{GionisTT13, AKPT18} only deal with opinions at their respective equilibria,
which require that their random walks continue till absorption. In contrast, our method involves
opinions at any finite time horizon, which means that our random walks go on till absorption (by a fully stubborn node),
or the number of steps in the walk being equal to the time horizon, whichever happens earlier.
Also, \cite{GionisTT13, AKPT18} enable random walks by augmenting the graph with a set of $n$ new nodes and edges, in order to prove the unbiasedness of their random walk estimates by leveraging results from the theory of absorbing random walks.
However, there is no such
augmentation in our random walks, which requires us to prove the unbiasedness of our estimated opinions in a different way. Furthermore, \cite{GionisTT13, AKPT18} do not use random walks in their algorithms
and do not mention the number of walks needed to ensure some accuracy guarantee on the estimated opinions,
which is one of the core contributions in our work.

Additionally, all three works (ours, \cite{GionisTT13, AKPT18}) make the assumption that setting the opinion and/or stubbornness values
of the seed users is ``easy'' and under our control. Recall that the Greedy algorithm in \cite{GionisTT13} for opinion maximization can be adapted
for a finite time horizon and to consider input stubbornness values as ours, which we denoted as GED-T. Our cumulative score, due to its
aggregate nature, is similar to opinion maximization in the single
campaign setting; therefore, GED-T and our DM approach (direct matrix-vector multiplication) perform
the same for the cumulative score (only). However, our RS method (reverse
sketching) is about two orders of magnitude faster than GED-T, even for the cumulative score (\S \ref{sec:exp}).

\section{Usefulness of a Finite Time Horizon}
\label{sec:time_horizon}

In many real-world applications, taking into account a finite time horizon is important.
Consider a paid movie service with a limited period discount (mentioned in the second paragraph of \S~1),
or an upcoming election. In such cases, the optimal seed set considering opinions at convergence
can be drastically different from the corresponding set with a finite time horizon.

There are examples in the past literature where certain properties that are true when the time horizon is infinite
cease to hold (or cannot be proved by a direct extension) for any given finite time horizon.
As a concrete example, in \cite{CautisMT19},
the expected spread is adaptive submodular under full feedback (analogous to an infinite time horizon)
but not under partial feedback (analogous to any given finite time horizon).
Similarly, despite \cite{GionisTT13} establishing that the sum of all opinions at the Nash equilibrium is submodular w.r.t. the seed set,
the submodularity of our cumulative score under similar problem settings but at any given finite time horizon
is non-trivial and cannot be proved by a direct adaptation of the proof in \cite{GionisTT13}.

\begin{figure}
  \vspace{-2mm}
  \centering
  \includegraphics[scale=0.18,angle=270]{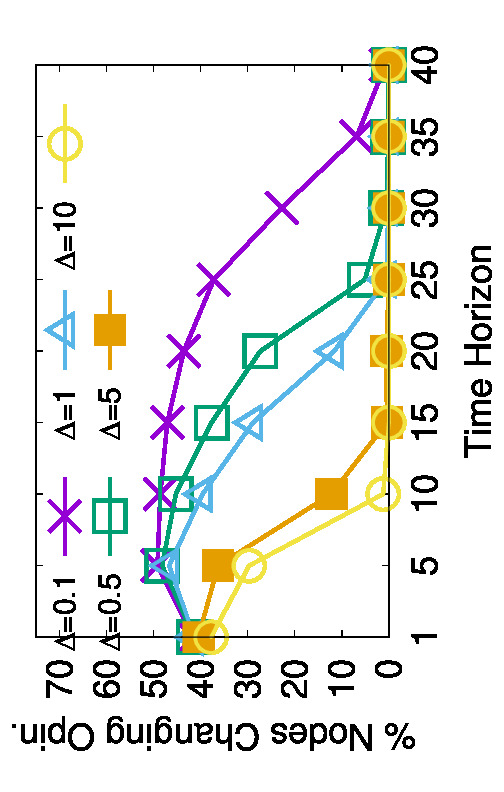}
  \vspace{-2mm}
  \caption{\small Variation of the percentage of nodes changing opinion from timestamp $t-1$ to $t$ as a function of $t$ for the {\sf Yelp} dataset. $\Delta$ denotes the tolerance, i.e., the maximum percentage opinion change (from $t-1$ to $t$) allowed for being considered as ``no change in opinion''.}
  \vspace{-4mm}
  \label{fig:conv}
\end{figure}

In addition, we empirically study the variation of opinions with respect to the time horizon. Specifically, at every time horizon $t$, we compute the fraction of nodes $v$ for which $\left| b_{qv}^{(t)} - b_{qv}^{(t - 1)} \right| > \frac{\Delta}{100} \times b_{qv}^{(t - 1)}$, where $\Delta$ is a tolerance parameter that decides how much change in opinion is considered negligible. Figure \ref{fig:conv} shows the variation of the above percentage with the time horizon for different values of the tolerance $\Delta$. We observe that there is a significant fraction of users changing their opinions before time horizon 30, especially when the tolerance $\Delta$ is small. We further compare the optimal seed sets ($k=100$, cumulative score) returned by our greedy algorithm for different time horizons, and find that they can differ from each other significantly. For example, the optimal seed sets at $t=5$, $10$, and $20$ have only 42\%, 48\%, and 61\% overlap with that at $t=30$, respectively. Finally, our experimental results in Figure \ref{fig:varyT} show that the cumulative score varies with respect to the time horizon. These demonstrate the importance of considering the time horizon in practice.

\eat{As mentioned in \S~2.1, the FJ model can reach convergence (after which none of the users' opinions changes with time) if and only if the edge weight matrix of the subgraph induced by all oblivious nodes is regular or there is no oblivious node, where oblivious nodes are those that have zero stubbornness and are not reachable from nodes with non-zero stubbornness. We tried to check whether this condition is true for our datasets. In the {\sf Twitter} datasets, there is no oblivious node; in fact, there is no node with zero stubbornness (those values were chosen uniformly at random). However, in the {\sf Yelp} dataset, the number of oblivious nodes is so large that the weight matrix of the induced subgraph does not fit in the main memory, and hence we cannot verify the above property. Thus, we check empirically whether the opinions of all nodes converge or not.
Specifically, at every time-stamp $t$, we compute the fraction of nodes $v$ for which $\left| b_{qv}^{(t)} - b_{qv}^{(t - 1)} \right| > \frac{\Delta}{100} \times b_{qv}^{(t - 1)}$, where $\Delta$ is a tolerance parameter that decides how much change in opinion is considered negligible. Fig. \ref{fig:conv}
shows the variation of the above percentage with the time horizon for different values of the tolerance $\Delta$. The larger the tolerance $\Delta$ is,
the smaller the time-stamp at which convergence is reached. For $\Delta = 0.1$, convergence is reached at time-stamp $40$. Beyond time-stamp $20$,
even for $\Delta = 1$, the fraction of nodes incurring changes in opinions is only about $4\%$. This explains our observation in Figure \ref{fig:varyT} of our paper,
where the cumulative score remains nearly the same beyond time-stamp $20$.}

\section{Differences between Random Walks in PageRank and in the FJ Model}
\label{sec:rw_pr}

Let us start with a brief overview of PageRank.
The PageRank vector \cite{ALNO07} $\pi \in \mathbbm{R}^n$ is the solution (at convergence) to: 
\begin{equation*}
	\small \pi^{(t + 1)} = c \pi^{(t)} P + (1 - c) \pi^{(0)}
\end{equation*}
where $P$ is the edge weight matrix (which is row-stochastic), $\pi^{(0)}$ is a column
vector with each element equal to $n^{-1}$, and $c \in (0, 1)$ is a damping factor.
As mentioned in \cite{ALNO07}, the elements of the PageRank vector can be estimated by a random walk method,
which is repeated $m$ times starting from each node: At each step, if the node has no outgoing edges, the walk stops;
otherwise, the random walk terminates with probability $1 - c$, and makes a transition according to the matrix $P$
with probability $c$. Once all walks are generated, the estimate $\widehat{\pi}_v$ of $\pi_v$ for a node $v \in V$
is the total number of visits to $v$ divided by the total number of visited nodes.

Based on this, there exist similarities between the equations of the FJ model
(Equation \ref{eq:fj}) and PageRank, as also in their corresponding random walk methods.
However, there are important differences between the two random walk methods.

--- In PageRank, a random walk goes on until the decision is taken to stop at a node (with probability $c$), or if a node with no outgoing edge is reached. But in the FJ model, in addition to stopping at a node with probability equal to its stubbornness, the random walk continues only till a specified time horizon.

--- In PageRank, the estimated value for a node $v$ is the fraction of the total number of visits to $v$ (in all random walks, not necessarily those starting from $v$). But in the FJ model, the estimated opinion value for a node $v$ is the average of the initial opinion values of the end nodes of those walks which start from $v$.

--- In PageRank, combining Theorem 1 and Equations 16-17 in \cite{ALNO07}, if the number of walks $m$ starting from each node $v$ satisfies
\begin{equation*}
	\small m \geq \frac{1 + q_{vv}}{1 - q_{vv}} \left( \frac{x_{1 - \frac{\alpha}{2}}}{\epsilon} \right)^2
\end{equation*}
where $q_{vv}$ is the probability of the walk returning to $v$ if it starts from $v$ and $x_{1 - \frac{\alpha}{2}}$ is a $\left( 1 - \frac{\alpha}{2} \right)$-quantile of the standard normal distribution, then
\begin{equation*}
	\small \Pr \left( \left| \widehat{\pi}_v - \pi_v \right| \leq \epsilon' \pi_v \right) \geq (1 - \alpha) (1 - \beta)
\end{equation*}
for any $\beta > 0$ and $\epsilon'$ satisfying
\begin{equation*}
	\small |\epsilon - \epsilon'| < \frac{x_{1 - \frac{\beta}{2}} (1 + \epsilon)}{\sqrt{nm}} \cdot \frac{c}{1 - c} \sqrt{\left( 1 - \frac{n_0}{n} \right) \left( 1 + c^3 \right)}
\end{equation*}
where $n_0$ is the number of nodes without any outgoing edge. In contrast, in the FJ model, the number of walks $\lambda_v$ from each node $v$ should satisfy different conditions to ensure different accuracy guarantees for various scores, as shown in Theorems \ref{th:cumu_rws} - \ref{th:bound_con_rw}.

\begin{figure}[tb!]
	\centering
	\includegraphics[scale=0.14,angle=270]{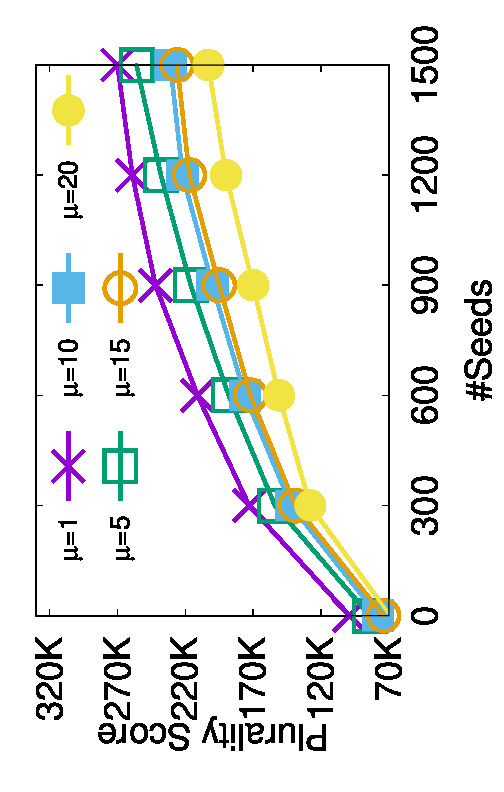}
	\includegraphics[scale=0.14,angle=270]{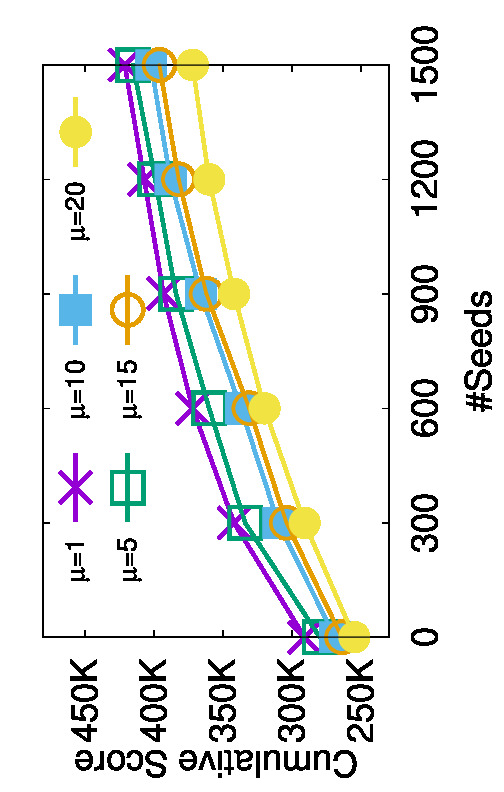}
	\vspace{-2mm}
	\caption{\small Opinion scores vs. $\mu$ for candidates ``Chinese'' on {\sf Yelp} (left) and ``Democratic'' on {\sf Twitter\_US\_Election} (right)}
	\label{fig:varyMU_score}
	\vspace{-4mm}
\end{figure}

In addition, there also exist differences in \textit{how the random walks for PageRank and the FJ model are
used for finding the corresponding top-$k$ nodes.} In PageRank, we need to generate $m$
random walks from each node {\em only once}, and then return the $k$ nodes with the largest PageRank estimates.
In contrast, for maximizing a voting-based score with the FJ model, we need to run a greedy algorithm,
where in each of the $k$ iterations, we find the node that maximizes the \textit{marginal gain} in the score, and include it
as a seed. For that, originally {\em in every iteration of the greedy algorithm}, we have to
generate the random walks for each of $\bigO(n)$ candidate seed nodes, in order to compute
its marginal gain estimate. Next, we optimize the process by
generating all walks right in the beginning and then reuse them in each of the $k$ iterations (\S~\ref{sec:random_walk}).
As a further optimization, we propose the sketch-based method (\S~\ref{sec:sketching}) which computes random walks starting
from only $\theta \ll n$ nodes, thereby making the process even more efficient, with quality guarantees on the $k$ seed nodes returned (Theorem 13),
which is our ultimate objective. To the best of our knowledge, the proofs of the above guarantees (Theorems \ref{th:cumu_rws} - \ref{th:theta}) are novel and do not follow from any prior
results on PageRank.

\section{Empirical Justification of the Default Value of $\mu$}
\label{sec:mu}

\eat{ 
To explain our choice of the default value of $\mu$ for our datasets, we conduct the following empirical studies. 
} First, we note that 
the edge weight distribution is based on $1 - e^{-a/\mu}$ (i.e., before normalization). 
\eat{ with varying $\mu$, is shown in Fig.~\ref{fig:varyMU_dist} (in this response letter, see previous page)
for two candidates from two datasets. When $\mu$ is too small (e.g., 1)
or too large (e.g., 20), we observe that over 95\% of the edges can have extremely large ($>0.5$) or small ($\leq 0.1$) weights (before
normalization). On the other hand, the weight distribution shows more diversity when $\mu=10$ and $\mu=15$.}  
In Fig. \ref{fig:varyMU_score}, we show the
Cumulative score (on {\em Twitter\_US\_Election}) and the \revise{plurality} score
(on {\em Yelp}) when varying the number of seeds, for different values of $\mu$.
In general, the difference is small: we observed that the normalization step (which ensures the row-stochastic property)
reduces the impact of $\mu$ on the edge probabilities. The curves corresponding to $\mu=10$ and $\mu=15$ lie in the middle
and nearly overlap. Thus, we choose $\mu=10$ as our default setting.


\section{Concentration Inequalities Used}
\label{sec:ineq}

\begin{theor}
[\cite{chung2006concentration}]
\label{th:conc}
Let $X_1, \ldots, X_\theta$ be non-negative independent random variables satisfying $X_i - \mathbbm{E}[X_i] \leq M \; \forall i \in [1, \theta]$. Let $X = \sum_{i=1}^\theta X_i$. For any $\beta > 0$,
\begin{small}
\begin{align*}
    \Pr \left( X - \mathbbm{E}[X] \geq \beta \right) &\leq \exp \left( - \frac{\beta^2}{2 \left( Var[X] + \frac{M\beta}{3} \right)} \right) \\
    \Pr \left( X - \mathbbm{E}[X] \leq -\beta \right) &\leq \exp \left( - \frac{\beta^2}{2 \sum_{i=1}^\theta \mathbbm{E} \left[ X_i^2 \right]} \right)
\end{align*}
\end{small}
\end{theor}
\begin{theor}
[\cite{TXS14}]
\label{th:chernoff}
Let $X_1, \ldots, X_\theta$ be i.i.d. random variables such that $X_i \in [0, 1]$ and $\mathbbm{E}[X_i] = \mu \; \forall i \in [1, \theta]$. For any $\epsilon > 0$,
\begin{equation*}
    \small \Pr \left( \left| X - \theta\mu \geq \epsilon \cdot \theta\mu \right| \right) \leq \exp \left( \frac{\epsilon^2}{2 + \epsilon} \cdot \theta\mu \right)
\end{equation*}
\end{theor}
\begin{theor}
[\cite{hoeffding1963probability}]
\label{th:hoeffding}
Let $X_1, \ldots, X_\theta$ be independent random variables such that $X_i \in [0, 1] \; \forall i \in [1, \theta]$. Let $\overline{X} = \frac{1}{\theta} \sum_{i=1}^\theta X_i$ and $\mu = \mathbbm{E}\left[\overline{X}\right]$. Then, for $0 \leq \epsilon < 1 - \mu$,
\begin{equation*}
    \small \Pr \left( \overline{X} - \mu \geq \epsilon \right) \leq \left[ \left( \frac{\mu}{\mu + \epsilon} \right)^{\mu + \epsilon} \left( \frac{1 - \mu}{1 - \mu - \epsilon} \right)^{1 - \mu - \epsilon} \right]^n
\end{equation*}
\end{theor}